\documentclass{PapER}

\newcommand{\er}[1]{{\color{black}#1}}
\newcommand{\vn}[1]{\text{VN}\,{#1}}
\newcommand{\cn}[1]{\text{CN}\,{#1}}
\begin{document}
\title{Spatially Coupled LDPC Codes with Sub-Block Locality}
\author{
\vspace*{0.8cm} Eshed Ram \qquad Yuval Cassuto\\
Andrew and Erna Viterbi Department of Electrical Engineering \\
Technion -- Israel Institute of Technology, Haifa 32000, Israel\\
E-mails: \{s6eshedr@campus, ycassuto@ee\}.technion.ac.il
}

\maketitle

\begin{abstract}
A new type of \er{spatially coupled low-density parity-check (SC-LDPC)} codes motivated by practical storage applications is presented. SC-LDPCL codes (suffix 'L' stands for locality) can be decoded locally at the level of sub-blocks that are much smaller than the full code block, thus offering flexible access to the coded information alongside the strong reliability of the global full-block decoding. Toward that, we propose constructions of SC-LDPCL codes that allow controlling the trade-off between local and global correction performance. In addition to local and global decoding, the paper develops a density-evolution analysis for a decoding mode we call semi-global decoding, in which the decoder has access to the requested sub-block plus a prescribed number of sub-blocks around it. SC-LDPCL codes are also studied under a channel model with variability across sub-blocks, for which decoding-performance lower bounds are derived.\footnote{Part of the results of this paper was presented at the 2018 International Symposium on Turbo Coding and the 2019 International Symposium on Information Theory.}
\end{abstract}


\section{Introduction}
\label{Sec:Intro}
Spatial coupling (SC) of low-density parity-check (LDPC) codes is an extremely useful technique to construct block codes with superior correction capability and efficient decoders. These properties make spatially coupled LDPC (SC-LDPC) codes attractive for implementation and deployment in real systems. In this paper\er{,} we endow SC-LDPC codes with an additional desired property: the ability to access and decode sub-blocks much smaller than the full code block. This property is especially needed in memory and storage systems that require flexible access (a.k.a. random access) to small data units alongside high data reliability.

SC-LDPC codes were extensively studied recently and were shown to have many desired properties. For example, in \cite{KudekarRichUrb13} it was proven that SC-LDPC codes achieve capacity universally on memoryless binary symmetric channels under belief propagation (BP) decoding due to a phenomenon called \emph{threshold saturation}; in \cite{Mitchell15} it was exemplified that the minimum distance of protograph-based SC-LDPC codes grows linearly with the block length without compromising in thresholds; \cite{Mitchell11} showed that typical protograph-based SC-LDPC codes present linear-growth of \er{the size of} minimal trapping sets. These properties imply good \er{bit-error rate (BER)} performance in the \emph{waterfall} and \emph{error floor} regions, for the BP decoder. Moreover, the special structure of SC-LDPC codes, where bits participating in a particular parity-check equation are spatially close to each other, renders a locality property that can be exploited to implement low-latency high-throughput belief-propagation based decoders; such decoders are pipelined decoders \cite{FelstromZigangirov99, Pusane08} and window decoders \cite{IyenPapa12,IyenSiegel13,Lentmaier11}. 

When used in data-storage applications, where decoding failures imply data-losses, an error-correcting code must protect against extremely high noise levels (although most noise instances are much milder), requiring very large block lengths and complex decoding, thus degrading the latency and throughput of the device. A possible solution to this problem \er{is} \emph{sub-block-access} codes \cite{CassHemo17,RamCassuto18a} that enable decoding small sub-blocks (i.e., \emph{local decoding}) for fast read access, while providing a high data-reliability ``safety net" decoding over the large code block (i.e., \emph{global decoding}). Formally, in a sub-block-access code, a code block of length $N$ is divided into $M$ sub-blocks of length $n$ each. Each sub-block is a codeword of one code, and the concatenation of the $M$ sub-blocks forms a codeword of another (stronger) code.
In this paper, we construct SC-LDPC codes with this sub-block structure that offer sub-block decoding capabilities; we call these codes SC-LDPCL codes (suffix 'L' stands for locality). The key to achieving this is designing spatially coupled protographs that have suitable correction performance under a variety of access-locality modes. 

\subsection{Contributions}
\label{Sub:contributions}

Our main scope in this paper \er{is SC-LDPCL codes constructed by coupling regular protographs}.  
The analysis focuses on the {\em binary erasure channel} (BEC), but can be readily extended to other channels, e.g. via \er{extrinsic-information transfer (EXIT) \cite{TenBrink04}} functions. Moreover, the constructed codes are simulated over the BEC and over the {\em additive white Gaussian noise} (AWGN) channel, performing as predicted by the theoretical analysis. 
The paper is organized as follows. In Section~\ref{Sec:SC-LDPCL Codes}\er{,} we derive bounds on decoding thresholds of protographs and show that existing protograph-based SC-LDPC codes do not enable sub-block access, in the sense that the decoding threshold of such access is zero. These results help characterize the design measures needed for non-trivial sub-block decoding performance, which lead to a construction offering a tradeoff between local and global decoding performance. In Section~\ref{Sec:SG}\er{,} we suggest a new BP-based decoding strategy we call \emph{semi-global} (SG) decoding, in which in addition to the requested sub-block, the decoder has access to some prescribed number of sub-blocks around it. This section derives (and simplifies) a density-evolution analysis of semi-global decoding for the construction in Section~\ref{Sec:SC-LDPCL Codes}.
In Section~\ref{Sub:SG analysis}\er{,} we examine the performance of SG decoding and show that it exhibits a significant complexity reduction compared to global decoding, while costing only a small fraction in the threshold. We then consider a practically motivated data-storage model in which variability is introduced to the channel quality \er{(as motivated by recent empirical studies \cite{TaraUchi16,ShaAl20})}. 
Using lower bounds we derive on decoding success probabilities, we show that SG decoding is highly motivated by this model.
Finally, in Section~\ref{Sec:Gen Construct} we generalize our SC-LDPCL construction (which in Section~\ref{Sec:SC-LDPCL Codes} is restricted to memory $1$), and suggest a richer family of SC-LDPCL codes, including codes with two-dimensional coupling. We then discuss SG decoding over these codes.

\subsection{Related Work}
\label{Sub:relate}

SC-LDPC codes date back to 1999 \cite{FelstromZigangirov99}, and have been studied extensively in the past decade. Many protograph-based \cite{Thorpe03} constructions of SC codes were suggested (see \cite{Mitchell15} and references therein) including reshaping and enhancing SC codes for improved asymptotic and finite-length performance \cite{TruMitch19}. 
\er{
	More recent contributions propose multi-dimensional spatially coupled LDPC codes for global performance improvements \cite{TruMitch19,LiuLi15,OlmMitch17,DolEsf20}, and for special channel models \cite{ScmalMahd14,Ohashi13}. As far as we know, none of this previous work constructed codes that enable sub-block access. Furthermore, decoding of our codes is operationally different since in the local and semi-global modes we seek decoding only a single target sub-block. 
	
	The semi-global decoding mode we propose and study in this paper resembles sliding-window decoders \cite{IyenSiegel13,Lentmaier11} that were suggested for channels with memory (e.g. bursty and Gilbert-Elliott channels) and parallel channels (see \cite{IyenPapa12} and \cite{ScmalMahd14}, respectively). 
	Our work differs from these prior works since the semi-global decoder accesses the codeword differently from the window decoder, for the purpose of reaping latency and complexity benefits. The semi-global access mode also motivates analysis over channels with sub-block variability	\cite{McEliece84}, which are not addressed by prior work.
}

A large body of work has been devoted to codes that possess certain locality properties, including locally recoverable codes \cite{BlauHetz16, TamoBarg14} and regenerating codes \cite{Dimakis10}; the former codes target the problem of reducing the number of nodes needed to recover a failed node, and the latter are designed to reduce the repair bandwidth. Both of these types assume an error model in which every node (sub-block in our context) is either fully known or fully erased. However, in many applications a finer error model is assumed, i.e., a few errors in each sub-block. We consider this model and suggest sub-block access with a certain level of data protection, combined with increased data-reliability access with full-block access. 
Earlier work on sub-block-access codes includes multi-block Reed-Solomon codes in \cite{CassHemo17} and sub-block-access LDPC codes in \cite{RamCassuto18a}. The former suggests algebraic constructions and properties, and the later deals with ordinary (i.e., not spatially coupled) LDPC codes. As we will see later, using SC-LDPC codes as our underlying code renders new design trade-offs and decoding strategies that are motivated by practical storage applications.

\section{Preliminaries}
\label{Sec:Pre}

\subsection{Protograph Based LDPC Codes}
\label{Sub:proto intro}
An LDPC protograph is a (small) bipartite graph $\mathcal{G}=\left(\mathcal{V}\cup\mathcal{C},\mathcal{E}\right)$, where $\mathcal{V}=\left \{v_1,\ldots,v_{|\mathcal{V}|}\right \},\mathcal{C}=\left \{c_1,\ldots,c_{|\mathcal{C}|}\right \},$ and $\mathcal{E}$ are the sets of variable nodes (VNs), check nodes (CNs), and edges, respectively. For every VN $v\in \mathcal{V}$, we denote by $d_v$ its edge degree. Similarly, we write $d_c$ for the edge degree of a CN $c \in \mathcal{C}$. 
A Tanner graph is generated from a protograph $\mathcal{G}$ by a lifting ("copy-and-permute") operation specified by a lifting parameter $L$ (for more details see \cite{Thorpe03} and \cite{Mitchell15}). The design rate of the derived LDPC code is independent of $L$ and given by $R_\mathcal{G}=1-\big| \mathcal{C} \big|/\big| \mathcal{V} \big| $.
If we let $L \to \infty$, then we can analyze the performance of the BP decoder on the resulting ensemble of Tanner graphs via density evolution on the original protograph. Formally, for the BEC we have:
\begin{fact} \label{Fact:DE proto}
Let $\mathcal{G}=\left(\mathcal{V}\cup\mathcal{C},\mathcal{E}\right)$ be an LDPC protograph, let $v\in \mathcal{V}$ be a variable node of degree $d_v$, and let $c \in \mathcal{C}$ be a check node of degree $d_c$. Let $\{e^v_{1},e^v_{2},\ldots,e^v_{d_v}\}$ be the set of all edges connected to $v$, and let $\{e^c_{1},e^c_{2},\ldots,e^c_{d_c}\}$ be the set of all edges connected to $c$. Consider a transmission over the $BEC(\epsilon)$, of a codeword from a binary linear code that corresponds to a random Tanner graph lifted from $\mathcal{G}$ with lifting parameter $L$, denoted by $\mathcal{G}^L$. 
\er{
 	For every $i\in \{1,2,\ldots,d_v\}$, let $x_\ell \left(e^v_{i}\right)$ and $u_\ell \left(e^v_{i}\right)$ be the fraction of $e^v_{i}$-type edges in lifted graph $ \mathcal{G}^L $ that carry VN-to-CN and CN-to-VN erasure messages, respectively, after $\ell$ BP iterations. Similarly, for every $j\in \{1,2,\ldots,d_c\}$, let $x_\ell \big(e^c_{j}\big)$ and $u_\ell \big(e^c_{j}\big)$ be the fraction of $e^c_{j}$-type edges in $ \mathcal{G}^L $ that carry VN-to-CN and CN-to-VN erasure messages after $\ell$ BP iterations.
	Then, as $L \to \infty$
	\begin{subequations}
	\begin{align}
	\label{Eq:DE Vars}
	x_\ell \left(e^v_{i},\epsilon\right) &= \epsilon \cdot \prod_{\substack{1\leq i'\leq d_v \\ i'\neq i}} u_{\ell} \left(e^v_{i'}\right), \\
	\label{Eq:DE Checks}
	u_\ell \left(e^c_{j}\right) &= 1- \prod_{\substack{1\leq j'\leq d_c \\ j'\neq j}}\left( 1- x_{\ell-1} \left(e^c_{j'}\right)\right),\\
	\label{Eq:DE init}
	x_{-1} \left(e^v_{i}\right)&= u_{-1} \left(e^c_{j}\right)=1.
	\end{align}
	\end{subequations}
	Moreover, as $L \to \infty$ the probability that $v$ is erased after $\ell$ BP iterations is given by 
	\begin{align} \label{Eq:P_e,l}
	P_{\ell}(v,\epsilon)=\epsilon \prod_{1\leq i\leq d_v} u_{\ell} \left(e^v_{i}\right).
	\end{align}
}
\end{fact}

The BP decoding threshold of an LDPC protograph $\mathcal{G}$ is defined by
\begin{align}\label{Eq:threshold}\epsilon_{\text{BP}}^*\left(\mathcal{G}\right) = \sup\{\epsilon\in[0,1]\colon \lim_{\ell\to \infty}P_{\ell}(v,\epsilon)=0,\quad  \forall  v \in \mathcal{V}\}.
\end{align}
For simplicity of notations, in the rest of the paper, we remove the subscript BP from the threshold notation. 

A protograph $\mathcal{G}=\left(\mathcal{V}\cup\mathcal{C},\mathcal{E}\right)$ is frequently represented through a bi-adjacency matrix $H_\mathcal{G}$, where the VNs in $\mathcal{V}$ are indexed by the columns of $H_\mathcal{G}$, the CNs in $\mathcal{C}$ by the rows, and an element in $H_{\mathcal G}$ represents the number of edges connecting the corresponding VN and CN. In this matrix representation, we write $\epsilon^*\left(H_\mathcal{G}\right)$ to denote the \er{(BP)} decoding threshold defined in \eqref{Eq:threshold}. If the protograph is $ (l,r) $-regular (every VN and CN are of degree $ l $ and $ r $, respectively), then we write $ \epsilon^*(l,r) $ to denote its threshold.

\subsection{SC-LDPC Codes}
\label{Sub:SCLDPC intro}
An $ (l,r) $-regular SC-LDPC protograph is constructed by coupling together a number of $ (l,r) $-regular protographs and truncating the resulting chain. This coupling operation introduces a convolutional structure to the code, which can be visualized through the matrix representation of the protograph. Let $ B =1^{l\times r}$ be an all-ones base matrix representing an $(l,r)$-regular LDPC protograph, and let $ \{B_\tau\}_{\tau=0}^T $ be binary matrices such that $ B=\sum_{\tau=0}^T B_\tau$ (in this paper, we consider only binary $ B $ matrices). Coupling a limitless number of copies of $ B $ amounts to diagonally placing copies of $\begin{pmatrix} B_0 ; B_1 ; \cdots ; B_T \end{pmatrix}$ \er{('$;$' represents vertical concatenation)} as in Figure~\ref{Fig:SC36}(b). 
By truncating the infinite matrix in Figure~\ref{Fig:SC36}(b) at some width, and removing all-zero rows, a spatially coupled LDPC protograph is constructed. This truncation results in a small number of \emph{terminating} CNs (of low degree), which effectuates a decrease in design rate and an increase in the decoding threshold, compared to the code ensemble corresponding to the base matrix $ B $.
However, as the length of the coupled chain increases, the design rate of the coupled protograph converges to the design rate of the underlying code ensemble, while its BP threshold exhibits a phenomenon known as \emph{threshold saturation} \cite{KudekarRichUrb13}, whereby it converges to the \emph{maximum a-posteriori} (MAP) threshold of the underlying code ensemble. 

Throughout most of this paper, we consider $ (l,r)$-regular SC-LDPC protographs with memory $ T=1$ (Sections~\ref{Sec:SC-LDPCL Codes}--\ref{Sec:SG}). The results are then extended to higher-memory codes in Section~\ref{Sec:Gen Construct}.

\begin{example}\label{Ex:Background}
Figure~\ref{Fig:SC36}(a) illustrates a spatially coupled $(3,6)$ protograph with 18 VNs. The protograph is generated by  
$B_0=\left(
1\;1\;0\;0\;0\;0\; ; 
1\;1\;1\;1\;0\;0\; ; 
1\;1\;1\;1\;1\;1
\right),$ and $ B_1=1^{3\times 6}-B_0$. The design rate of the coupled protograph is $R=0.389$, and the BP threshold is $0.512$. Figure~\ref{Fig:SC36} will serve as a basis for a running example in the paper.
\end{example}

\begin{figure}
\begin{center}
\begin{tikzpicture}\label{Tikz:36 SCLDPC}
\tikzstyle{cnode}=[rectangle,draw,fill=gray!70!white,minimum size=4mm]
\tikzstyle{vnode}=[circle,draw,fill=gray!70!white,minimum size=4mm]
\pgfmathsetmacro{\x}{3}
\pgfmathsetmacro{\y}{1}

\foreach \m in {1,2,3}
{
	
	\foreach \c in {1,2,3}
	{
		\node[cnode] (c\c\m) at (\m*\x-\x,\c*\y-\y) {};	
	}
	\foreach \v in {1,2,3,4,5,6}
	{
		\node[vnode] (v\v\m) at (\m*\x-0.5*\x,\v*\y-2.5*\y) {};	
	}	
}
\node[cnode] (c34) at (4*\x-\x,3*\y-\y) {};	
\node[cnode] (c24) at (4*\x-\x,2*\y-\y) {};	
\foreach \m in {1,2,3}
{	
	\pgfmathtruncatemacro{\k}{\m + 1}
	\draw[thick] (v1\m)--(c1\m) ;
	\draw[thick] (v1\m)--(c2\k) ;
	\draw[thick] (v1\m)--(c3\k) ;	
	\draw[thick] (v2\m)--(c1\m) ;
	\draw[thick] (v2\m)--(c2\k) ;
	\draw[thick] (v2\m)--(c3\k) ;
	
	\draw[thick] (v3\m)--(c1\m) ;
	\draw[thick] (v3\m)--(c2\m) ;
	\draw[thick] (v3\m)--(c3\k) ;
	\draw[thick] (v4\m)--(c1\m) ;
	\draw[thick] (v4\m)--(c2\m) ;
	\draw[thick] (v4\m)--(c3\k) ;
	
	\draw[thick] (v5\m)--(c1\m) ;
	\draw[thick] (v5\m)--(c2\m) ;
	\draw[thick] (v5\m)--(c3\m) ;
	\draw[thick] (v6\m)--(c1\m) ;
	\draw[thick] (v6\m)--(c2\m) ;
	\draw[thick] (v6\m)--(c3\m) ;			
}	
\node (a) [below=0.5cm of v12] {(a)}; 

\node  [right=of c24] {$ \begin{pmatrix}
	B_0		&		&		&				\\
	B_1		&B_0	&		&				\\
	\vdots	&B_1	&B_0	&				\\
	B_T		&\vdots	&B_1	&\ddots			\\
			&B_T	&\vdots	&\ddots			\\
			&		&B_T	&\ddots			\\
			&		&		&\ddots
	\end{pmatrix}$}; 
\node (x) [right=6.5cm of c33] {};
\node  [above=3mm of x] {\Huge $ 0 $};
\node (z) [right=5.1cm of c13] {};
\node  [below=0.4cm of z] {\Huge $ 0 $};
\node (b) [ right=3mm of a] {(b)};
\end{tikzpicture}
\end{center}
     
     \caption{\label{Fig:SC36}
     (a) The $(3,6)$-regular SC-LDPC protograph from Example~\ref{Ex:Background}. (b) The infinite matrix representing the protograph coupling operation.}      
\end{figure}
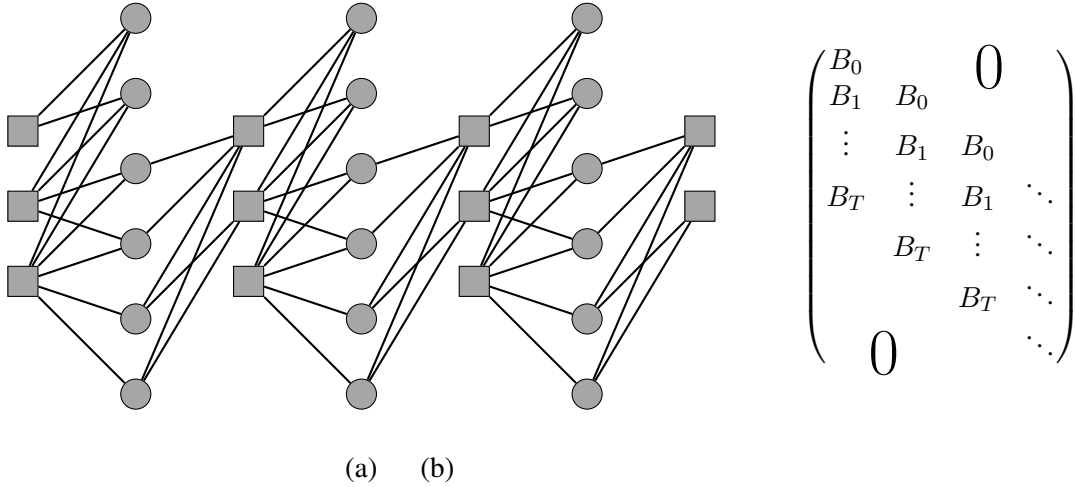

\section{Sub-Blocked SC-LDPC Codes}
\label{Sec:SC-LDPCL Codes}

Consider a coupled protograph $\mathcal{G}=\left(\mathcal{V}\cup\mathcal{C},\mathcal{E}\right)$ (in the rest of the paper, this notation will refer to the {\em coupled} protograph). To obtain a sub-blocked SC-LDPC code (as done in \cite{RamCassuto18a} without spatial coupling), we divide $\mathcal{V}$ into $M>1$ disjoint sets $\{\mathcal{V}_m\}_{m=1}^M$, and refer to $\mathcal{V}$ as the code block and to the $M$ subsets $\{\mathcal{V}_m\}_{m=1}^M$ as sub-blocks (SBs). In what follows, let $H=H_\mathcal{G}$ be a bi-adjacency matrix representing the coupled protograph $\mathcal{G}$, and let $m\in\{1,2,\ldots,M\}$ be a SB index. When decoding SB $ m $ locally, all of the VNs outside the sub-block are treated as erasures; hence only CNs connected inside sub-block $ m $ are relevant to local decoding. We call these CNs \emph{local checks} (LCs). CNs that are not LCs are called coupling checks (CCs).

\begin{definition}\label{Def:LC and CC}~
	
	\begin{enumerate}
		\item If VN $j\in \mathcal{V}$ belongs to SB $m$, we write $j \in \mathcal{V}_m$.
		\item CN $i\in \mathcal{C}$ is said to be an LC in SB $m$ if and only if $\{j\colon H_{i,j}=1\} \subseteq \mathcal{V}_m,\;$ and we write $i \in \mathcal{C}_m$. 
		\item The local protograph of SB $m$ is the sub-graph $\mathcal{G}_m=\left( \mathcal{V}_m\cup\mathcal{C}_m,\mathcal{E}_m\right)$, where $\mathcal{E}_m$ is the set of edges in $\mathcal{E}$ that connect between VNs in $\mathcal{V}_m$ and CNs in $\mathcal{C}_m$.
		\item The global and local BP decoding thresholds are given by
		$\epsilon_G^*\triangleq \epsilon^*\left(\mathcal{G}\right)$ and $\epsilon_m^*\triangleq \epsilon^*\left(\mathcal{G}_m\right),
		$
		respectively.
	\end{enumerate}
\end{definition}

\begin{example}\label{Ex:H}
	Let $\mathcal{G}$ be the coupled protograph from Example~\ref{Ex:Background} (see Figure~\ref{Fig:SC36}(a)).
	If we divide $\mathcal{V}$ into $M=3$ equally sized SBs, then $\mathcal{V}_1=\{1,2,3,4,5,6\}$, $\mathcal{V}_2=\{7,8,9,10,11,12\}$, $\mathcal{V}_3=\{13,14,15,16,17,18\}$, and $\mathcal{C}_1=\{1,2,3\}$, $\mathcal{C}_2=\{6\}$, $\mathcal{C}_3=\{9,10,11\}$. The local protographs $\mathcal{G}_1,\mathcal{G}_2$ and $\mathcal{G}_3$ are illustrated in Figure~\ref{Fig:36_123}. The local decoding thresholds in this case are all zero, i.e., $\epsilon_1^*=\epsilon_2^*=\epsilon_3^*=0$. As we will see later, zero local thresholds are a general phenomenon in SC-LDPC codes, unless proper design measures are taken.
\end{example}

\begin{figure}
	\begin{center}
		\begin{tikzpicture}\label{Tikz:36 Local Graph}
		\tikzstyle{cnode}=[rectangle,draw,fill=gray!70!white,minimum size=4mm]
		\tikzstyle{vnode}=[circle,draw,fill=gray!70!white,minimum size=4mm]
		\pgfmathsetmacro{\x}{3}
		\pgfmathsetmacro{\y}{0.7}
		
		\foreach \m in {1,2,3}
		{

			\foreach \v in {1,2,3,4,5,6}
			{
				\node[vnode] (v\v\m) at (\m*\x-0.5*\x,\v*\y-2.5*\y) {};	
			}
			
		}
		\node (G1) [below=5mm of v11] {$ \mathcal{G}_1 $};	
		\node (G2) [below=5mm of v12] {$ \mathcal{G}_2 $};	
		\node (G3) [below=5mm of v13] {$ \mathcal{G}_3 $};	
		
		\foreach \c in {1,2,3}
		{
			\node[cnode] (c\c1) at (1*\x-\x,\c*\y-\y) {};	
		}
		\node[cnode] (c12) at (2*\x-\x,1*\y-\y) {};	
		\node[cnode] (c13) at (3*\x-\x,1*\y-\y) {};	
		\node[cnode] (c34) at (4*\x-\x,3*\y-\y) {};	
		\node[cnode] (c24) at (4*\x-\x,2*\y-\y) {};

		\draw[thick] (v11)--(c11) ;
		\draw[thick] (v21)--(c11) ;
		\draw[thick] (c21)--(v31)--(c11) ;
		\draw[thick] (c21)--(v41)--(c11) ;.
		\draw[thick] (c31)--(v51)--(c21)--(v51)--(c11) ;
		\draw[thick] (c31)--(v61)--(c21)--(v61)--(c11)  ;
		
		\foreach \v in {1,...,6}
		{
			\foreach \m in {2,3}
			\draw[thick] (v\v\m)--(c1\m);
			
		}
		\draw[thick] (v13)--(c24)--(v23);
		\draw[thick] (v33)--(c34)--(v43);
		\draw[thick] (v13)--(c34)--(v23);
		\end{tikzpicture}
	\end{center}
	\caption{\label{Fig:36_123}
		The \emph{local protographs} from Example~\ref{Ex:H}.}
\end{figure}
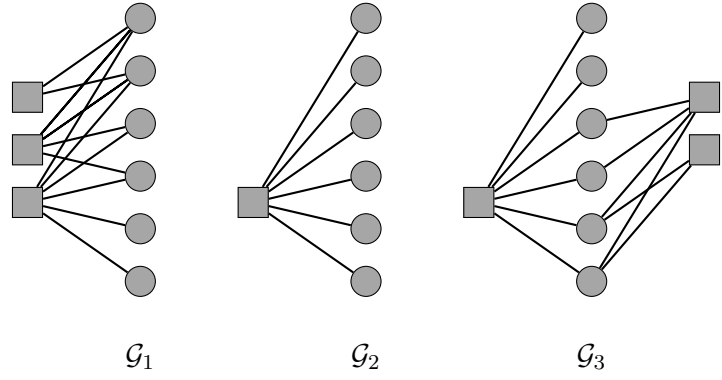

\subsection{Zero Local Threshold}
\label{Sub:Basics}
In this subsection, we state results concerning thresholds of sub-block protographs induced by the coupling process. Giving an explicit analytical expression is, in general, not an easy task since many densities should be tracked. Instead, bounds on the threshold are derived. We show that these results imply that the local thresholds in SC-LDPC protographs \er{are} zero, unless some specific design measures (which we address later) are taken.

\begin{lemma}\label{Lemma:Th UB}
Let $H$ be a bi-adjacency matrix representing a protograph $\mathcal{G}=(\mathcal{V}\,\cup\,\mathcal{C},\mathcal{E})$. Let $\mathcal{J}\subseteq \{1,2,\ldots,\big| \mathcal{V} \big|\}$ and $\mathcal{I}\subseteq \{1,2,\ldots, \big| \mathcal{C} \big| \}$ be sets of column and row indices, respectively, and let $H_\mathcal{J}$ (resp. $H^{(\mathcal{I})}$) be the sub-matrix consisting of the columns (resp. rows) of $H$ indexed by $\mathcal{J}$ (resp. $\mathcal{I}$). Then,
\begin{align} \label{Eq:Th Bounds}
\epsilon^*\left(H^{(\mathcal{I})}\right) \leq \epsilon^*\left(H\right) \leq \epsilon^*\left(H_\mathcal{J}\right).
\end{align}
\end{lemma} 

\begin{proof}
See Appendix~\ref{App:Th UB}.
\end{proof}

The scope of the next lemma is a protograph that has poor BP performance. This protograph appears as a sub-graph in many SC-LDPC protographs, and its properties strongly affect local decoding.

\er{
	\begin{lemma}\label{Lemma:tril}
		For $ p\geq 1 $, let $ A $ be a $ p\times p $ lower triangular matrix with a full-ones diagonal, i.e., 
		\[
		A = \begin{pmatrix}
		1		&	0		&	\cdots	&	0		&	0\\
		a_{2,1}	&	1		&	\cdots	&	0		&	0\\
		\vdots	&	\vdots	&	\ddots	&	\vdots	&	\vdots\\
		a_{p-1,1}&a_{p-1,2}	&	\cdots	&	1		&	0\\
		a_{p,1}	&	a_{p,2}	&	\cdots	&a_{p,p-1}	&	1\\
		\end{pmatrix}\,,
		\]
		where $ a_{i,j}\in\mathbb{N} $ for $ i>j $.	Let $ \mathbf{c}\in\mathbb{N}^p$ be a column vector. Then, the threshold of 
		$ \left (\mathbf{c}\;|\;A\right )$
		is zero, where the symbol $|$ represents horizontal concatenation.
	\end{lemma}

	\begin{proof}
	We first prove the case where $  a_{i,j}=1 $ for every $ i>j $, and $\mathbf{c}=\mathbf1$, i.e., a full-ones column vector. In this case, the protograph represented by $ \left (\mathbf{c}\;|\;A\right )$ is illustrated in Figure~\ref{Fig:BadProto}, where the lower and upper edge connections correspond to $ \mathbf{c} $ and $ A $, respectively, and labels are given inside nodes (VN $ p+1 $ corresponds to the first column in $  \left (\mathbf{c}\;|\;A\right )$).
	\begin{figure}
		\begin{center}
			\begin{tikzpicture}
			\tikzstyle{cnode}=[rectangle,draw,fill=gray!70!white,minimum size=8mm]
			\tikzstyle{vnode}=[circle,draw,fill=gray!70!white,minimum size=9mm]
			\pgfmathsetmacro{\x}{3}
			\pgfmathsetmacro{\y}{12}
			
			\node (c1) [cnode] {\tiny{$1$}};\node (c2) [cnode,right=\x mm of c1] {\tiny{$2$}};
			\node (dots1) [right=\x mm of c2] {{$\cdots$}};
			\node (c3) [cnode,right=\x mm of dots1] {\tiny{$p-1$}};\node (c4) [cnode,right=\x mm of c3] {\tiny{$p$}};

			\node[vnode,below=\y mm of c1] (v1) {\tiny$1$};	\draw [thick] (c1.south)--(v1.north);\draw [thick] (c2.south)--(v1.north);\draw [thick] (c3.south)--(v1.north);\draw [thick] (c4.south)--(v1.north);
			\node[vnode,below=\y mm of c2] (v2) {\tiny$2$};	\draw [thick] (c2.south)--(v2.north);\draw [thick] (c3.south)--(v2.north);\draw [thick] (c4.south)--(v2.north);
			\node[vnode,below=\y mm of c3] (v3) {\tiny$p-1$};	\draw [thick] (c3.south)--(v3.north);\draw [thick] (c4.south)--(v3.north);
			\node[vnode,below=\y mm of c4] (v4) {\tiny$p$};	\draw [thick] (c4.south)--(v4.north);
			
			\node[vnode,above=\y mm of dots1] (vlp1) {\tiny$p+1$};	
			\node (dots2) at ($(v2)!0.5!(v3)$) {{$\cdots$}};
			
			\draw [thick] (c4.north)--(vlp1.south);\draw [thick] (c3.north)--(vlp1.south);\draw [thick] (c2.north)--(vlp1.south);\draw [thick] (c1.north)--(vlp1.south);
			\end{tikzpicture}
		\end{center}
		\caption{\label{Fig:BadProto} The protograph represented by $ \left (\mathbf{c}\;|\;A\right ) $ in Lemma~\ref{Lemma:tril} in the full-ones case.}
	\end{figure}
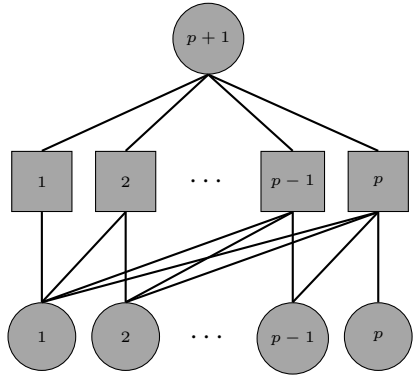
	
	Let $ \epsilon\in(0,1] $, and for every $ \ell\geq 1 $, $ i\in\{1,\ldots,l\} $ and $ j\in\{1,\ldots,p+1\} $, let $ x_\ell(i,j,\epsilon) $ (resp. $ u_\ell(i,j,\epsilon) $) denote the erasure rate of a $ \vn{j} \to \cn{i}$ (resp. $\cn{i} \to \vn{j}$) message in iteration $ \ell $. 	
	In view of \eqref{Eq:DE Vars}--\eqref{Eq:DE Checks}, since $ \vn{p} $ is of degree 1, then $ x_\ell(p,p,\epsilon )=\epsilon $ for every $ \ell $. Thus, $ u_\ell(p,j,\epsilon )>\epsilon $ for every iteration $ \ell $ and VN $ j\neq p $. Consequently, $ x_\ell(p-1,p-1,\epsilon )>\epsilon^2 $ and $ u_\ell(p-1,j,\epsilon )>\epsilon^2 $ for every $ \ell $ and  $ j\neq p-1 $. Similarly, we get by induction that for every iteration $ \ell $
	\begin{align}\label{Eq:xell}
	\begin{array}{lll}
	x_\ell(j,j,\epsilon )	&>&\epsilon\prod_{i=j+1}^{p}\epsilon^{2^{p-i}} \\
	&=&\epsilon^{2^{p-j}}
	\end{array}
	,\quad  \forall 1\leq j\leq p.
	\end{align} 
	This implies that for every iteration $ \ell $ and CN $ i\in\{1,\ldots,p\} $, $ u_\ell(i,p+1,\epsilon )>\epsilon^{2^{p-i}}  $, so for every iteration, the erasure rate of VN $ p+1 $ is bounded below by
	\begin{align}\label{Eq:VNlp1}
	\begin{array}{lll}
	P_\ell(p+1,\epsilon ) &>& \epsilon \prod_{i=1}^p \epsilon^{2^{p-i}} \\
	&=&  \epsilon^{2^{p}} \\
	&>& 0\,.
	\end{array}
	\end{align}
	Since this holds for every $ \epsilon\in(0,1] $, then the threshold is zero.
	
	We now relax the assumptions that  $  a_{i,j}=1 $ for every $ i>j $ and that $\mathbf{c}=(c_1,\ldots,c_p)=\mathbf1$, and consider the general case. For every VN $ j\in\{1,2,\ldots,p-1\} $ (lower part in Figure~\ref{Fig:BadProto}), let $ s_j = \sum_{i=j+1}^p a_{i,j}2^{p-i}$. The same arguments above hold with  modification in \eqref{Eq:xell} and \eqref{Eq:VNlp1} given by
	\begin{align}\label{Eq:xell2}
	\begin{array}{lll}
	x_\ell(j,j)	&>&\epsilon\prod_{i=j+1}^{p}\epsilon^{a_{i,j}2^{p-i}} \\
	&=&\epsilon^{1+s_j}
	\end{array}
	,\quad  \forall 1\leq j\leq p.
	\end{align} 
	and
	\begin{align}\label{Eq:VNlp12}
	\begin{array}{lll}
	P_\ell(p+1,\epsilon) &>& \epsilon \prod_{i=1}^p \epsilon^{c_i(1+s_i)} \\
	&>& 0\,,
	\end{array}
	\end{align}
	respectively.
	\end{proof}
}

\begin{theorem}\label{Th:no local}
Let $H$ be a binary bi-adjacency matrix representing an $ (l,r) $-regular SC-LDPC protograph $\mathcal{G}=(\mathcal{V}\cup\mathcal{C},\mathcal{E})$ constructed by truncating the infinite matrix in Figure~\ref{Fig:SC36}(b), and suppose $\mathcal{V}$ is divided into $M>1$ SBs. 
\er{
	If there are no two rows in $ \left (B_0 ; B_1\right ) $ that are all ones, then
	\begin{align*}
	\epsilon^*_m=0,\quad 2\leq m\leq M-1.
	\end{align*}
	If in addition, the partitioning is done via the cutting-vector method \cite{MitchellISIT2014}, then $\epsilon^*_1=\epsilon^*_M=0$.
}
\end{theorem}

Theorem~\ref{Th:no local} states a negative result on sub-block locality in existing SC-LDPC codes, and motivates a construction of multi-sub-block SC-LDPC codes which we address later.

\begin{remark}\label{Rem:why th=0}
Recall that the matrix $ B $ is an $ l $-by-$ r $ all-ones matrix. The ``no two full rows" property of the matrices $ \{B_\tau\}_{\tau=0}^1 $ in Theorem~\ref{Th:no local} holds in many SC-LDPC protographs in the literature, since it induces high global thresholds. 
In fact, the family of protographs covered by Theorem~\ref{Th:no local} is larger than it may seem in a first look. For example, the $(l,r)$ SC-LDPC ensemble from \cite[Definition 3]{Mitchell15} with $l=\mathrm{gcd}(l,r)$, is included in Theorem~\ref{Th:no local}. 
\end{remark}

\er{
\begin{proof}[Proof of Theorem~\ref{Th:no local}]
	Consider first non-termination sub-blocks, i.e., $m\in \{2,3,\ldots,M-1\}$.
	Since the base matrix $ B $ is an all-ones matrix, then any row in $ \left (B_0 ; B_1\right ) $ that is not full ones has ones outside of the SB boundaries, i.e., the corresponding check node is connected to an erasure; thus, the local decoder cannot use this check-node. Consequently, the local code has at most one local check node. This leads to a zero threshold.

	Now, consider sub-block $ m=1 $ whose (local) protograph is represented by $ B_0 $. Assume cutting-vector partitioning, and let  $ \mathbf\xi =\left (\xi_1,\xi_2,\ldots,\xi_{l-1},\xi_l \right )$ be the cutting vector, such that for every $ i\in\{1,2,\ldots,l-1,l\} $, $ \xi_i\in \{1,2,\ldots,r\}$ and $1\leq \xi_1 <\xi_2<\cdots<\xi_{l-1}<\xi_l \leq r$.
	In addition, let $ \xi_0=0 $ and $\xi_{l+1}=r+1 $. Since $ l\leq r-1 $ (else the rate is zero), then there exists $ i\in \{1,\ldots,l+1\} $ such that $ \xi_i-\xi_{i-1}\geq 2 $. Let $ I=\max\{1\leq i\leq l+1\colon \xi_i-\xi_{i-1}\geq 2\} $.  If $ I=l+1 $, then $ \xi_l\leq \xi_{l+1}- 2=r-1 $, and the rightmost column of $ B_0 $ is a zero column, which leads to a zero threshold. Else, let $ \tilde{B}_0 $ be the sub-matrix consisting of columns $ \xi_{I}-1 $ upto $ r $ of $ B $. 
	Note that $ \tilde{B}_0$ is in the form of 
	\[
	\left (
	\begin{array}{c}
	0\\
	\hline
	\mathbf{c}\;|\;A
	\end{array} \right ),
	\]
	where $ 0 $ is a zero matrix, and $ \left (\mathbf{c}\;|\;A\right )$ is in the form of Lemma~\ref{Lemma:tril}. Thus, $\epsilon^*(\tilde{B}_0) =0$, which combined with Lemma~\ref{Lemma:Th UB} implies that $ \epsilon^*(B_0)=0$, i.e., the threshold of sub-block $ m=1 $ is zero.
	The proof that $ \epsilon^*(B_1)=0 $ follows in a similar way, so the threshold of sub-block $ M $ is also zero.
\end{proof}
}

\begin{corollary} \label{Coro:l=2}
If $l=2$, then no $(l,r)$-regular SC-LDPC protograph can have a non-zero local decoding threshold.
\end{corollary}

\begin{proof}
If $ B=1^{2\times r }$, for some $ r\geq 3 $, then any decomposition of $ B $ into 2 non-zero matrices results in the ``no two full rows" condition in Theorem~\ref{Th:no local}.
\end{proof}

\subsection{A SC-LDPCL construction}
\label{Sub:Loc Dec}

Motivated by Theorem~\ref{Th:no local}, we introduce a construction of $(l,r)$-regular  SC-LDPC protographs having sub-block locality. 
The inputs to the construction are: the degrees $(l,r)$ (in view of Corollary~\ref{Coro:l=2}, we assume that $l\geq 3$), the number of SBs $M$, and a new coupling parameter $t\in\{1,2,\ldots,l-2\}$; the resulting protograph is an $(l,r,t)$ spatially coupled protograph with $M$ SBs, each consisting of $ r $ variable nodes. As we will see, $t$ serves as a design tool to control the trade-off between local and global decoding thresholds. 
\begin{construction}[SC-LDPCL]\label{Construct: SC-LDPCL}
Let $A_1$ be a $t\times r$ matrix given by 
\begin{center}
	\begin{tikzpicture}[every left delimiter/.style={xshift=.75em},	every right delimiter/.style={xshift=-.75em},
	]
		\matrix [matrix of math nodes,left delimiter=(,right delimiter=)] (m)
		{
			\underline{1} 	&               &               &               &               & \underline{0} \\
			\underline{1} 	& \underline{1} &               &\ph1			&               & \underline{0} \\
			\underline{1} 	& \underline{1} & \underline{1} &               &               & \underline{0} \\[-3mm]
			\vdots    		&    \vdots     &    \vdots     & \ddots        &               &  \vdots       \\
			\underline{1} 	& \underline{1} & \underline{1} & \cdots 		& \underline{1} & \underline{0} \\
		}; 
		\node (0)  at (m-2-4) {\Huge $ 0 $};
		\node [left = 1 mm of m] {$ A_1= $};
		\node [right = 1 mm of m] {$, $};
	\end{tikzpicture}
\end{center}
where $\underline{1}$ and $\underline{0}$ are length-$\left \lfloor \tfrac{r}{t+1} \right \rfloor$ all-one row vector and length-$\left (r-t\left \lfloor \tfrac{r}{t+1} \right \rfloor\right )$ all-zero row vector, respectively. Let $A_2$ be an all-ones $(l-t)\times r$ matrix.
We build the $(l,r,t)$ protograph as in Figure~\ref{Fig:SC36}(b) with memory $T=1$, and $M$ copies of $\begin{pmatrix} B_0 ; B_1 \end{pmatrix}$ on the diagonal, where
$B_0=\begin{pmatrix}A_1;A_2\end{pmatrix}$ and $B_1=1^{l\times r}-B_0.$
\end{construction}

The resulting coupled protograph $\mathcal{G}$ has $rM$ VNs and $lM+t$ CNs, so the design rate is $R_{\mathcal{G}}=1-\frac{l}{r}-\frac{t}{rM}$. For every $m\in\{2,\ldots,M-1\}$, the local graph $\mathcal{G}_m$ is represented by $A_2$ which is $(l-t,r)$-regular, and for $m=1$ and $m=M$, the local graph $\mathcal{G}_m$ is represented by 
$\begin{pmatrix} A_1 ; A_2 \end{pmatrix}$  and $\begin{pmatrix} A_2 ; \bar{A_1} \end{pmatrix}$, 
respectively, where $\bar{A_1}$ is the complement of $A_1$. Thus, for every $m\in\{2,\ldots,M-1\}$, $\epsilon^*_m=\epsilon^*(l-t,r)>0$, and Lemma~\ref{Lemma:Th UB} implies that for $m\in\{1,M\}$, $\epsilon^*_m\geq\epsilon^*(l-t,r)>0$, where $\epsilon^*(l-t,r)$ is the BP threshold of the $(l-t,r)$-regular LDPC code ensemble (i.e., uncoupled). \er{In general, due to termination check nodes, SBs $ 1 $ and $ M $ have better local thresholds than the SBs $\{2,\ldots,M-1\}$.}
\er{
	\begin{remark}
		The codes constructed by Construction~\ref{Construct: SC-LDPCL} have memory $ T=1 $. In Section~\ref{Sec:Gen Construct} we give a generalized construction for $ T\geq 1 $.
	\end{remark}		
}

\begin{example} \label{Ex:SC-LDPCL 3,6}
Figure~\ref{Fig:SC-LDPCL36} illustrates the $(3,6,1)$ SC-LDPCL protograph with $M=3$ SBs. In this case\er{,} we have
$A_1=\begin{pmatrix}
1 & 1 & 1 & 0 & 0 & 0 
\end{pmatrix}$,
the design rate is $R=0.4444$, and the thresholds are: $\epsilon^*_G=0.4772$, $\epsilon^*_1=\epsilon^*_3=0.4298$, and $\epsilon^*_2=0.2$ ($\epsilon^*_2$ corresponds to the $(2,6)$-regular ensemble). Note that the global-threshold loss compared to the ordinary $ (3,6) $ SC-LDPC protograph from Example~\ref{Ex:Background}, which does not enable sub-block decoding, is $ 6.97\% $, while the design rate increases by $ 12.46\% $ (these differences diminish as the number of sub-blocks increases).
\end{example}

\begin{figure}
\begin{center}
\begin{tikzpicture}\label{Tikz:36 SCLDPCL}
   	\tikzstyle{cnode}=[rectangle,draw,fill=gray!70!white,minimum size=4mm]
    \tikzstyle{vnode}=[circle,draw,fill=gray!70!white,minimum size=4mm]
   	\pgfmathsetmacro{\x}{3}
   	\pgfmathsetmacro{\y}{1}
   	
   	\foreach \m in {1,2,3}
   	{
   		
   		\foreach \c in {1,2,3}
   		{
   			\node[cnode] (c\c\m) at (\m*\x-\x,\c*\y-\y) {};	
   		}
   		\foreach \v in {1,2,3,4,5,6}
   		{
   			\node[vnode] (v\v\m) at (\m*\x-0.5*\x,\v*\y-2.5*\y) {};	
    	}	
    }
    \node[cnode] (c34) at (4*\x-\x,3*\y-\y) {};	
	\node (cc) [above=1mm of v62,align=center] {\footnotesize coupling checks}; 
	\node (c32Aux) [cnode,outer sep=1mm,thick] at (c32){};
	\node (c33Aux) [cnode,outer sep=1mm,thick] at (c33){};
	\draw[->,thick] (cc.west)to[in= 90, out = 225] (c32Aux.north);
	\draw[->,thick] (cc.east)to[in= 90, out = 315] (c33Aux.north);
    \foreach \m in {1,2,3}
    {	
    	\pgfmathtruncatemacro{\k}{\m + 1}
    	\draw[thick] (v1\m)--(c1\m) ;
    	\draw[thick] (v1\m)--(c2\m) ;
    	\draw[thick] (v1\m)--(c3\k) ;	
    	\draw[thick] (v2\m)--(c1\m) ;
    	\draw[thick] (v2\m)--(c2\m) ;
    	\draw[thick] (v2\m)--(c3\k) ;
    		
    	\draw[thick] (v3\m)--(c1\m) ;
    	\draw[thick] (v3\m)--(c2\m) ;
    	\draw[thick] (v3\m)--(c3\k) ;
    	\draw[thick] (v4\m)--(c1\m) ;
    	\draw[thick] (v4\m)--(c2\m) ;
    	\draw[thick] (v4\m)--(c3\m) ;
    		
    	\draw[thick] (v5\m)--(c1\m) ;
    	\draw[thick] (v5\m)--(c2\m) ;
    	\draw[thick] (v5\m)--(c3\m) ;
    	\draw[thick] (v6\m)--(c1\m) ;
    	\draw[thick] (v6\m)--(c2\m) ;
    	\draw[thick] (v6\m)--(c3\m) ;			
    }	
	\node (a) [below=0.5cm of v12] {(a)}; 
	
	\node  [right=25mm of v53] (B1){$B_0= 
		\begin{pmatrix}
			1&1&1&0&0&0\\
			1&1&1&1&1&1\\
			1&1&1&1&1&1
		\end{pmatrix}$}; 
	\node  [below=of B1] {$B_1= 
		\begin{pmatrix}
		0&0&0&1&1&1\\
		0&0&0&0&0&0\\
		0&0&0&0&0&0
		\end{pmatrix}$}; 
	
	\node (b) [ right=7cm of a] {(b)};
\end{tikzpicture}
\end{center}
    
     \caption{\label{Fig:SC-LDPCL36}
     (a) The $(3,6,1)$ SC-LDPCL protograph with $M=3$ SBs; (b) the partition according to Construction~\ref{Construct: SC-LDPCL} with $ l=3,\;r=6,\;t=1 $.}
\end{figure}
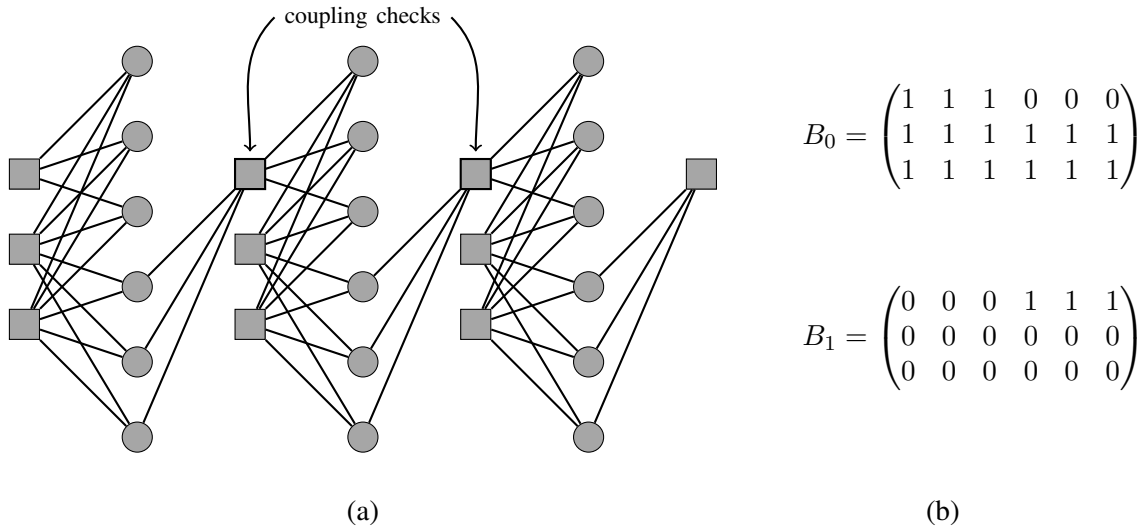

The coupling  parameter $t$ serves as a design tool that controls the trade-off between the local and global thresholds. More precisely, $ t $ designates the number of CCs connecting adjacent sub-blocks: when $t$ is small, more CNs are LCs, and the local threshold is higher on the expense of a lower global threshold; when $ t $ is large, the situation is reversed: global threshold is higher and local threshold is reduced. If one takes $t=0$, there are no CCs and the resulting protograph consists of $ M $ uncoupled $ (l,r) $-regular protographs, and if one takes $t=l-1$, the protograph is strongly coupled and there is only one LC (i.e., no locality). In view of the triviality of these extreme values, we restrict $ t\in\{1,2,\ldots,l-2\} $ in Construction~\ref{Construct: SC-LDPCL}.

\er{
	The matrices $A_1$ and $ A_2 $ in Construction~\ref{Construct: SC-LDPCL} specify the $ t $ coupling and $ l-t  $ local checks connections, respectively. Together, they describe the partition and coupling of a regular $ (l,r) $ protograph as described in Section~\ref{Sub:SCLDPC intro}. 
	When we construct sub-block locality SC-LDPC codes from local protographs, $ A_2 $ has to be the matrix $1^{(l-t)\times r}$, else the corresponding rows are no longer local checks (because a zero in $A_2$ implies a one in $B_1$ that connects the check to the sub-block on the left). Other local codes are possible if we allow non-regular protographs, which is not the scope of this paper.
	In contrast, the structure of $A_1$ does not affect the locality property of the code.
} 
The specific choice of $A_1$ in Construction~\ref{Construct: SC-LDPCL} is known as the ``cutting vector" approach \cite{MitchellISIT2014}, and one can use other partitions such as in \cite{HomaLara19} \er{for $ A_1 $ to optimize the performance}. 
We explore more partitions that induce sub-blocked SC codes in Section~\ref{Sec:Gen Construct}.

\begin{example} \label{Ex:SC-LDPCL}
	Table~\ref{Tbl:SC-LDPCL} details the design rates and thresholds of the \er{$(4,16,t)$ SC-LDPCL protographs for $t\in\{0,1,2,3\}$}, with $M=12$ SBs. The table exemplifies the role of $ t $ in trading off local and global performance (note that the table includes the extreme values of $t=0,3$). Further, sub-blocks $ 1 $ and $ 6 $ show better local thresholds than sub-blocks $ 2,3,4,5 $, and this difference is more prominent with higher $ t $ values. This phenomenon is due to $ t $ terminating check nodes in the first and last sub-blocks, which increase the local threshold compared to the inner sub-blocks. The parameter $t$ also affects the code design rate, as expressed by the right-most column.
\end{example}

\begin{table} [htbp]
	\er{
	\caption{\label{Tbl:SC-LDPCL} Thresholds and design rates for $ (4,16,t) $ SC-LDPCL protographs.}
	\begin{center}
		
		\begin{tabular}{c||c|c|c|c|c}
			
			$t$ & $\epsilon^*_1$& $\epsilon^*_2,\ldots,\epsilon^*_{11}$& $\epsilon^*_{12}$& $\epsilon^*_G$& $R$ \\
			\hline \hline
			$0$ &    0.1931    &      0.1931    &       0.1931   &     0.1931   &  0.75\\
			$1$ &    0.2036    &      0.1568    &       0.2036   &     0.2119   &  0.7438\\
			$2$ &    0.1995    &      0.0667    &       0.2142   &     0.2313   &  0.7375\\
			$3$ &    0   	   &      0 	    &       0	     & 	   0.2455   &  0.7313
		\end{tabular}
		
	\end{center}
	}
\end{table} 

\begin{example}\label{Ex: LocalVsGlobal Sim BEC}
	Figures~\ref{Fig:GlobBER48_BEC} and~\ref{Fig:LocBER48_BEC} show BEC performance of global and local decoding, respectively, of three SC-LDPCL codes constructed by Construction~\ref{Construct: SC-LDPCL} with degrees $ l=4,r=8$, number of sub-blocks $ M=3 $, lifting parameter $ L=625 $ and $ t=1,2,3 $ \er{(SB and full-block lengths are $ 5,000 $ and $15,000  $, respectively); local decoding is done on SB $ 2 $ which is a $ (4-t,8) $-regular code for $ t=1,2,3 $.} When $ t=3 $ we get the ordinary SC-LDPC $ (4,8) $ code (i.e., no locality); indeed the top curve in Figure~\ref{Fig:LocBER48_BEC} (dashed-blue-triangles) shows that this code has poor local-decoding performance (the output BER is approximately the channel parameter $ \epsilon $). The other options $ t=1 $ (solid-green-circles) and $ t=2 $ (dotted-red-squares) have much better local-decoding performance, where $t=1$ is superior to $t=2$, but less attractive in global decoding plotted in Figure~\ref{Fig:GlobBER48_BEC}.
\end{example}

\begin{figure}
	\begin{minipage}{0.45\textwidth}
	\begin{tikzpicture}

\begin{axis}[%
width=2.5in,
height=2in,
at={(0,0)},
scale only axis,
xmin=0.35,
xmax=0.5,
xlabel style={font=\color{black}},
xlabel={$ \epsilon $},
ymode=log,
ymin=1e-07,
ymax=1,
yminorticks=true,
ylabel style={font=\color{black}},
ylabel={Bit Erasure Rate},
axis background/.style={fill=white},
axis x line*=bottom,
axis y line*=left,
xmajorgrids,
ymajorgrids,
yminorgrids,
legend style={at={(0,1)}, anchor=north west, legend cell align=left, align=left, draw=black}
]
\addplot [color=green, line width=2.0pt, mark=*, mark options={solid, green}]
  table[row sep=crcr]{%
0.384210526315789	0\\
0.4	3.28528886241811e-07\\
0.405263157894737	3.79549266247379e-05\\
0.415789473684211	0.0163181238095238\\
0.426315789473684	0.200041123809524\\
0.436842105263158	0.330251428571429\\
0.447368421052632	0.361108466666667\\
0.457894736842105	0.383078952380952\\
0.468421052631579	0.40307979047619\\
0.478947368421053	0.421011495238095\\
0.489473684210526	0.437882380952381\\
0.5	0.454067390476191\\
};
\addlegendentry{$t = 1$}

\addplot [color=red, dotted, line width=2.0pt, mark=square, mark options={solid, red}]
  table[row sep=crcr]{%
0.415789473684211	0\\
0.426315789473684	2.01778396208855e-06\\
0.436842105263158	0.000838669428349669\\
0.447368421052632	0.0242655085197018\\
0.457894736842105	0.154464732618287\\
0.468421052631579	0.321044614331356\\
0.478947368421053	0.390918530351438\\
0.489473684210526	0.420408165982048\\
0.5	0.439293910695269\\
};
\addlegendentry{$t = 2$}

\addplot [color=blue, dashed, line width=2.0pt, mark=triangle, mark options={solid, blue}]
  table[row sep=crcr]{%
0.426315789473684	0\\
0.441	3.09201955252969e-07\\
0.447368421052632	1.953125e-05\\
0.457894736842105	0.00150381443643162\\
0.468421052631579	0.0398292410714286\\
0.478947368421053	0.202050938644689\\
0.489473684210526	0.337470457493895\\
0.5	0.39611522245116\\
};
\addlegendentry{$t = 3$}

\end{axis}
\end{tikzpicture}%
	\caption{\label{Fig:GlobBER48_BEC} BEC global-decoding simulations of three $ (4,8,t) $ SC-LDPCL codes: $ t=1 $ (solid-green-circles), $ t=2 $ (dotted-red-squares), and $t=3$ (dashed-blue-triangles). All codes are of total length $ n=15000 $ with $ M=3 $ sub-blocks.}
	\end{minipage}
\hspace*{6mm}
	\begin{minipage}{0.45\textwidth}
		%SCL t = 1 n = 5000 R = 0.625
%SCL t = 2 n = 2504 R = 0.75
%SCL t = 3 n = 1664 R = 0.875
\begin{tikzpicture}

\begin{axis}[%
width=2.5in,
height=2in,
at={(0in,0in)},
scale only axis,
xmin=0.01,
xmax=0.4,
xlabel style={font=\color{black}},
xlabel={$ \epsilon $},
ymode=log,
ymin=1e-07,
ymax=1,
yminorticks=true,
ylabel style={font=\color{black}},
ylabel={Bit Erasure Rate},
axis background/.style={fill=white},
axis x line*=bottom,
axis y line*=left,
xmajorgrids,
ymajorgrids,
yminorgrids,
legend style={at={(1,0)}, anchor=south east, legend cell align=left, align=left, draw=black}
]
\addplot [color=green, line width=2.0pt, mark=*, mark options={solid, green}]
  table[row sep=crcr]{%
0.268421052631579	3.82608695652174e-07\\
0.289473684210526	1.508038585209e-05\\
0.310526315789474	0.0291636\\
0.331578947368421	0.205560542857143\\
0.352631578947369	0.266417085714286\\
0.373684210526316	0.309432685714286\\
0.394736842105263	0.345608342857143\\
0.41578947368421	0.377716685714286\\
0.436842105263158	0.4074986\\
0.457894736842105	0.435315342857143\\
0.478947368421053	0.461629914285714\\
0.5	0.486755028571429\\
};
\addlegendentry{$t=1$}

\addplot [color=red, dotted, line width=2.0pt, mark=square, mark options={solid, red}]
  table[row sep=crcr]{%
0.01 0.0000011016 \\
0.05 0.000025375 \\
0.1	0.000399532177088087\\
0.121052631578947	0.00096240301232314\\
0.142105263157895	0.0033341510725696\\
0.163157894736842	0.0140341738931995\\
0.184210526315789	0.0404580100410771\\
0.205263157894737	0.07436438840712\\
0.226315789473684	0.111111706983113\\
0.247368421052632	0.148236022364217\\
0.268421052631579	0.184294043815609\\
0.289473684210526	0.218486478776814\\
0.310526315789474	0.250731572341397\\
0.331578947368421	0.281769055225924\\
0.352631578947369	0.3117007644911\\
0.373684210526316	0.340299178457325\\
0.394736842105263	0.367844933820174\\
0.41578947368421	0.394102863989046\\
0.436842105263158	0.419615301232314\\
0.457894736842105	0.444260668644455\\
0.478947368421053	0.468260497489731\\
0.5	0.49183107028754\\
};
\addlegendentry{$ t=2 $}

\addplot [color=blue, dashed, line width=2.0pt, mark=triangle, mark options={solid, blue}]
  table[row sep=crcr]{%
0.01 0.00067795 \\
0.05 0.0151 \\
0.1	0.0521474072802198\\
0.121052631578947	0.0721390796703297\\
0.142105263157895	0.0933378262362638\\
0.163157894736842	0.116052884615385\\
0.184210526315789	0.139843406593407\\
0.205263157894737	0.163757125686813\\
0.226315789473684	0.188985233516483\\
0.247368421052632	0.21363126717033\\
0.268421052631579	0.238624914148352\\
0.289473684210526	0.262711195054945\\
0.310526315789474	0.287526871565934\\
0.331578947368421	0.311736778846154\\
0.352631578947369	0.335697287087912\\
0.373684210526316	0.359574433379121\\
0.394736842105263	0.382889165521978\\
0.41578947368421	0.406218835851648\\
0.436842105263158	0.429215659340659\\
0.457894736842105	0.451487379807692\\
0.478947368421053	0.474051682692308\\
0.5	0.496067736950549\\
};
\addlegendentry{$ t=3 $}

\end{axis}
\end{tikzpicture}%
		\caption{\label{Fig:LocBER48_BEC} BEC local-decoding simulations of three $ (4,8,t) $ SC-LDPCL codes: $ t=1 $ (solid-green-circles), $ t=2 $ (dotted-red-squares), and $t=3$ (dashed-blue-triangles)\er{, that are $ (3,8) $-regular, $ (2,8) $-regular and $ (1,8) $-regular, respectively.} All sub-blocks are of length $ n=5000 $}
	\end{minipage}
\end{figure}

\begin{example}\label{Ex: LocalVsGlobal Sim AWGN}
	Figures~\ref{Fig:GlobBER48_AWGN} and ~\ref{Fig:LocBER48_AWGN} show AWGN performance -- global and local, respectively -- \er{of codes generated from the $ (4,8,t) $ protographs with $ t=1,2,3 $. The code has $ M=9 $ sub-blocks, and the lifting parameter is $ L=208 $ (SB and full-block lengths are $ 1,664 $ and $14,976  $, respectively); local decoding is done on SB $ 2 $ which is a $ (4-t,8) $-regular code for $ t=1,2,3 $.} As in the BEC plots, the AWGN plots exemplify the global-vs.-local trade-off introduced by the $ t $ parameter in Construction~\ref{Construct: SC-LDPCL}. Moreover, due to rate loss, the strongly coupled code ($ t=3 $) has worse local performance than the uncoded scheme.
	\er{
		Note that we used Construction~\ref{Construct: SC-LDPCL} without further optimization for the AWGN channel. While there is room for optimizing the protograph for any specific channel, we simulate AWGN without further optimization to show the general behavior of different parameters, in particular, that the local performance is extremely bad if no design measures are taken.
	}
\end{example}

\begin{figure}
	\begin{minipage}{0.45\textwidth}
		\begin{tikzpicture}

\begin{axis}[%
width=2.5in,
height=2.5in,
at={(0in,0in)},
scale only axis,
xmin=0.9,
xmax=2.185,
xlabel style={font=\color{black}},
xlabel={$E_b/N_0 \text{ [dB]}$},
ymode=log,
ymin=1e-08,
ymax=1,
yminorticks=true,
ylabel style={font=\color{black}},
ylabel={Bit Error Rate},
axis background/.style={fill=white},
axis x line*=bottom,
axis y line*=left,
xmajorgrids,
ymajorgrids,
yminorgrids,
legend style={at={(0,0)}, anchor=south west, legend cell align=left, align=left, draw=black}
]
\addplot [color=blue, dashed, line width=2pt, mark=triangle, mark options={solid, blue}]
  table[row sep=crcr]{%
0.889473684210543	0.0635890758547008\\
1.06842105263161	0.0164115918803419\\
1.24736842105261	0.00186298076923077\\
1.42631578947368	0.000251891121031746\\
1.60526315789474	2.67629275645308e-06\\
1.7842105263158	2.4484e-07\\
};
\addlegendentry{$ t=3 $}

\addplot [color=red, dotted, line width=2pt, mark=square, mark options={solid, red}]
  table[row sep=crcr]{%
1.07894736842104	0.0574311102236422\\
1.23684210526318	0.0130953807241747\\
1.39473684210526	0.000895899893503727\\
1.55263157894734	7.91785618988923e-05\\
1.71052631578948	4.46618743343983e-05\\
1.86842105263156	6.12241229784806e-06\\
2.0263157894737	2.42103077457281e-06\\
2.18421052631578	7.8842425849503e-07\\
};
\addlegendentry{$ t=2 $}

\addplot [color=green, line width=2pt, mark=*, mark options={solid, green}]
  table[row sep=crcr]{%
1.2684210526316	0.0862168888888889\\
1.40526315789475	0.049088\\
1.54210526315791	0.00634555555555556\\
1.67894736842106	0.000528249566724437\\
1.81578947368422	4.69572342504176e-06\\
1.95263157894738	8.0041164027214e-08\\
};
\addlegendentry{$ t=1 $}

\end{axis}

\end{tikzpicture}%
		\caption{\label{Fig:GlobBER48_AWGN}AWGN global-decoding simulations of three $ (4,8,t) $ SC-LDPCL codes ($ t=1,2,3$) with lifting factor $ L=208 $, full-block length $14976$, and $M=9 $ sub-blocks.
		}
	\end{minipage}
	\hspace*{6mm}
	\begin{minipage}{0.45\textwidth}
		\begin{tikzpicture}

\begin{axis}[%
width=2.5in,
height=2.5in,
at={(0in,0in)},
scale only axis,
xmin=0,
xmax=8,
xlabel style={font=\color{black}},
xlabel={$E_b/N_0 \text{ [dB]}$},
ymode=log,
ymin=1e-7,
ymax=1,
yminorticks=true,
ylabel style={font=\color{black}},
ylabel={Bit Error Rate},
axis background/.style={fill=white},
axis x line*=bottom,
axis y line*=left,
xmajorgrids,
ymajorgrids,
yminorgrids,
legend style={at={(0,0)}, anchor=south west, legend cell align=left, align=left, draw=black}
]
\addplot [color=black,dashed, line width = 2pt]
  table[row sep=crcr]{%
1	0.0562819519765415\\
1.00481218889689	0.0561824053928427\\
1.50001632527693	0.0464009690444679\\
2.00354215347085	0.0374467845308662\\
2.50367492227593	0.0296015439919003\\
3.00402473228428	0.0228282293296829\\
3.50282174420482	0.0171433258104942\\
4.00477758438916	0.0124609669715244\\
4.50145847321397	0.00878431909328362\\
5.00535422652996	0.00592773140269299\\
5.50789702031801	0.00383450258007257\\
6.0069954272841	0.00237141554815178\\
6.50020565912246	0.00139948065236355\\
7.0038890701116	0.000768916868082257\\
7.50800323299717	0.000394364241671991\\
8	0.000190907774075993\\
};
\addlegendentry{Uncoded}

\addplot [color=blue, dashed, line width=2pt, mark=triangle, mark options={solid, blue}]
  table[row sep=crcr]{%
1	0.138082932692308\\
1.73684210526318	0.1175\\
2.4736842105263	0.0953786057692307\\
3.21052631578948	0.0746153846153846\\
3.94736842105266	0.0554807692307692\\
4.68421052631578	0.0366526442307692\\
5.42105263157896	0.0238882211538462\\
6.15789473684208	0.0136117788461538\\
6.89473684210526	0.00762620192307692\\
7.63157894736844	0.00305889423076923\\
};
\addlegendentry{$ t=3 $}

\addplot [color=red, dotted, line width=2pt, mark=square, mark options={solid, red}]
  table[row sep=crcr]{%
4.5	0.0189137380191693\\
4.86842105263156	0.0108186900958466\\
5.23684210526318	0.00452476038338658\\
5.60526315789474	0.00195686900958466\\
5.9736842105263	0.000403354632587859\\
6.34210526315792	9.98402555910543e-05\\
6.71052631578948	5.36055063576131e-05\\
7.07894736842104	4.27124088090072e-05\\
7.44736842105266	1.08817717265454e-05\\
7.81578947368422	5.14514200591936e-06\\
};
\addlegendentry{$ t=2 $}

\addplot [color=green, line width=2pt, mark=*, mark options={solid, green}]
  table[row sep=crcr]{%
2.5	0.065783\\
2.67894736842106	0.054338\\
2.85789473684213	0.032899\\
3.03684210526313	0.010799\\
3.2157894736842	0.000814\\
3.39473684210526	0.000246601941747573\\
3.57368421052632	5.53572729259232e-07\\
};
\addlegendentry{$ t=1 $}

\end{axis}

\end{tikzpicture}%
		\caption{\label{Fig:LocBER48_AWGN}AWGN local-decoding simulations of the codes from Figure~\ref{Fig:GlobBER48_AWGN}; sub-blocks are of length $ 1664 $. \er{$ t=1$, $2$, and $ 3 $  correspond to $ (3,8) $, $ (2,8) $, and $ (1,8) $-regular codes, respectively.}}
	\end{minipage}
\end{figure}

\section{Semi-Global Decoding} \label{Sec:SG}
In this section, we suggest a decoding strategy called \emph{semi-global} (SG) decoding, in which the decoder decodes a target SB $m\in\{1,2,\ldots,M\}$ with the help of additional $d$ neighbor SBs. $d$ is a parameter that bounds the number of additional SBs read for decoding one SB; hence, the smaller $d$ is, the faster access the code offers for single SBs. As exemplified later, the SG mode has a substantial complexity advantage over the global mode with a very small cost in threshold.

Consider a SC-LDPCL protograph with $ M>1 $ SBs; assume that the user is interested in SB $ m\in\{1,\ldots,M\}$. We call SB $ m $ the \emph{target}. In SG decoding\er{,} the decoder  uses $ d $ \emph{helper} SBs to decode the target in two phases: the \emph{helper phase}, and the \emph{target phase}. In the former, helper SBs are decoded locally, incorporating information from other previously decoded helper SBs. In the latter, the target SB is decoded while incorporating information from its neighboring helper SBs. 

\begin{example}\label{Ex:SG decoding}
	Figure~\ref{Fig:SG decoding} exemplifies SG decoding with $ d=4 $ helper SBs. The helper phase consists of decoding helpers $ m - 2 $ and $ m + 2 $ locally, and decoding helpers $ m - 1 $ and $ m + 1 $ using the information from helpers $ m - 2 $ and $ m + 2 $, respectively. In the target phase, SB $ m $ is decoded using information from both SB $ m - 1$ and $ m + 1$. 
	\begin{figure}
		\begin{center}
			\begin{tikzpicture}[>=latex]\label{Tikz:SG decoding}
			
			\tikzstyle{SB}=[rectangle,very thick,draw, rounded corners, minimum width=1.1cm,minimum height=0.5cm,fill=white]
			\pgfmathsetmacro{\x}{1.1}
			\pgfmathsetmacro{\y}{1.2}
			\node (SB0) [SB] at (0*\x,0) {\footnotesize $m$};
			\node (SB1) [SB] at (1*\x,0) {\footnotesize $m + 1$};
			\node (SB2) [SB,fill=gray!30!white] at (2*\x,0) {\footnotesize $m + 2$};
			\node (SB-1) [SB] at (-1*\x,0) {\footnotesize $m - 1$};
			\node (SB-2) [SB,fill=gray!30!white] at (-2*\x,0) {\footnotesize $m - 2$};
			\node (rdots1) [right=0.3mm of SB2] {\footnotesize $ \dots $};
			\node (ldots) [left =0.3mm of SB-2] {\footnotesize step 1 $ \dots $};

			\node (SB0) [SB] at (0*\x,-\y) {\footnotesize$m$};
			\node (SB1) [SB,fill=gray!30!white] at (1*\x,-\y) {\footnotesize$m + 1$};
			\draw [->] (SB2)--(SB1) ;
			\node (SB2) [SB] at (2*\x,-\y) {\footnotesize$m + 2$};
			\node (SB-1) [SB,fill=gray!30!white] at (-1*\x,-\y) {\footnotesize$m - 1$};
			\draw [->] (SB-2)--(SB-1) ;
			\node (SB-2) [SB] at (-2*\x,-\y) {\footnotesize$m - 2$};
			\node (rdots2) [right=0.3mm of SB2] {\footnotesize$ \dots $};
			\node (ldots) [left =0.3mm of SB-2] {\footnotesize step 2 $ \dots $};
			\draw [thick,decorate,decoration={brace,amplitude=5pt,raise=6pt}] (rdots1.north)--(rdots2.south) node [black,midway,xshift=9mm,text width=1cm,align=center] {\footnotesize helper phase};
			
			\node (SB0) [SB,fill=gray!30!white] at (0*\x,-2*\y) {\footnotesize$m$};
			\draw [->] (SB1)--(SB0) ;
			\draw [->] (SB-1)--(SB0) ;
			\node (SB1) [SB] at (1*\x,-2*\y) {\footnotesize$m + 1$};
			
			\node (SB2) [SB] at (2*\x,-2*\y) {\footnotesize$m + 2$};
			\node (SB-1) [SB] at (-1*\x,-2*\y) {\footnotesize$m - 1$};
			
			\node (SB-2) [SB] at (-2*\x,-2*\y) {\footnotesize$m - 2$};
			\node (rdots3) [right=0.3mm of SB2] {\footnotesize$ \dots $};
			\node (ldots) [left =0.3mm of SB-2] {\footnotesize step 3 $ \dots $};
			\draw [thick,decorate,decoration={brace,amplitude=3pt,raise=33pt}] (SB2.north)--(SB2.south) node [black,midway,xshift=17mm,text width=1cm,align=center] {\footnotesize target phase};
			
			\end{tikzpicture}
		\end{center}
		\caption{\label{Fig:SG decoding}Example of SG decoding with target SB $ m\in[1:M] $, and $ d=4 $; the steps are shown from top to bottom. The gray SBs are those that are decoded in a given step, and the arrows represent information passed between sub-blocks.}
	\end{figure}
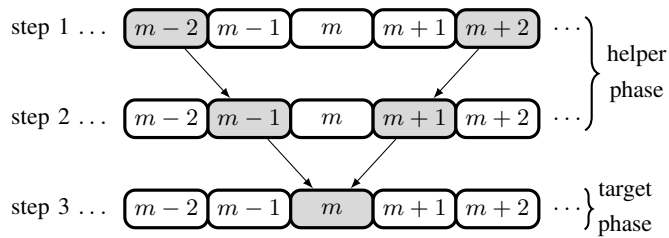
	
\end{example}

Note that semi-global decoding resembles window decoding of SC-LDPC codes (see \cite{IyenPapa12,IyenSiegel13,Lentmaier11}) but differs in: 1) \er{for a given target,} there is no overlap between two window positions, which decreases latency and complexity, and 2) decoding can start close to the target SB (i.e., not necessarily at the first or last SBs), allowing low-latency access to sub-blocks anywhere in the block. The SC-LDPCL protographs we propose for SG decoding are constructed with built-in structure to enable these distinctions.

The complexity reduction of SG decoding, compared to global decoding, comes from both specifying $d<M$, and from the fact that messages between sub-blocks are exchanged in one direction only. To see this, consider the $ (3,6,1) $ SC-LDPCL protograph in Figure~\ref{Fig:SC-LDPCL36}(a), and assume SG-decoding of target SB $ 2 $ with helpers SBs $ 1 $ and $ 3 $ (i.e., $ d=2 $). In the helper phase, we decode SB $1$ and $3$ locally -- possibly in parallel -- so the coupling checks are erased, and the decoder ignores all edges connected to them. In the target phase, the coupling checks are no longer erased, but they send information towards the target SB only. As a result, the six protograph edges connecting the coupling checks to SBs 1 and 3 do not participate in SG decoding.

Semi-global decoding is highly motivated by the locality property of sub-blocks in SC-LDPCL codes (SBs can be decoded locally), the spatial coupling of SBs (SBs can help their neighbor SBs), and by practical channels in storage devices, i.e., channels with variability \er{\cite{TaraUchi16,ShaAl20}}. Later, in Section~\ref{Sub:Channel}, we study the performance of SG decoding over such a channel. 

\subsection{SG Density-Evolution Analysis}
\label{Sub:SG DE}

We now \er{perform} an exact density-evolution analysis for target and helper SBs during SG erasure decoding. 
Due to the protograph's edge-regular structure, the general density-evolution equations in \eqref{Eq:DE Vars}--\eqref{Eq:DE init} can be reduced, yielding a simpler method to evaluate their performance. We denote the incoming (resp. outgoing) edges carrying messages to (resp. \er{from}) a helper SB by $ \underline \delta_I $ (resp. $ \underline \delta_O $). \er{Note that for termination helpers, i.e., the endpoint sub-blocks, we have $  \underline \delta_I=0  $.} The incoming messages to the target SB from the left-side and right-side helper SBs are denoted by $ \underline \delta_L $ and $ \underline \delta_R $, respectively. Note that $ \underline \delta_O $ of some helper is either $ \underline \delta_I $ of the next helper, or one of the incoming messages to the target, $ \underline\delta_L $ or $\underline \delta_R $ (see Figure~\ref{Fig:SG graph} below). With a slight abuse of notation, we also use $ \underline \delta_I\, \underline \delta_O ,\underline \delta_L,\underline \delta_R$ to mark the erasure probabilities carried on these edges. Since in the SG mode we decode SBs sequentially, then incoming erasure probabilities $ \underline \delta_L,\;\underline \delta_R,\;\underline \delta_I $ remain fixed during each decoding step.

\er{
	Decoding of sub-graphs with incoming and outgoing erasure rates was considered in a recent parallel work \cite[Section 3.D]{TruMitch19}, for the purpose of inter-connecting sub-chains of SC-LDPC codes.
	In this work, smaller units (sub-blocks) are inter-connected, and for the purpose of enabling efficient decoding of a single target sub-block. Toward that, we derive compact density-evolution (DE) equations and perform threshold analysis for decoding the target.
}

Consider a helper SB. From the structure of \er{the} $(l,r,t)$ SC-LDPCL protographs in Construction~\ref{Construct: SC-LDPCL}, there are $ t $ incoming messages  $ \underline \delta_I=\left(\delta_{I,1},\ldots,\delta_{I,t} \right) $, and $ t $ outgoing messages $ \underline \delta_O=\left( \delta_{O,1},\ldots,\delta_{O,t}\right) $. For every $ i\in\{1,\ldots,t\} $, the coupling check (CC) receiving $\delta_{I,i}$ (resp. sending $ \delta_{O,i} $) is denoted by $c_{I,i}$ (resp. $c_{O,i}$); see Figure~\ref{Fig:SG graph}. When the decoder tries to decode a helper SB, the CCs $\{c_{O,i}\}_{i=1}^t$ cannot help, and the decoder ignores the edges connected to them; the edges that participate in the iterative-decoding procedure are edges connected to local checks (LCs) and edges connected to CCs $\{c_{I,i}\}_{i=1}^t$ only. When the decoder finishes decoding the helper, it calculates $\underline\delta_{O}$ via the edges connected to $\{c_{O,i}\}_{i=1}^t$. 
In view of Construction~\ref{Construct: SC-LDPCL}, for every $ i\in\{1,\ldots,t\}$, $c_{I,i}$ (resp. $ c_{O,i} $) is connected to $r - i\left \lfloor \tfrac{r}{t+1} \right \rfloor$ (resp. $i\left\lfloor\tfrac{r}{t+1}\right\rfloor$) VNs in a helper SB. Despite the multiplicity of edges connected to $c_{I,i}$ and $c_{O,i}$, the one-directionality of the decoding algorithm allows us to consider a single (combined) constant input message $\delta_{I,i}$ and a single (combined) constant output message $\delta_{O,i}\,.$

In the target SB, only few adjustments of the above are needed. First, we have two active incoming messages $ \underline \delta_L $ and $ \underline \delta_R $, and we now mark the coupling check connected to $ \delta_{L,i},\;\delta_{R,i} $ by $ c_{L,i},\;c_{R,i} $, respectively. In view of these observations, we formally define the \emph{semi-global graph} $\mathcal{G}_{SG}$ as follows.

\begin{definition}[Semi-global graph: target]\label{Def:SG Graph}
	Let $\mathcal{G}$ be a $(l,r,t)$-SC-LDPCL protograph constructed by Construction~\ref{Construct: SC-LDPCL}. The \emph{semi-global graph} corresponding to $\mathcal{G}$, $\mathcal{G}_{SG}=\left (\mathcal{V}\cup\mathcal{C},\mathcal{E} \right )$, is a bipartite graph equipped with a VN labeling function $L_\mathcal{V}\colon \mathcal{V} \to \{1,2,\ldots,t+1\}$, an edge labeling function $L_\mathcal{E}\colon \mathcal{E} \to \left \{1,2,\ldots,(t+1)^2\right \}$, and $2t$ incoming edges $\{\delta_{R,1},\ldots,\delta_{R,t}\}$, and $\{\delta_{L,1},\ldots,\delta_{L,t}\}$ such that:
	\begin{enumerate}
		\item $\mathcal{V}=\{v_1,v_2,\ldots,v_r\}$ is a set of $r$ VNs.
		\item $\mathcal{C}$ is a set of $l+t$ CNs: $l-t$ of them are local checks (LCs), $t$ of them are right coupling checks (RCCs), and  another $t$ are left coupling checks (LCCs).
		\item We mark the $2t$ coupling checks as $c_{R,1},\ldots,c_{R,t}$, and $c_{L,1},\ldots,c_{L,t}$. For every $i\in\{1,2,\ldots,t\}$, $c_{R,i}$ (resp. $ c_{L,i} $) is connected to an incoming edge $\delta_{R,i}$ (resp. $ \delta_{L,i} $).
		\item The edges in $\mathcal{E}$ are determined by Construction~\ref{Construct: SC-LDPCL}. There is one edge between every LC and every VN. For every $i\in\{1,2,\ldots,t\}$, $c_{R,i}$ is connected to $r-i\left \lfloor\tfrac{r}{t+1}\right \rfloor$ VNs: one edge to each VN $v_j$, where $j\in \left \{ 1+i\left\lfloor\tfrac{r}{t+1}\right \rfloor,\ldots,r-1,r \right \}$, and for every $i\in\{1,2,\ldots,t\}$, $c_{L,i}$ is connected to $i\left \lfloor\tfrac{r}{t+1}\right \rfloor$ VNs: one edge to each VN $v_j$, where $j\in \left \{1,2,\ldots, i\left \lfloor\tfrac{r}{t+1}\right \rfloor\right \}$.
		\item For every $k\in\{1,2,\ldots,t+1\}$ and every VN $v\in \mathcal{V}$, $L_\mathcal{V}(v)=k$ if $v$ is connected to $k-1$ RCCs.
		\item For every edge $e=\{v,c\}\in \mathcal{E}$, if $L_\mathcal{V}(v)=k$ then,
		\er{
			\begin{align*}
			L_\mathcal{E}(e)=\left\{
			\begin{array}{ll}
			k\,, & c \text{ is a LC}\,, \\
			s_{k,t}+i\,, &c=c_{R,i}\,, \\
			v_{k,t}+i\,, &c=c_{L,i}\,, \\
			\end{array}
			\right.
			\end{align*}}
		where $ s_{k,t}\triangleq t + 1 + \tfrac{(k-1)(k-2)}{2} $, and $ v_{k,t}\triangleq 2t+1+\frac{t(t-1)}{2}+\frac{(k-1)(2t-k)}{2} $.
	\end{enumerate}
\end{definition}

\begin{remark}
	Definition~\ref{Def:SG Graph} refers to the target SB. The helper graph is similar with the only difference that $ t $ of the incoming edges (in right helpers $\{\delta_{L,1},\ldots,\delta_{L,t}\}$ and in left helpers $\{\delta_{R,1},\ldots,\delta_{R,t}\}$) turn into outgoing edges $\{\delta_{O,1},\ldots,\delta_{O,t}\}$, and the checks connected to them do not participate during the SB's decoding, except in sending the final message to the neighbor SB at the end of the SB decoding. This yields $ \frac{t^2+3t+2}{2}$ edge labels (in contrast to $ (t+1)^2 $ edge labels in Definition~\ref{Def:SG Graph}).
	 
\end{remark}

\begin{example}\label{Ex:SG graph}
	Figure~\ref{Fig:SG graph} illustrates the semi-global graph corresponding to the $(l=3,r=6,t=1)$ SC-LDPCL protograph, with the target on the right and the helper on the left. Node labels are drawn inside nodes, and edge labels are drawn on edges; there are $ t + 1 = 2 $ VN labels in both graphs,  $ (t+1)^2=4 $ edge labels in the target SB, and $ \frac{t^2+3t+2}{2}=3 $ edge labels in the helper SB. In the helper SB, the outgoing coupling check $ c_{O} $ is connected with dotted edges emphasizing the fact that it does not participate in the iterative decoding algorithm.
\end{example}

\begin{figure}
	\begin{center}	
		
		\begin{tikzpicture}[scale=0.6,>=latex]\label{Tikz:SG graph}
		\tikzstyle{cnode}=[rectangle,draw,fill=gray!70!white,minimum size=6mm]
		\tikzstyle{vnode}=[circle,draw,fill=gray!70!white,minimum size=2mm]
		\pgfmathsetmacro{\hscale}{0.7}
		\pgfmathsetmacro{\vscale}{0.8}
		\node[cnode] (LC1) at (-4*\vscale,\hscale*5) {\scriptsize LC };
		\node[cnode] (LC2) at (-4*\vscale,\hscale*9) {\scriptsize LC };
		
		\node[cnode] (CR1) at (3*\vscale,\hscale*2) {\scriptsize$ c_{I} $};
		\draw (CR1) node[right=7mm] (din1) {\scriptsize$ \delta_{I} $};
		\draw[->,thick] (din1)--(CR1);
		
		\node[cnode] (CL1) at (-4*\vscale,\hscale*12) {\scriptsize$ c_{O} $};
		\draw (CL1) node[left=7mm] (dout1) {\scriptsize$ \delta_{O} $};
		\draw[->,thick,dotted] (CL1)--(dout1);
		
		\foreach \y in {1,...,3}
		{
			\node[vnode] (VN\y) at (0*\vscale,\hscale*2*\y) {{\scriptsize $ 2 $}};
			\draw[violet] (CR1)--(VN\y) node[pos=0.35] {\scriptsize$ 3 $};
			\draw[blue] (LC1)--(VN\y) node[pos=0.2] {\scriptsize$ 2 $};
			\draw[blue] (LC2)--(VN\y) node[pos=0.9] {\scriptsize$ 2 $};
		}
		\foreach \y in {4,...,6}
		{
			\node[vnode] (VN\y) at (0*\vscale,\hscale*2*\y) {\scriptsize$ 1 $};
			\draw[thick,dotted] (CL1)--(VN\y) node[pos=0.2] {};
			\draw[red] (LC1)--(VN\y) node[pos=0.9] {\scriptsize$ 1 $};
			\draw[red] (LC2)--(VN\y) node[pos=0.2] {\scriptsize$ 1 $};
			\draw[thick,dotted] (CL1)--(VN\y) ;
		}
		
		\draw (VN1) node[below=5mm] (a) {(helper)};
		\pgfmathsetmacro{\vshift}{10}
		\node[cnode] (LC1) at (6*\vscale,\hscale*5) {\scriptsize LC };
		\node[cnode] (LC2) at (6*\vscale,\hscale*9) {\scriptsize LC };
		
		\node[cnode] (CR1) at (13*\vscale,\hscale*2) {\scriptsize$ c_{R} $};
		\draw (CR1) node[right=7mm] (din1) {\scriptsize$ \delta_{R} $};
		\draw[->,thick] (din1)--(CR1);
		
		\node[cnode] (CL1) at (6*\vscale,\hscale*12) {\scriptsize$ c_{L} $};
		\draw (CL1) node[left=7mm] (dout1) {\scriptsize$ \delta_{L} $};
		\draw[->,thick] (dout1)--(CL1);
		
		\foreach \y in {1,...,3}
		{
			\node[vnode] (VN\y) at (10*\vscale,\hscale*2*\y) {{\scriptsize $ 2 $}};
			\draw[violet] (CR1)--(VN\y) node[pos=0.35] {\scriptsize$ 3 $};
			\draw[blue] (LC1)--(VN\y) node[pos=0.2] {\scriptsize$ 2 $};
			\draw[blue] (LC2)--(VN\y) node[pos=0.9] {\scriptsize$ 2 $};
		}
		\foreach \y in {4,...,6}
		{
			\node[vnode] (VN\y) at (10*\vscale,\hscale*2*\y) {\scriptsize$ 1 $};
			\draw[dotted] (CL1)--(VN\y) node[pos=0.2] {\scriptsize$ 4 $};
			\draw[red] (LC1)--(VN\y) node[pos=0.9] {\scriptsize$ 1 $};
			\draw[red] (LC2)--(VN\y) node[pos=0.2] {\scriptsize$ 1 $};
			\draw (CL1)--(VN\y) ;
		}
		
		\draw (VN1) node[below=5mm] (target) {(target)};
		
		\end{tikzpicture}
		
	\end{center}
	\caption{\label{Fig:SG graph} The $(3,6,1)$ semi-global graph $\mathcal{G}_{SG}$ (with node and edge labels) as described in Definition~\ref{Def:SG Graph}; dotted edges do not participate during SB decoding, except in sending messages to $c_O$ at the end.}
	
\end{figure}
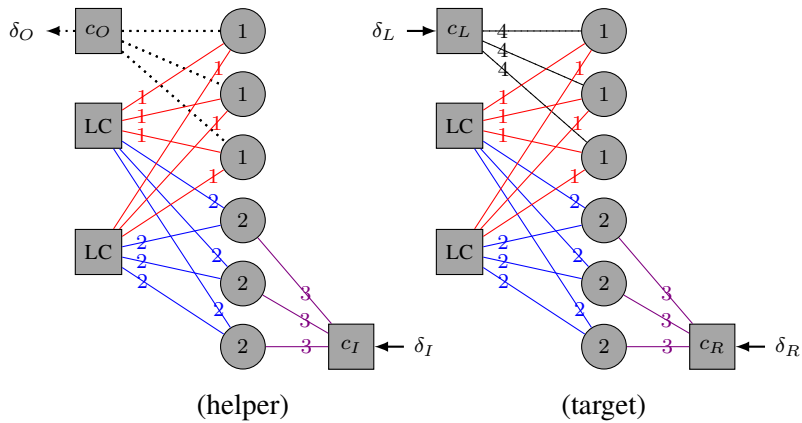

Since edges with the same labels are connected to VNs and CNs of the same degree, then in terms of density evolution, in any iteration, the erasure fraction of edges $ e_1,e_2\in \mathcal{E} $ coincide if $ L_\mathcal{E}(e_1)= L_\mathcal{E}(e_2) $. This structure is the key observation for simplifying the DE equations in \eqref{Eq:DE Vars}--\eqref{Eq:DE init} for the semi-global graph. Figure~\ref{Fig:SG nodes} illustrates the node and edge labels in a SG target graph, and shows the edge-label indexing of the DE equations that we derive in the following.

\begin{figure}
	\begin{center}
				
		\begin{tikzpicture}\label{Tikz:SG nodes}[>=latex]
		
		\tikzstyle{vnode}=[circle,draw,fill=gray!70!white,minimum size=12mm]
		\tikzstyle{cnode}=[rectangle,draw,fill=gray!70!white,minimum size=12mm]
		\pgfmathsetmacro{\hscale}{1}
		\pgfmathsetmacro{\vscale}{0.7}
		\node[vnode] (VNk) at (0,0) {$ k $};
		\draw  (VNk) node[below left=1 and 2] (type-k) {$ (k) $};
		\draw[thick] (VNk)--(type-k) node[pos=0.5,below=0.2] {$ l - t $} node[pos=0.5] {$ /$};
		\draw  (VNk) node[above right= 0.8 and 2] (type-k1) {$ (s_{k,t} + 1) $};
		\draw[thick] (VNk)--(type-k1.west) ;
		\draw  (VNk) node[above right= 0 and 2] (type-k2) {$ (s_{k,t} + 2) $};
		\draw[thick] (VNk)--(type-k2.west) ;
		\draw  (VNk) node[above right= -0.8 and 2] (type-kdots) {$ \qquad\vdots $};
		\draw  (VNk) node[above right= -1.6 and 2] (type-kk-1) {$ (s_{k,t} + k - 1) $};
		\draw[thick] (VNk)--(type-kk-1.west) ;
		\draw  (VNk) node[above left= 0.8 and 2] (type-kk) {$ (v_{k,t} + k) $};
		\draw[thick] (VNk)--(type-kk.east) ;
		\draw  (VNk) node[above left= 0.2 and 2] (type-kk1) {$ (v_{k,t} + k+1) $};
		\draw[thick] (VNk)--(type-kk1.east) ;
		\draw  (VNk) node[above left= -0.4 and 2.7] (type-kkdots) {$ \qquad \vdots $};
		\draw  (VNk) node[above left= -1 and 2] (type-kkt) {$ (v_{k,t} + t) $};
		\draw[thick] (VNk)--(type-kkt.east) ;
		
		\node[cnode] (LC) at (6*\hscale,0) {$ LC $};
		\draw  (LC) node[above right= 0.8 and 2] (type-1) {$ (1) $};
		\draw[thick] (LC)--(type-1.west) node[pos=0.5,above=0.2] {$ w $} node[pos=0.5] {$ /$};
		\draw  (LC) node[above right= 0 and 2]  {$ \;\;\vdots $};
		\draw  (LC) node[above right= -0.8 and 2] (type-t) {$ (t) $};
		\draw[thick] (LC)--(type-t.west) node[pos=0.5,above=0.2] {$ w $} node[pos=0.5] {$ /$};
		\draw  (LC) node[above right= -1.6 and 2] (type-t+1) {$ (t + 1) $};
		\draw[thick] (LC)--(type-t+1.west) node[pos=0.5,below=0.2] {$ r - tw $} node[pos=0.5] {$ /$};
		
		\node[cnode] (CRi) at (7*\hscale,5*\vscale) {$ c_{R,i} $};
		\draw  (CRi) node[right= 1] (dRi) {$ \delta_{R,i} $};
		\draw[thick,->] (dRi.west)--(CRi);
		\draw  (CRi) node[above left= 0.8 and 2] (type-i+1) {$ (s_{i+1,t}+i) $};
		\draw[thick] (CRi)--(type-i+1.east) node[pos=0.5,above=0.2] {$ w $} node[pos=0.5] {$ /$};
		\draw  (CRi) node[above left= 0 and 2]  {$ \vdots \qquad$};
		\draw  (CRi) node[above left= -0.8 and 2] (type-tt) {$ (s_{t,t}+i) $};
		\draw[thick] (CRi)--(type-tt.east) node[pos=0.5,above=0.2] {$ w $} node[pos=0.5] {$ /$};
		\draw  (CRi) node[above left= -1.6 and 2] (type-t+1t) {$ (s_{t+1,t}+i) $};
		\draw[thick] (CRi)--(type-t+1t.east) node[pos=0.5,below=0.2] {$ r - tw $} node[pos=0.5] {$ /$};
		
		\node[cnode] (CLi) at (-2*\hscale,5*\vscale) {$ c_{L,i} $};
		\draw  (CLi) node[left= 1] (dLi) {$ \delta_{L,i} $};
		\draw[thick,->] (dLi.east)--(CLi);
		\draw  (CLi) node[above right= 0.8 and 2] (type-1) {$ (v_{1,t}+i) $};
		\draw[thick] (CLi)--(type-1.west) node[pos=0.7,above=0.1] {$ w $} node[pos=0.5] {$ /$};
		\draw  (CLi) node[above right= 0 and 2] (type-2) {$ (v_{2,t}+i) $};
		\draw[thick] (CLi)--(type-2.west) node[pos=0.7,above=0.1] {$ w $} node[pos=0.5] {$ /$};
		\draw  (CLi) node[above right= -0.8 and 2.7]  {$ \vdots \qquad$};
		\draw  (CLi) node[above right= -1.6 and 2] (type-i) {$ (v_{i,t}+i) $};
		\draw[thick] (CLi)--(type-i.west) node[pos=0.7,above=0.1] {$ w $} node[pos=0.5] {$ /$};
		\end{tikzpicture}	
	\end{center}
	\caption{\label{Fig:SG nodes} The node and edge labels of an SG graph (target) $ \mathcal{G}_{SG} $ corresponding to an $ (l,r,t) $ SC-LDPCL protograph: $ k\in\{1,\ldots,t+1\}$ is a VN label and $ i\in\{1,\ldots,t\} $ is a CC index, $ s_{k,t}\triangleq t + 1 + \tfrac{(k-1)(k-2)}{2} $, $ v_{k,t}\triangleq 2t+1+\frac{t(t-1)}{2}+\frac{(k-1)(2t-k)}{2} $ and $ w\triangleq\left  \lfloor\tfrac{r}{t+1}\right \rfloor $. Node labels are drawn inside nodes, and edge labels appear in parenthesis. }
	
\end{figure}
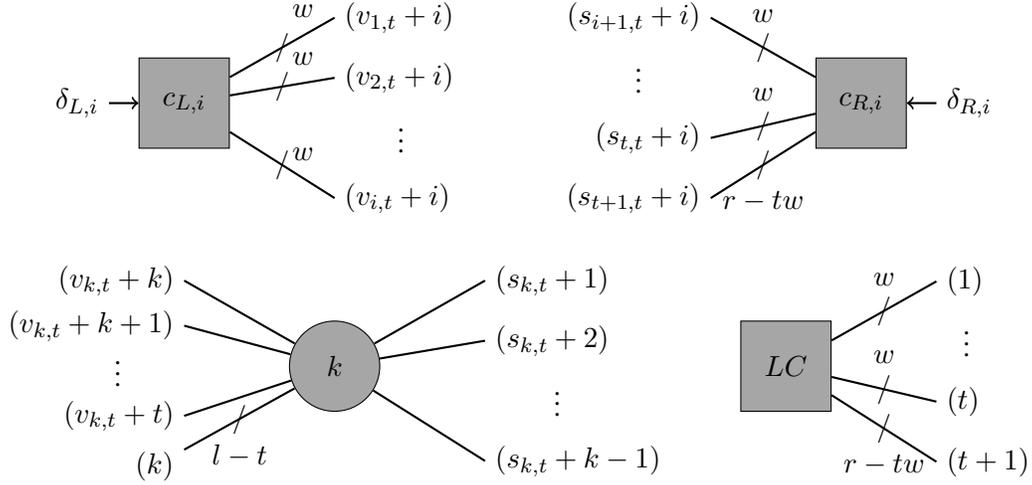
Recall the indices $ s_{k,t} $ and $ v_{k,t} $ from Definition~\ref{Def:SG Graph}, and define $ w\triangleq \left \lfloor\tfrac{r}{t+1}\right \rfloor $.
$s_{k,t}$ and $v_{k,t}$ are indices used below to capture the inter-sub-block coupling connectivity of Construction~\ref{Construct: SC-LDPCL}. Let $ x_\ell^{\left (L_\mathcal{E}(e)\right )} $ and $ u_\ell^{\left (L_\mathcal{E}(e)\right )} $ be the VN to CN and CN to VN erasure-message fractions, respectively, over an edge $ e\in\mathcal{E} $. 
We start with $ x_0^{(j)}=1$ for all $j\in\{1,2,\ldots,(t+1)^2\}$ just before the first iteration.
In view of Definition~\ref{Def:SG Graph} (see also Figure~\ref{Fig:SG nodes}), for every iteration $ \ell\geq 1 $, and every node label $ k\in \{1,\ldots,t + 1\} $, the fractions of erasure messages emanating from a VN labeled by $ k $ are given by
\begin{align}
\begin{split}
\label{Eq:SG_DE_VN}
&x_\ell^{(k)}=
\epsilon\cdot\left (u_\ell^{(k)}\right )^{l-t-1}
\prod_{j=1}^{k-1} u_\ell^{(s_{k,t}+j)}
\prod_{j=k}^{t} u_\ell^{(v_{k,t}+j)}\,,\\
&x_\ell^{(s_{k,t}+i)}=
\epsilon\cdot\left (u_\ell^{(k)}\right )^{l-t}
\prod_{j=1\atop j\neq i}^{k-1} u_\ell^{(s_{k,t}+j)}
\prod_{j=k}^{t} u_\ell^{(v_{k,t}+j)}\,,
\quad    \forall i \in \{1,2,\ldots,k - 1\}\,,\\
&x_\ell^{(v_{k,t}+i)}=
\epsilon\cdot\left (u_\ell^{(k)}\right )^{l-t}
\prod_{j=1}^{k-1} u_\ell^{(s_{k,t}+j)}
\prod_{j=k\atop j\neq i}^{t} u_\ell^{(v_{k,t}+j)}\,,
\quad    \forall i \in \{k,k + 1,\ldots,t\}\,,
\end{split}
\end{align}
where the product over an empty set is defined to be $ 1 $, and the fractions of erasure messages from any LC are
\begin{align}
\begin{split}
\label{Eq:SG_DE_LC}
&u_\ell^{(k)}=1-\left (1-x_{\ell-1}^{(k)}\right )^{w-1}\left (1-x_{\ell-1}^{(t+1)}\right )^{r-tw}\prod_{j=1\atop j\neq k}^t\left (1-x_{\ell-1}^{(j)}\right )^{w}\,,
\quad \forall k\in\{1,2,\ldots,t\}\;,\\
&u_\ell^{(t+1)}=1-\left (1-x_{\ell-1}^{(t+1)}\right )^{r-tw-1}\prod_{j=1}^t\left (1-x_{\ell-1}^{(j)}\right )^{w}\,.
\end{split}
\end{align}

Further, for every $ i\in\{1,2,\ldots,t\} $, the erasure-message fractions from $ c_{R,i} $ and from $ c_{L,i} $ are given by
\begin{align}
\begin{split}
\label{Eq:SG_DE__RCC}
&u_\ell^{(s_{k,t}+i)}=1 - 
\left (1 - \delta_{R,i}\right ) 
\left (1 - x_{\ell - 1}^{(s_{k,t}+i)}\right )^{w-1} 
\left (1 - x_{\ell - 1}^{(s_{t+1,t}+i)}\right )^{r-tw}   \\
& \hspace*{3cm}
\cdot\prod_{j=i+1\atop j\neq k}^t 
\left (1 - x_{\ell - 1}^{(s_{j,t}+i)}\right )^{w},\quad \forall k\in\{i + 1,\ldots,t\},\\
&u_\ell^{(s_{t+1,t}+i)}=1 - 
\left (1 - \delta_{R,i}\right )  \left(1 - x_{\ell - 1}^{(s_{t+1,t}+i)}\right )^{r-tw-1} 
\prod_{j=i+1}^t\left (1 - x_{\ell - 1}^{(s_{j,t}+i)}\right )^{w},
\end{split}
\end{align}
and
\begin{align}
\label{Eq:SG_DE__LCC}
&u_\ell^{(v_{k,t}+i)}=1 - 
\left (1 - \delta_{L,i}\right ) 
\left (1 - x_{\ell - 1}^{(v_{k,t}+i)}\right )^{w-1}  
\prod_{j=1\atop j\neq k}^i 
\left (1 - x_{\ell - 1}^{(v_{j,t}+i)}\right )^{w},\quad \forall k\in\{1,2,\ldots,i\}\,,
\end{align}
respectively. 

Consider now a helper SB. The difference\er{s} from the target-SB DE analysis \er{are} that 1) half of the incoming messages are disabled (i.e. are erasures), and 2) we need to calculate the outgoing erasure fraction. 
Assume a helper SB to the right of the target, corresponding to substituting $\underline \delta_L=\underline 1$ in \eqref{Eq:SG_DE__LCC}; the other option of a left helper is symmetric with  $\underline \delta_R=\underline 1$.
Substituting in \eqref{Eq:SG_DE__LCC}, 
\begin{align*}
u_\ell^{(v_{k,t}+i)}=1,\quad \forall\ell\geq 0,\,\forall k\in\{1,\ldots,t\},\,\forall i\in\{k,\ldots, t\},
\end{align*}
so the first two equations in \eqref{Eq:SG_DE_VN} change to
\begin{align}\label{Eq:SG_DE_VN_helper}
\begin{split}
&x_\ell^{(k)}=
\epsilon\cdot\left (u_\ell^{(k)}\right )^{l-t-1}
\prod_{j=1}^{k-1} u_\ell^{(s_{k,t}+j)}\,,\\
&x_\ell^{(s_{k,t}+i)}=
\epsilon\cdot\left (u_\ell^{(k)}\right )^{l-t}
\prod_{j=1\atop j\neq i}^{k-1} u_\ell^{(s_{k,t}+j)}\,,
\quad    \forall i \in \{1,2,\ldots,k - 1\}\,.
\end{split}
\end{align}
For better clarity the reader may skip to subsection~\ref{Sub:t=1} where equations \eqref{Eq:SG_DE_VN}--\eqref{Eq:SG_DE_VN_helper} are given for the special case $t=1$.

\begin{lemma}\label{Lemma: mono SG}
	For every semi-global graph and every edge label $ j\in\{1,2,\ldots,(t+1)^2 \}$, the sequences $ x^{(j)}_\ell,u^{(j)}_\ell $ defined in \eqref{Eq:SG_DE_VN}--\eqref{Eq:SG_DE_VN_helper} are monotonically non-increasing in $ \ell $ and are bounded in $ [0,1]. $
\end{lemma}
\begin{proof}
	Follows by a mathematical induction on $ \ell $; details are left to the reader.
\end{proof}

In view of Lemma~\ref{Lemma: mono SG}, for every edge label $ j\in\{1,2,\ldots,(t+1)^2 \}$, we define the limits $ x^{(j)}=\lim\limits_{\ell \to \infty}x^{(j)}_\ell,\;u^{(j)}=\lim\limits_{\ell \to \infty}u^{(j)}_\ell $. 
When the decoder finishes decoding a \emph{helper} SB (successfully or not), it sends messages to the next SB. The $ t $ erasure fractions of these messages are encapsulated in $ \underline\delta_O=(\delta_{O,1},\ldots,\delta_{O,t}) $ (see Figure~\ref{Fig:SG graph}). 
Similarly to \eqref{Eq:SG_DE_VN_helper}, the third equation in \eqref{Eq:SG_DE_VN} becomes
\begin{align*}
x^{(v_{k,t}+i)}=\epsilon \cdot \left ( u^{(k)}\right)^{l-t}\prod_{i=1}^{k-1}u^{(s_{k,t}+i)}\,.
\end{align*}
Moreover, for every $ i \in \{1,2,\ldots,t\}\,, $ $ c_{O,i} $ is connected to VNs labeled by $ k\in \{1,2,\ldots,i\} $ with $ w$ 
edges, thus the outgoing erasure fraction from $ c_{O,i} $ to the next SB is given by
\begin{align}\label{Eq:SG dout}
\delta_{O,i} = 1-\prod_{k=1}^i\left (1-x^{(v_{k,t}+i)} \right )^w.
\end{align}

\begin{theorem}[Semi-global density evolution]\label{Th:SG-DE t>1}
Let $\mathcal{G}_{SG}$ be a semi-global graph corresponding to a SB in an $(l,r,t)$ SC-LDPCL protograph, let $\epsilon$ be the channel erasure probability in this SB, and let $\underline \delta_L,\,\underline \delta_R$ be the incoming  erasure fractions from neighbor SBs. For every edge label $j\in\left \{1,2,\ldots,(t+1)^2\right \}$, let $x^{(j)}_\ell$ be the fraction of VN-to-CN erasure messages over any edge $e\in\mathcal{E}$ labeled with $L_\mathcal{E}(e)=j$, at iteration $\ell$ of the BP decoding algorithm over a lifted random Tanner graph, as the lifting parameter tends to infinity. Then, for a target SB, $x^{(j)}_\ell$ is given by equations \eqref{Eq:SG_DE_VN}--\eqref{Eq:SG_DE__LCC}, and for a helper SB it is given by equations \eqref{Eq:SG_DE_LC}, \eqref{Eq:SG_DE__RCC}, and \eqref{Eq:SG_DE_VN_helper}. Furthermore, for a helper SB, for every $ i\in \{1,2,\ldots,t \}$ $ \delta_{O,i} $ is given by \eqref{Eq:SG dout}.
\end{theorem}

\subsection{The $ t=1 $ Case}
\label{Sub:t=1}
In view of Definition~\ref{Def:SG Graph}, if $ t=1 $, then there are $(t+1)^2= 4 $ different edge types in the target semi-global graph (see Figure~\ref{Fig:SG graph}). However, as shown below, it is sufficient to track only $ 2 $ edge types. This simplification renders a two-dimensional graphical representation of the density-evolution equations. 

Substituting $ t=1 $ into \eqref{Eq:SG_DE_VN}--\eqref{Eq:SG_DE__LCC} yields $ 4 $ density-evolution equations, namely (note the scalar $\delta_L,\delta_R$)
\begin{align}
	\label{Eq:SG DE t=1}
	\begin{split}
	&x_\ell^{(1)}=\epsilon\left[ 1-\left( 1-x_{\ell-1}^{(1)}\right)^{w-1}\left(1-x_{\ell-1}^{(2)}\right)^{r-w}\right]^{l-2}\left[1-\left(1-x_{\ell-1}^{(4)}\right)^{w-1}\left(1-\delta_L\right) \right],\\
	&x_\ell^{(2)}=\epsilon\left[ 1-\left( 1-x_{\ell-1}^{(1)}\right)^{w}\left(1-x_{\ell-1}^{(2)}\right)^{r-w-1}\right]^{l-2} \left[1-\left(1-x_{\ell-1}^{(3)}\right)^{r-w-1}\left(1-\delta_R\right) \right],\\
	&x_\ell^{(3)}=\epsilon\left[ 1-\left( 1-x_{\ell-1}^{(1)}\right)^{w}\left(1-x_{\ell-1}^{(2)}\right)^{r-w-1}\right]^{l-1},\\
	&x_\ell^{(4)}=\epsilon\left[ 1-\left( 1-x_{\ell-1}^{(1)}\right)^{w-1}\left(1-x_{\ell-1}^{(2)}\right)^{r-w}\right]^{l-1},
	\end{split}
\end{align}
where $ w=\left\lfloor\tfrac{r}{2}\right\rfloor$. Since $ x^{(3)} $ and $ x^{(4)} $ both depend solely on $ x^{(1)} $ and $ x^{(2)} $, then we can substitute the last two equations into the first two equations, and get that for a fixed erasure probability $ \epsilon $ and fixed incoming erasure messages $ \delta_L ,\,\delta_R$, the quantities $ x^{(1)}_\ell $ and $ x^{(2)}_\ell $ are functions of $ x^{(1)}_{\ell-1},\,x^{(1)}_{\ell-2},\,x^{(2)}_{\ell-1}\,,$ and $ x^{(2)}_{\ell-2} $, written as

\begin{align}\label{Eq:SG t=1 simple DE}
\begin{split}
&x^{(1)}_\ell = \tilde{f}\left(\epsilon,\delta_L,x^{(1)}_{\ell-1},x^{(1)}_{\ell-2},x^{(2)}_{\ell-1},x^{(2)}_{\ell-2}\right), \ell\geq 2\\
&x^{(2)}_\ell = \tilde{g}\left(\epsilon,\delta_R,x^{(1)}_{\ell-1},x^{(1)}_{\ell-2},x^{(2)}_{\ell-1},x^{(2)}_{\ell-2}\right), \ell\geq 2\\
&x^{(1)}_1=x^{(2)}_1=\epsilon,\\
&x^{(1)}_0=x^{(2)}_0=1.
\end{split}
\end{align}
The functions $ \tilde{f}$ and $\tilde{g} $ are both continuous and monotonically non-decreasing, so by taking the limit $ \ell\to\infty $ in \eqref{Eq:SG t=1 simple DE} and marking $ x^{(k)}=\lim\limits_{\ell\to\infty}x^{(k)}_\ell,\;k\in \{1,2\} $,  we get a two-dimensional fixed-point characterization:
\begin{align}\label{Eq:SG fixed point}
\begin{split}
x^{(1)}= \tilde{f}\left(\epsilon,\delta_L,x^{(1)},x^{(1)},x^{(2)},x^{(2)}\right)\triangleq f\left(\epsilon,\delta_L,x^{(1)},x^{(2)}\right),\\
x^{(2)}= \tilde{g}\left(\epsilon,\delta_R,x^{(1)},x^{(1)},x^{(2)},x^{(2)}\right)\triangleq g\left(\epsilon,\delta_R,x^{(1)},x^{(2)}\right).
\end{split}
\end{align}
In \cite{RamCassuto18a}\er{,} a 2-D fixed-point characterization for ordinary  (i.e. non spatially coupled) LDPC codes was derived. In contrast to the derivations above that consider SG decoding, in \cite{RamCassuto18a} the analyzed decoding mode is global decoding, and consequently neither $ \underline\delta_R $ nor $ \underline \delta_L $ appear in the analysis.
\begin{remark}\label{Rem:DE t=1}
Equation \eqref{Eq:SG fixed point} refers to the target SB. If one considers a helper SB, one should set $ \delta_L $ (or $ \delta_R $) to $ 1 $. In this case, from \eqref{Eq:SG DE t=1} we get $ \left (x^{(4)}/\epsilon\right )^{l-2} = \left (x^{(1)}/\epsilon\right )^{l-1} $, which \er{together} with \eqref{Eq:SG dout} implies that the outgoing erasure fraction is given by 
\begin{align*}
\delta_O =1-\left(1-x^{(4)}\right)^{\left\lfloor \tfrac{r}{2}\right\rfloor} = 1-\left(1-\epsilon\left(\tfrac{x^{(1)}}{\epsilon} \right)^{\tfrac{l-1}{l-2}} \right)^{\left\lfloor \tfrac{r}{2}\right\rfloor}\;.
\end{align*}

\end{remark}

\begin{example}
	Figure~\ref{Fig:SG_DE} exemplifies equation \eqref{Eq:SG t=1 simple DE} and \eqref{Eq:SG fixed point} (solid black and dashed colored, respectively) for the $ (3,6,1) $ SC-LDPCL protograph (see Figure~\ref{Fig:SG graph}(target)). In both plots\er{,} $ \epsilon=0.5 $ and $ \delta_L=0.3 $, while in the left one $ \delta_R=0.3 $ and in the right $ \delta_R=0.5 $. As seen in the plots, if the $ f $-curve (dotted blue) intersects the $ g $-curve (dashed red), then the iterative process halts and fails. If the two curves do not intersect then the erasure fractions $ x^{(1)}_\ell $ and $ x^{(2)}_\ell $ approach zero as iterations proceed, and decoding succeeds.
\end{example}
\begin{figure}
	\begin{center}
		\input{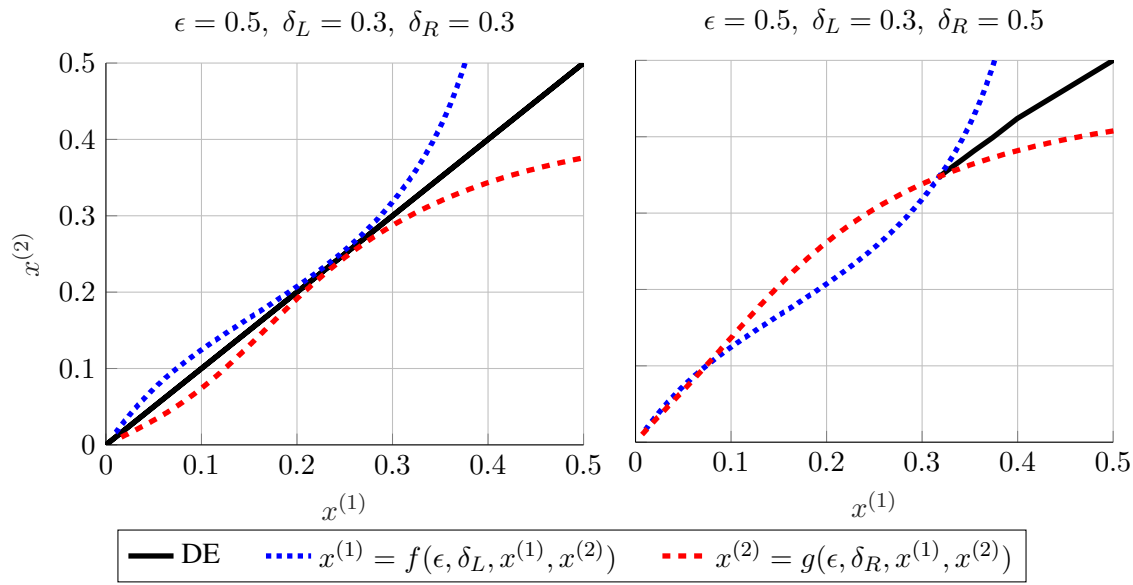}
	\end{center}
	\caption{\label{Fig:SG_DE} A graphical representation of the SG density-evolution equations in \eqref{Eq:SG DE t=1}. We used the $ (l=3,r=6,t=1) $ SC-LDPCL protograph and an erasure probability of $ \epsilon=0.5 $. The incoming erasure messages for the left figure are $\delta_L=0.3,\;\delta_R=0.3 $, where the density-evolution curve converges to the origin, indicating a decoding success. On the right-hand figure we have $ \delta_L=0.3,\;\delta_R=0.5 $, which leads to a halt in the BP process at $ (x^{(1)},x^{(2)})=(0.318,0.348)$.}
\end{figure}

\section{Semi-Global Performance Analysis}\label{Sub:SG analysis}
In this section\er{,} we analyze the SG decoding performance of SC-LDPCL codes over the BEC. 
First, we wish to compare the global and SG modes in terms of thresholds and complexity. Evidently, the threshold and complexity induced by the SG mode depend on the number of helpers $ d $; the larger $ d $ is, the higher the threshold and complexity are. 
\er{Note that unlike global decoding, each SG decoding instance aims to recover one target sub-block requested to be read from the entire codeword. This difference enables the complexity reduction of SG decoding compared to global decoding.}

\begin{remark} 
For simplicity, we assume for the rest of this section that $ t+1 $ divides $ r $, i.e., 
\begin{align}\label{Eq:w assumption}
w\triangleq\left \lfloor\tfrac{r}{t+1}\right\rfloor=\tfrac{r}{t+1}\,.
\end{align}
This assumption means that the SG graph $ \mathcal{G}_{SG} $ from Definition~\ref{Def:SG Graph} is symmetric in the sense that the degrees and connectivity of coupling checks $ \{c_{R,1},\ldots,c_{R,t}\} $ are identical to those of $ \{c_{L,1},\ldots,c_{L,t}\} $ (see Section~\ref{Sub:Loc Dec}).
\end{remark}

\subsection{SG Decoding Complexity}\label{Sub:complex}

We assume a fixed lifting parameter for the code, and a fixed number of BP iterations in any step of SG decoding performed over a subgraph of $ \mathcal{G} $; hence, to evaluate complexity we count the number of edges participating in the entire decoding process. We assume an $ (l,r,t) $ SC-LDPCL protograph with $ M $ SBs, each consisting of $ r $ VNs (i.e., a total count of $ Mr $ VNs in the protograph). In what follows\er{,} we mark by $ \chi_G$ and $\chi_{SG}$ the complexity of global and SG decoding, respectively.

In view of Construction~\ref{Construct: SC-LDPCL}, there are $ M l r  $ edges in the SC-LDPCL protograph, so the global-decoding complexity is given by $\chi_G= M l  r $. In view of Definition~\ref{Def:SG Graph} (see also Figure~\ref{Fig:SG graph}), in every helper SB the number of edges participating in decoding is  $ (l-t)r+\sum_{i=1}^t(r-iw) $, and in the target SB that number is $ (l-t)r+2\sum_{i=1}^t(r-iw) $. Since $w=\tfrac{r}{t+1}  $, for SG decoding with $d$ helper SBs we get
\begin{align*}
\chi_{SG}
&=d\left (lr-\frac{wt(t+1)}{2}\right )+(l+t)r-wt(t+1)\\
&=d\left (lr-\frac{rt}{2}\right )+lr,
\end{align*}
and the complexity reduction is given by (see next sub-section for a numerical example)
\begin{align}\label{Eq:complex gain}
1-\frac{\chi_{SG}}{\chi_{G}}=1-\frac{d(l-\tfrac{t}{2})+l}{Ml}.
\end{align}

\subsection{SG Decoding Thresholds}
\label{Sub: SG thresholds} 
Motivated by \eqref{Eq:complex gain}, we now study the thresholds of SG decoding. We define $\epsilon^*_{SG}(m,d)$ as the maximum over $ \epsilon\in[0,1] $ such that SG decoding successfully decodes a target SB $ m\in\{1,\ldots,M\} $ using $ d $ neighbor helper SBs (see Figure~\ref{Fig:SG decoding} for  $ d=4 $). Using the SG density-evolution equations in \eqref{Eq:SG_DE_VN}--\eqref{Eq:SG dout}, we can easily calculate this threshold.
Figure~\ref{Fig:SGTh} illustrates SG thresholds $\epsilon^*_{SG}(m=6,d)$ of the $ (5,12,t \in \{1,2,3\}) $ SC-LDPCL protographs with $ M=11 $ SBs \er{and a classical SC-LDPC protograph with the same length and degrees but without locality (i.e., $ t=l-1=4 $)} as a function of $ d $. For every $ t\in\{1,2,3\} $, the curve starts ($ d=0 $) from the local threshold of that code $ \epsilon^*_6$ (see Definition~\ref{Def:LC and CC}), steeply increases due to adjacent helpers ($d=2$), and ends ($ d=10 $) close to the global threshold of that code $ \epsilon^*_G $ due to terminating SBs ($d=M-1=10$). \er{For the no-locality code, the local threshold ($ d=0 $) is zero, as Theorem~\ref{Th:no local} predicts, and the SG thresholds for $ d>0 $ are significantly lower than those of the locality-enabled constructions ($ t\in\{1,2,3\} $). Moreover, in contrast to the locality-enabled constructions, starting from termination in the no-locality code does not increase the threshold.}

In view of Figure~\ref{Fig:SGTh}, it appears that the main advantage \er{in SG decoding of SC-LDPCL codes} is due to immediate adjacent helpers where $ d=2 $, and due to the end-point helpers. While this is true for the fixed-erasure-probability channel assumed here, in the next sub-section we show that under channels with variability, it is beneficial to use intermediate values of $ d $ as well.

Note that the global threshold of the $ (5,12,t=3) $ protograph with $ M=11 $ is $ \epsilon^*_G = 0.375 $ while $\epsilon^*_{SG}(5,10)=0.361$. Thus the threshold reduction is only $ 3.7\% $. On the other hand, \eqref{Eq:complex gain} implies that for $ t=3 $ the complexity reduction for $ d=10 $ equals $ 1- \tfrac1{55}\left(10\left(5-\tfrac{3}{2}\right)+5\right) =27\%$; hence we see a substantial decrease in complexity with only a small loss in threshold. 

\begin{figure}
	\begin{center}
		\definecolor{mycolor1}{rgb}{0.00000,0.44700,0.74100}%
\definecolor{mycolor2}{rgb}{0.85000,0.32500,0.09800}%
\definecolor{mycolor3}{rgb}{0.92900,0.69400,0.12500}%
\begin{tikzpicture}

\begin{axis}[%
width=3.5in,
height=1.5in,
at={(0,0)},
scale only axis,
xmin=0,
xmax=10,
xlabel style={font=\color{white!15!black}},
xlabel={$d$},
ymin=0,
ymax=0.374908447265625,
ylabel style={font=\color{white!15!black}},
ylabel={$\epsilon^*_{SG}(6,d)$},
axis background/.style={fill=white},
title={$(5,12)$ SC-LDPCL SG Thresholds},
xmajorgrids,
ymajorgrids,
legend columns=2, 
legend style={
	at={(0.99,0.01)}, 
	anchor=south east, 
	legend cell align=left, 
	align=left, 
	draw=white!15!black,
	/tikz/column 2/.style={
		column sep=5pt,
	},
}
]
\addplot [color=blue, line width=2pt,mark=triangle]
  table[row sep=crcr]{%
0	0.257081004460348\\
2	0.279988628191835\\
4	0.282953788359113\\
6	0.283531031105678\\
8	0.283662608496438\\
10	0.306574901103081\\
};
\addlegendentry{$t=1$}

\addplot [color=red, line width=2pt,mark=o,dashed,mark options={solid}]
  table[row sep=crcr]{%
0	0.210475416273336\\
2	0.26789409572356\\
4	0.277937704222\\
6	0.28093527694253\\
8	0.282060733764965\\
10	0.330502105460619\\
};
\addlegendentry{$t=2$}

\addplot [color=green, line width=2pt,mark=square,dotted,mark options={solid}]
  table[row sep=crcr]{%
0	0.0909545545459242\\
2	0.245072798036443\\
4	0.27116000039589\\
6	0.28031085826674\\
8	0.284697269127063\\
10	0.360336453513012\\
};
\addlegendentry{$t=3$}

\addplot [color=black, line width=1pt,double,mark options={solid}]
table[row sep=crcr]{%
	0	3.00407409667969e-05\\
	2	0.0644469261169434\\
	4	0.110935688018799\\
	6	0.120707035064697\\
	8	0.12260103225708\\
	10	0.122977733612061\\
};
\addlegendentry{No locality}

\end{axis}
\end{tikzpicture}%
		\caption{\label{Fig:SGTh}  $\epsilon^*_{SG}(m = 6,d)$ for $ (5,12,t\in\{1,2,3\}) $ protographs with $ M = 11 $ SBs.}
	\end{center}
\end{figure}
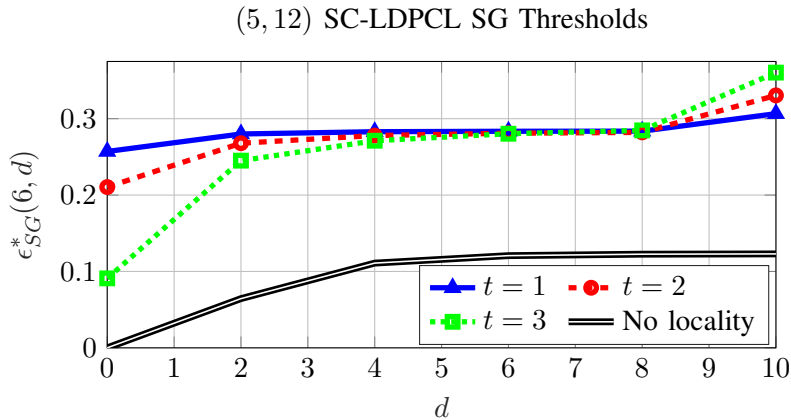  

\begin{example}\label{Ex: SG sim}
	Figure~\ref{Fig:SG512_BEC} shows simulation results for semi-global decoding over the BEC. The plots compare $ (5,12,t) $ SC-LDPCL codes ($ t\in\{1,2,3\} $) with five decoding modes: local, semi-global with $ d=2,8,10 $, and global decoding. As seen in the plots, SG decoding with $ d=10 $ helpers performs very close to global decoding. Further, for low values of $ d $ the $ t=1 $ code is superior while for high values of $ d $, the  $ t=3 $ is superior, as predicted by our threshold calculations. Finally, the main advantage in SG decoding comes from the adjacent helpers (i.e., most significant improvement when switching from local decoding to SG with $ d=2 $), and from termination sub-blocks ($ d=10 $ helpers). 
\end{example}

\begin{figure}
	\input{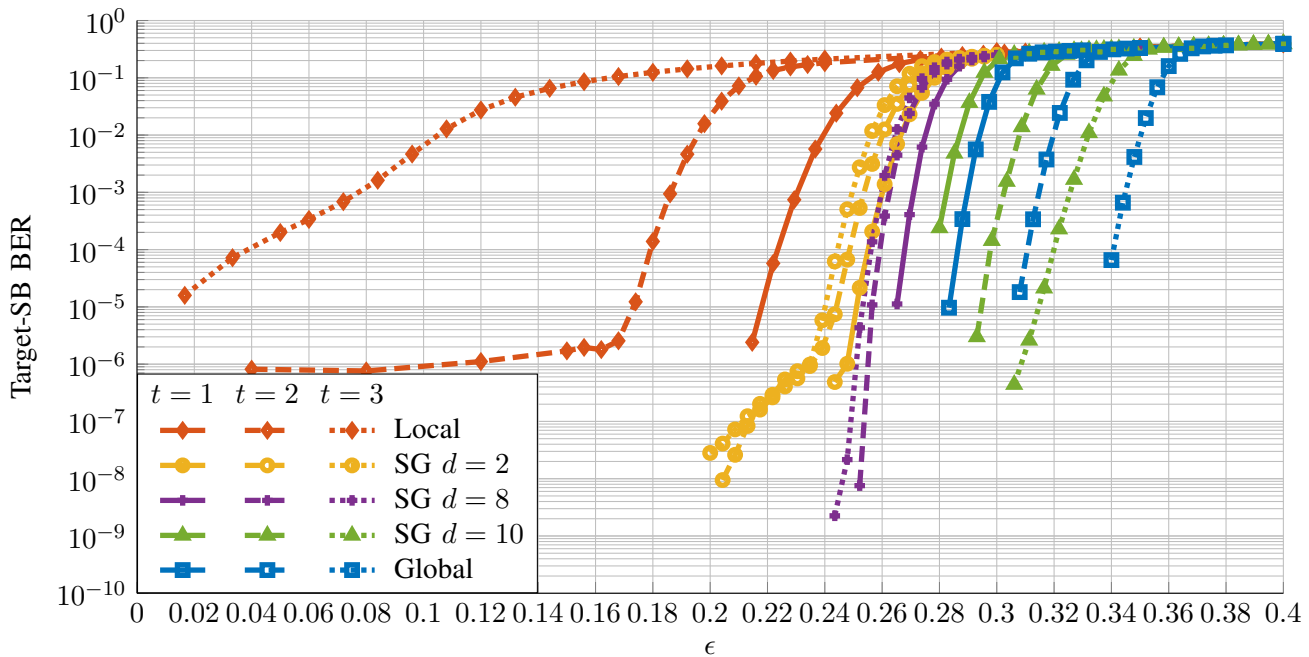}
	\caption{\label{Fig:SG512_BEC} BEC SG-decoding performance of three $ (5,12,t) $ SC-LDPCL codes \er{(dotted, dashed and solid for $ t=1,2$ and $3 $, respectively)} with $ M=11 $ sub-blocks corresponding to Figure~\ref{Fig:SGTh}; the target is sub-block number six. The plots refer to five decoding modes: local \er{(orange)}, SG with $ d=2 $ helpers \er{(yellow)}, SG with $ d=8 $ \er{(purple)}, SG with $ d=10 $ \er{(green)}, and global decoding \er{(blue)}.}
\end{figure}

\subsection{Analysis Over the Sub-Block Varying BEC}\label{Sub:Channel}
We now consider a channel model originally proposed in \cite{McEliece84}, in which the channel parameter varies between sub-blocks. This channel fits many data-storage systems, where bits of the same sub-block (e.g., on the same physical memory page) suffer from a certain noise level, but across sub-blocks the noise levels may vary considerably due to differences in operating or manufacturing conditions \er{such as cell wear, temperature, etc. \cite{TaraUchi16,ShaAl20}}.
Let $M\in \mathbb{N}$ be the number of SBs in the code, let $E_1,E_2,\ldots,E_M$ be i.i.d. random variables taking values in $[0,1]$, and let $F$ be the \er{cumulative distribution function (CDF)} of $E_m$, i.e., for every $m\in \{1,2,\ldots,M\}$ and $x\in [0,1]$, $F(x) = \Pr\left( E_m\leq x\right)$. In the channel introduced in \cite{McEliece84}, all of the bits of SB $m$ are transmitted over the same channel, which in our case is the $BEC(E_m)$; in other words, first $E_m$ is realized, and then the bits of SB $m$ pass through the  $BEC(E_m)$. The standard BEC, where the erasure probability $ \epsilon $ is constant, is a special case of the sub-block varying BEC where $F(x)$ is the step function at $x=\epsilon$.

Semi-global decoding is highly motivated by this channel, since even if the target SB suffers a high erasure rate, and local decoding fails, potentially the helpers have low erasure rates. Subsequently, the decoded helpers can send sufficient information toward the target SB in order to successfully decode it at the target phase. 
\er{Note that the decoder has no information about the channel state, i.e., the instantaneous channel parameters are not known during decoding. In case that this side information is available at the decoder, then other semi-global scheduling schemes (i.e., other than the symmetric scheme described in Section~\ref{Sec:SG}) could result in better performance.}

\er{
	Sliding-window decoders and multi-dimensional spatially coupled (MD-SC) codes were suggested for channels with memory (e.g. bursty and Gilbert-Elliott channels) and parallel channels (see \cite{IyenPapa12,Ohashi13,ScmalMahd14}). Our work differs from these previous works since 1) semi-global decoding differs from window decoding (as explained above), 2) in contrast to MD-SC codes where entire SC chains are connected, SC-LDPCL codes connect short locally decodable sub-blocks, and 3) the channel model we consider differs from those previously considered and better fits the setup where sub-blocks are mapped to distinct physical storage units. All of these distinctions lead to a new analysis which we perform in the following.
}
\begin{definition}\label{Def:pd}
For every even $ j$ and $\underline\delta_1,\underline\delta_2\in[0,1]^t $, we define $ p_j\left (\underline\delta_1,\underline\delta_2\right )$ as the asymptotic (as the lifting parameter tends to infinity) SG-decoding success-probability to decode a target SB $ m $ with $ d=j $ helpers: $ \tfrac{j}{2} $ helper SBs to the right (i.e. larger indices than the target) and $ \tfrac{j}{2} $ helpers to the left, where $ \underline\delta_1 $ and $ \underline\delta_2 $ are the $ t $ erasure probabilities incoming from the SB left to SB $m-\tfrac{j}{2} $ and from the SB right to SB $ m+\tfrac{j}{2}  $, respectively. 
\end{definition}

In general, $ p_j(\underline\delta_1,\underline\delta_2) $ is a function of both the protograph and the channel-parameter's CDF $ F(\cdot) $. Regardless of the protograph or channel, we expect $ p_j(\underline\delta_1,\underline\delta_2) $ to be monotonically non-decreasing with $ j $.
In the following analysis\er{,} we assume that the protograph is large enough, such that no helper SB is among the first or last SBs (i.e., no termination). 
Given an even number of helper SBs $ d $, our goal is to evaluate $ p_d\left(\underline 1,\underline 1 \right )$, and the intermediate values of $ p_j\left (\underline\delta_1,\underline\delta_2\right ) $, with $ j\in\{0,2,\ldots,d\} $ and $\underline\delta_1,\underline\delta_2\in [0,1]$, will help us track probabilities along the SG process.

\begin{definition}\label{Def:EpsSG}
Let $ \underline{\delta}_1,\underline\delta_2\in [0,1]^t $. We define:
	\begin{enumerate}
		\item  $ \epsilon^*\left (\underline\delta_1,\underline\delta_2\right )$ as the target's threshold given that the incoming erasure probabilities are $ \underline{\delta}_{L}=\underline\delta_1 $, and $ \underline{\delta}_{R}=\underline\delta_2 $ (see Figure~\ref{Fig:SG graph}(target))
		\item  $ \Delta\colon [0,1]\times[0,1]^t\to[0,1]^t $ as the helper function that calculates the outgoing erasure probabilities $\underline \delta_{O}$ given a SB erasure probability $ \epsilon$ and incoming erasure probabilities $\underline \delta_{I}$ (see Figure~\ref{Fig:SG graph}(helper)), i.e., 
		\begin{align*}
		\left (\delta_{O,1},\ldots,\delta_{O,t}\right )=\Delta\left (\epsilon,\delta_{I,1},\ldots,\delta_{I,t}\right ).
		\end{align*}
		\item  $ \Delta_k\colon [0,1]^k\times[0,1]^t\to[0,1]^t $ as the  recursive function, $ k\in\mathbb{N}^+ $:
		\begin{subequations}
		\begin{align}
		\label{Eq:Delta1}
		&\Delta_1\left( \epsilon,\underline \delta_I\right )=\Delta\left (\epsilon,\underline \delta_I\right )\\
		\label{Eq:Deltak}
		&\Delta_j\left ( \epsilon_1,\ldots,\epsilon_j,\underline \delta_I\right ) = \Delta_{j-1}\left (\epsilon_1,\ldots,\epsilon_{j-1},\Delta\left (\epsilon_j,\underline \delta_I\right ) \right ),\quad j\geq  2.
		\end{align}	
		\end{subequations}
	\end{enumerate}
\end{definition}

\begin{remark}
The functions $ \epsilon^*(\cdot) $, and $ \Delta(\cdot,\cdot)$ from Definition~\ref{Def:EpsSG} are deterministic functions that depend on the semi-global graph $ \mathcal{G}_{SG} $, although this dependence is not written explicitly. 
\end{remark}

\begin{remark}
If we remove the assumption in \eqref{Eq:w assumption}, then we will have to replace $ \Delta $ in items 2)+3) of Definition~\ref{Def:EpsSG} with two functions: right-to-left and left-to-right. Since we assume symmetry, then these two functions coincide and we refer to them both as $\Delta(\cdot,\cdot) $.
\end{remark}

\begin{theorem}\label{Th:Recursive SG}
	For every varying-erasure channel, and every even $ j\geq0 $,  
	\begin{align}\label{Eq:Recursive SG}
	\begin{split}
	&p_0(\underline\delta_1,\underline\delta_2)=\Pr\left (E<\epsilon^*(\underline\delta_1,\underline\delta_2)\right ),\\
	&p_j\left (\underline\delta_1,\underline\delta_2\right )=\mathbb{E}\left [  p_{j-2}\left (\Delta \left (E_1, \underline\delta_1\right ),\Delta \left (E_2, \underline\delta_2\right )\right )\right ],\;j\geq 2.
	\end{split}
	\end{align}
	where $ E,E_1,E_2 $ are i.i.d. random variables representing the channel parameter, and $ \mathbb{E}[\cdot] $ is the expectation of its argument over the choices of $E_1,E_2$.
\end{theorem}

\begin{proof}
See Appendix~\ref{App:Recursive SG}.
\end{proof}

Theorem~\ref{Th:Recursive SG} provides an exact recursive expression for $ p_j\left (\underline\delta_1,\underline\delta_2\right ) $. However, this calculation depends on the stochastic argument $E_i$ of $ \Delta \left (E_i, \cdot \right ) $,  and in some cases, such as if the channel parameter is a continuous random variable, it is hard to evaluate the expectation in \eqref{Eq:Recursive SG}. To go around this difficulty, we derive a provable lower bound on $ p_j\left (\underline\delta_1,\underline\delta_2\right ) $ by quantizing the erasure-rate domain.

\begin{theorem}\label{Th:SG LB}
	Let $F(\cdot)$ be the CDF of a varying BEC channel, let $ \underline \delta_1,\underline \delta_2\in[0,1]^t $ and  $K\in \mathbb{N}$, and let $0=e_{0}<e_{1}<e_{2}<\ldots<e_{K}=1$ be a partition of $ [0,1] $. For every $j$ even indices $\underline{i}=\left (i_{-j/2},\ldots,i_{-1},i_{1},\ldots,i_{j/2}\right )\in \{1,\ldots,K\}^{j}$, let 
	\begin{align}\label{Eq:yi}
	 y_{\underline{i}}(\underline \delta_1,\underline \delta_2) \triangleq  \epsilon^* \left(
	\Delta_{j/2}\left( e_{i_{-1}},\ldots,e_{i_{-j/2}},\underline \delta_1\right ) ,
	\Delta_{j/2}\left( e_{i_1},\ldots,e_{i_j/2},\underline \delta_2\right ) \right ).
	\end{align}
	Then,
	\begin{align}\label{Eq:SG LB}
	 p_{j}(\underline \delta_1,\underline \delta_2) \geq     \sum_{\underline{i}\in \{1,2,\ldots,K\}^{j}}   F\left (y_{\underline{i}}(\underline \delta_1,\underline \delta_2)\right )  \prod_{k=-j/2\atop k\neq 0}^{j/2}   \left [ F(e_{i_k}) - F(e_{i_k-1})\right ].
	\end{align}
\end{theorem}

\begin{proof}
See Appendix~\ref{App:SG LB}.
\end{proof}

Note that unlike \eqref{Eq:Recursive SG}, evaluating the right-hand side of \eqref{Eq:SG LB} only uses deterministic arguments in closed-form ($F$) and recursive ($ y_{\underline i}$) functions.

\begin{remark}
For any given $d$, increasing the parameter $K$ tightens the bound in \eqref{Eq:SG LB}. On the other hand, increasing $ K $ increases the calculation complexity (finer quantization). Through the parameter $ K $, one can control this tightness-vs.-complexity trade-off.
\end{remark}

\begin{remark}	\label{Remark:choice LB} 
Although the bound in Theorem~\ref{Th:SG LB} holds for every choice of $K$ and $\left \{e_{i}\right \}_{i=1}^{K-1}$, it is preferable to have at least $\epsilon_L\triangleq \epsilon^*(\underline 1,\underline 1)$ and $\epsilon_S\triangleq\epsilon^*(\underline 1, \underline 0)$ as points of calculation since they capture success in the extreme case where previously decoded helpers completely fail ($\epsilon_L$: helpers from both sides fail, $\epsilon_S$: helpers from one side fail). For example, one may set 
\begin{align} \label{Eq:choice LB}
\underline{e} = \left (0,\epsilon_L,\epsilon_L+\xi_L,\ldots,\epsilon_S,\epsilon_S+\xi_S,\ldots,1\right ),
\end{align}
where $\xi_L=\lfloor \tfrac{K}{2} \rfloor(\epsilon_S-\epsilon_L)$ and $\xi_S=\left \lfloor \tfrac{K}{2} \right \rfloor(1-\epsilon_S)$.
\end{remark}

In order to further reduce calculation complexity, we state the next lower bound. Similar to the definition of $p_j(\underline \delta_1,\underline \delta_2)$, we denote by $ \hat{p}_j(\underline \delta) $ the SG success probability when all $ j $ helper SBs are \emph{at one side of the target}, either all left or all right of it, given that the farthest helper from the target has input erasure-probability vector $ \underline \delta $.
\begin{proposition}\label{Prop:SG_bounds}
For every even $ j\geq 2 $,
\begin{align}
	\label{Eq:SG_rec1}
	p_j(\underline 1,\underline 1)\hspace*{3mm}&\geq P_L^2 \cdot p_{j-2}\left(\underline 0,\underline 0\right)+2P_L(1-P_L)p_{j-2}\left(\underline 1,\underline 0\right) +\left (1-P_L\right )^2 p_{j-2}\left(\underline 1,\underline 1\right),\\
	\label{Eq:SG_rec2}
	\begin{split}
	p_j(\underline\delta_1,\underline\delta_2)&\geq
	P_L\cdot\left (1-\hat{p}_{\frac{j}{2}-1}(\underline\delta_1) \right )\cdot\left(1-\hat{p}_{\frac{j}{2}-1}(\underline\delta_2) \right )\\
	&+P_S\cdot\left (\hat{p}_{\frac{j}{2}-1}(\underline\delta_2)\left (1-\hat{p}_{\frac{j}{2}-1}(\underline\delta_1) \right ) +\hat{p}_{\frac{j}{2}-1}(\underline\delta_1)\left (1-\hat{p}_{\frac{j}{2}-1}(\underline\delta_2) \right ) \right )\\
	&+P_D \cdot\hat{p}_{\frac{j}{2}-1}(\underline\delta_1)\cdot\hat{p}_{\frac{j}{2}-1}(\underline\delta_2) 
	\end{split}
\end{align}
where 
\begin{align}\label{Eq:extreme th}
	P_L\triangleq \Pr\left (E<\epsilon^*(\underline 1,\underline 1)\right ),\quad 
	P_S\triangleq \Pr\left (E<\epsilon^*(\underline 1,\underline 0)\right ),\quad
	P_{D}\triangleq\Pr\left (E<\epsilon^*(\underline 0,\underline 0)\right ).
\end{align}
\end{proposition}

\begin{proof}
	See Appendix~\ref{App:SG bounds}.
\end{proof}
In view of \eqref{Eq:Recursive SG} and \eqref{Eq:extreme th}, we have for $ j=0 $,
\begin{align}\label{Eq:SG d=0}
p_0(\underline 1,\underline 1)=P_L,\quad p_0(\underline 1,\underline 0)=p_0(\underline 0,\underline 1)=P_S,\quad p_0(\underline 0,\underline 0)=P_{D}.
\end{align}	
Proposition~\ref{Prop:SG_bounds} leads to a simple way to lower bound  $ p_j(\underline 1,\underline 1) $: first calculate the exact values for $ j=0 $ in \eqref{Eq:SG d=0}, and then use the recursive bounds in \eqref{Eq:SG_rec1}--\eqref{Eq:SG_rec2} to lower bound  $ p_j(\underline 1,\underline 1) $. For example
for $ j=2 $ we get from \eqref{Eq:SG_rec2}
\begin{align*}
p_2(\underline 1,\underline 1)
&\geq P_L^2 \cdot p_{0}\left(\underline 0,\underline 0\right)+2P_L(1-P_L)p_{0}\left(\underline 1,\underline 0\right) +\left (1-P_L\right )^2 p_{0}\left(\underline 1,\underline 1\right)\\
&= P_L^2 P_{D}+2P_L(1-P_L)P_S +\left (1-P_L\right )^2 P_L\,.
\end{align*}

Lower bounds on $ \hat p_j(\cdot) $ can be derived similarly to the bound in Proposition~\ref{Prop:SG_bounds}; we omit the details here. To show an application of the bounds in Theorem~\ref{Th:SG LB} and Proposition~\ref{Prop:SG_bounds}, we next use them to evaluate the balanced semi-global strategy proposed in Section~\ref{Sec:SG} (evaluated by $ p_j(\underline1,\underline 1) $) in comparison to the one-sided semi-global strategy (evaluated by $ \hat{p}_j(\underline 1) $). 

\begin{example}\label{Ex:SG}
	Figure~\ref{Fig:SG} compares between the success probability of SG-decoding when applying the balanced strategy ($ p_j(\underline 1,\underline 1)$) and when applying the one-sided strategy ($ \hat{p}_j(\underline 1) $). The plots refer to $ (l = 5,r = 12,t\in\{1,2,3\}) $ SC-LDPCL protographs over the varying erasure channel $ BEC(E),\;E\sim \mathrm{Unif}[0,0.4] $ \er{(the uniform distribution is given as a concrete example for computing the bounds; our results in Theorem~\ref{Th:SG LB} and Proposition~\ref{Prop:SG_bounds} are derived for any distribution).} As seen in Figure~\ref{Fig:SG}, the balanced strategy (solid-blue curves) performs better than the one-sided strategy (dotted-black curves) for every value of $ t=1,2,3 $ and every $ j\in\{2,4,6,8,10\}$. The bounds are computed according to Theorem~\ref{Th:SG LB} and Proposition~\ref{Prop:SG_bounds}. We used Theorem~\ref{Th:SG LB} to get a lower bound for $ j=2$, with $ K=40 $ and the partition in \eqref{Eq:choice LB}; for the higher values of $j$ we used Proposition~\ref{Prop:SG_bounds}.
	
	Figure~\ref{Fig:SG} also exemplifies the trade-off between locality and coupling in SC-LDPCL protographs (as seen in Figure~\ref{Fig:SGTh} for the standard erasure channel). If $ j=0 $ (local decoding), it is preferable to use the $ t=1 $ protograph which is highly localized. However, if $ j\geq 6 $, the $ t=3 $ protograph, which is strongly coupled, is superior. In the range $ j\in\{2,4\} $, the $ t=2 $ protograph is superior.
\end{example}

\begin{figure}
	\begin{center}
	\begin{tikzpicture}

\begin{axis}[%
width=4.5in,
height=2in,
at={(0,0)},
scale only axis,
xmin=0,
xmax=10,
xlabel style={font=\color{white!15!black}},
xlabel={\small number of helpers $ d $},
ymin=0.2,
ymax=1,
ylabel style={font=\color{white!15!black}},
ylabel={\small Lower bounds on $p_d(\underline 1,\underline 1)$, $\hat p_d(\underline 1)$},
axis background/.style={fill=white},
title={\small $(5,12,t)$ SC-LDPCL over BEC($E$), $E\sim \mathrm{Unif}[0,0.4]$},
xmajorgrids,
ymajorgrids,
legend style={at={(0.99,0.01)}, anchor=south east, legend cell align=left, align=left, draw=white!15!black,legend columns=2}
]
\addplot [color=black,dotted, line width=1pt,mark=triangle,mark options={solid}]
table[row sep=crcr]{%
	0	0.642702511150871\\
	1	0.728120957774175\\
	2	0.737010491521478\\
	3	0.737877735451586\\
	4	0.737877735451586\\
	5	0.737877735451586\\
	6	0.737877735451586\\
	7	0.737877735451586\\
	8	0.737877735451586\\
	9	0.737877735451586\\
	10	0.737877735451586\\
	11	0.737877735451586\\
	12	0.737877735451586\\
	13	0.737877735451586\\
	14	0.737877735451586\\
	15	0.737877735451586\\
	16	0.737877735451586\\
	17	0.737877735451586\\
	18	0.737877735451586\\
	19	0.737877735451586\\
	20	0.737877735451586\\
	21	0.737877735451586\\
	22	0.737877735451586\\
	23	0.737877735451586\\
	24	0.737877735451586\\
	25	0.737877735451586\\
	26	0.737877735451586\\
	27	0.737877735451586\\
	28	0.737877735451586\\
	29	0.737877735451586\\
};
\addlegendentry{$\hat{p}_d(\underline 1),\;t=1$}

\addplot [color=blue, line width=1pt,mark=triangle]
table[row sep=crcr]{%
	0	0.642702511150871\\
	2	0.798808013083043\\
	4	0.81074073232627\\
	6	0.81074073232627\\
	8	0.81074073232627\\
	10	0.81074073232627\\
	12	0.81074073232627\\
	14	0.81074073232627\\
	16	0.81074073232627\\
	18	0.81074073232627\\
	20	0.81074073232627\\
	22	0.81074073232627\\
	24	0.81074073232627\\
	26	0.81074073232627\\
	28	0.81074073232627\\
};
\addlegendentry{${p}_d(\underline 1,\underline 1),\;t=1$}

\addplot [color=black, line width=1pt,mark=o,dotted,mark options={solid}]
table[row sep=crcr]{%
	0	0.52618854068334\\
	1	0.713618151602248\\
	2	0.759946147453353\\
	3	0.769876553035749\\
	4	0.769876553035749\\
	5	0.769876553035749\\
	6	0.769876553035749\\
	7	0.769876553035749\\
	8	0.769876553035749\\
	9	0.769876553035749\\
	10	0.769876553035749\\
	11	0.769876553035749\\
	12	0.769876553035749\\
	13	0.769876553035749\\
	14	0.769876553035749\\
	15	0.769876553035749\\
	16	0.769876553035749\\
	17	0.769876553035749\\
	18	0.769876553035749\\
	19	0.769876553035749\\
	20	0.769876553035749\\
	21	0.769876553035749\\
	22	0.769876553035749\\
	23	0.769876553035749\\
	24	0.769876553035749\\
	25	0.769876553035749\\
	26	0.769876553035749\\
	27	0.769876553035749\\
	28	0.769876553035749\\
	29	0.769876553035749\\
};
\addlegendentry{$\hat{p}_d(\underline 1),\;t=2$}

\addplot [color=blue,  line width=1pt,mark=o]
table[row sep=crcr]{%
	0	0.52618854068334\\
	2	0.826038894725608\\
	4	0.882035317982811\\
	6	0.892086887630882\\
	8	0.895675862640329\\
	10	0.895675862640329\\
	12	0.895675862640329\\
	14	0.895675862640329\\
	16	0.895675862640329\\
	18	0.895675862640329\\
	20	0.895675862640329\\
	22	0.895675862640329\\
	24	0.895675862640329\\
	26	0.895675862640329\\
	28	0.895675862640329\\
};
\addlegendentry{${p}_d(\underline 1,\underline 1),\;t=2$}

\addplot [color=black, line width=1pt,mark=square,dotted,mark options={solid}]
  table[row sep=crcr]{%
0	0.227386386364811\\
1	0.563988675456112\\
2	0.730800221219561\\
3	0.760715996173865\\
4	0.771494154103133\\
5	0.771514797265216\\
6	0.771514797265216\\
7	0.771514797265216\\
8	0.771514797265216\\
9	0.771514797265216\\
10	0.771514797265216\\
11	0.771514797265216\\
12	0.771514797265216\\
13	0.771514797265216\\
14	0.771514797265216\\
15	0.771514797265216\\
16	0.771514797265216\\
17	0.771514797265216\\
18	0.771514797265216\\
19	0.771514797265216\\
20	0.771514797265216\\
21	0.771514797265216\\
22	0.771514797265216\\
23	0.771514797265216\\
24	0.771514797265216\\
25	0.771514797265216\\
26	0.771514797265216\\
27	0.771514797265216\\
28	0.771514797265216\\
29	0.771514797265216\\
};
\addlegendentry{$\hat{p}_d(\underline 1),\;t=3$}

\addplot [color=blue, line width=1pt,mark=square]
  table[row sep=crcr]{%
0	0.227386386364811\\
2	0.739894807886517\\
4	0.830402560545825\\
6	0.904979380475769\\
8	0.91964935082562\\
10	0.924682842402471\\
12	0.924692354929364\\
14	0.924692354929364\\
16	0.924692354929364\\
18	0.924692354929364\\
20	0.924692354929364\\
22	0.924692354929364\\
24	0.924692354929364\\
26	0.924692354929364\\
28	0.924692354929364\\
};
\addlegendentry{${p}_d(\underline 1,\underline 1),\;t=3$}

\end{axis}
\end{tikzpicture}%
	\caption{\label{Fig:SG}Lower bounds on $ p_j(\underline 1,\underline 1) $ and $ \hat{p}_j(\underline 1) $ for the $ (5,12,t\in\{1,2,3\}) $ SC-LDPCL protographs over the  $ BEC(E),\;E\sim \mathrm{Unif}[0,0.4] $.}
	\end{center}
\end{figure}
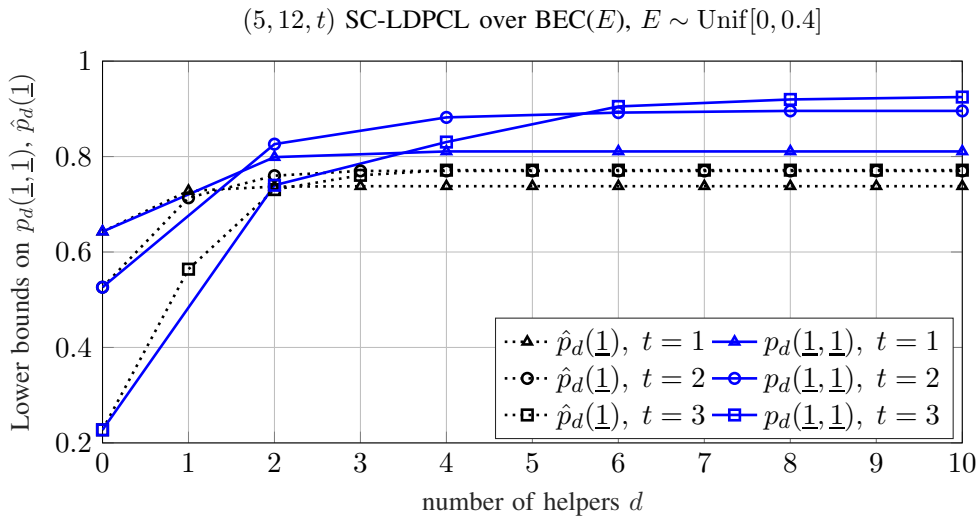

\section{Generalized Constructions}
\label{Sec:Gen Construct}

The family of SC-LDPCL protographs introduced in Section~\ref{Sec:SC-LDPCL Codes} and analyzed in Section~\ref{Sec:SG} share a common property in which sub-blocks are connected only to their adjacent neighbors. This follows from Construction~\ref{Construct: SC-LDPCL}, which uses memory $ T=1 $, i.e., the base matrix $B$ is decomposed into two matrices: $ B_0$ and $ B_1 $. In general, $ T $ can be greater than $1$. In this section, we present a generalization of Construction~\ref{Construct: SC-LDPCL} to  $T\geq 1$. This generalization enriches the family of SC-LDPCL codes and enables additional SG-decoding strategies. 

\er{
	In addition to generalizing the construction to larger memory parameters, we extend in this section SC-LDPCL codes beyond the $1$-dimensional chain of the classical SC-LDPC codes. The structure of the resulting codes is similar to existing multi-dimensional SC-LDPC codes  \cite{LiuLi15,OlmMitch17,TruMitch19,DolEsf20,ScmalMahd14,Ohashi13}, but as in the $1$-dimensional case, our codes enable the extra feature of decoding a requested target sub-block by accessing only a small number of sub-blocks around it in the array.
}

\subsection{Generalized Construction}

\begin{construction}\label{Construct: Generalized}
Let $ 3 \leq l < r $ and $ T\geq 1 $ be integers, and let $ t\in\{1,2,\ldots,l-2\} $. Let $ P\in \{0,\ldots,T\}^{l\times r} $ such that
\begin{align}\label{Eq: Partition 1}
P(i,j)=0,\quad \forall\; t<i\leq l,\; 1\leq j \leq r\,.
\end{align}	
Set $ l\times r $ matrices $ B_0,\ldots,B_T $ such that for every $ \tau\in \{0,\ldots,T\} $, $ B_\tau(i,j)=1 $ if $ P(i,j)=\tau $, and  $ B_\tau(i,j)=0 $ otherwise.
Construct a coupled protograph $ H $ by diagonally placing $ B_0,\ldots,B_T $ in $ H $ as in Figure~\ref{Fig:SC36}(b).
\end{construction}

$P(i,j)$ is an index matrix used to specify the graph coupling. All-zero rows in $P$ correspond to local checks (LCs) in the constructed protograph, and mixed rows correspond to coupling checks (CCs).
Note that \eqref{Eq: Partition 1} assures sub-block access: the local graph is an $ (l-t,r) $-regular graph, with $ l-t\geq 2 $. 
We next show some examples of protograph classes constructed by Construction~\ref{Construct: Generalized}; the classes differ in their inter-sub-block connections.

\begin{example}\label{Ex:Gen const non-hyper}
Let $4 \leq l,\; T\in\{2,\ldots,l-2\},\; t=T $, and $ w=\left \lfloor\tfrac{r}{t+1}\right \rfloor $. Set
\er{
	\begin{align*}
	P(i,j) = \left\{ 
	\begin{array}{ll}
	0\,, & t<i\leq l\,,\\
	0\,, & 1\leq i\leq t,\; 1\leq j\leq iw\,,\\
	i\,, & 1\leq i\leq t,\; iw< j\leq r\,.
	\end{array}\right.
	\end{align*}
}
For example, for $ l=4,\;r=8,\;t=T=2$ we have
\begin{align}\label{Eq:partition mat non-hyper}
P=\begin{pmatrix}
	0	&	0	&	1	&	1	&	1	&	1	&	1	&	1	\\
	0	&	0	&	0	&	0	&	2	&	2	&	2	&	2	\\
	0	&	0	&	0	&	0	&	0	&	0	&	0	&	0	\\
	0	&	0	&	0	&	0	&	0	&	0	&	0	&	0	
\end{pmatrix}.
\end{align} 
\end{example}

In the above example, every coupling check (CC) connects two sub-blocks: the current sub-block, represented by $0$ entries in CC rows, and the $\tau$-th sub-block away, represented by the entries $\tau>0$.

We can consider another partition, in which coupling checks connect more than two sub-blocks. As we show later, this choice can lead to better global decoding thresholds. 

\begin{example}\label{Ex:Gen const hyper}
For $ l,r,T $ and $ t $ as in Construction~\ref{Construct: Generalized}, let $ w=\left \lfloor\tfrac{r}{t+1}\right \rfloor $ and $ q_i = \left \lfloor\tfrac{r-iw}{T}\right \rfloor  $, $ 1\leq i \leq t $. Set 
\er{
	\begin{align*}
	P(i,j) = \left\{ 
	\begin{array}{ll}
	0\,, & t<i\leq l\,,\\
	0\,, & 1\leq i\leq t,\; 1\leq j\leq iw\,,\\
	\tau\,, & 1\leq i\leq t,\; iw+(\tau-1)q_i< j\leq iw+\tau q_i\,,
	\end{array}\right.\;,
	\end{align*}
}
where $ \tau>0 $. For example, for $ l=4,\,r=8,\,t=T=2 $ we get
\begin{align}\label{Eq:partition mat hyper}
	P=\begin{pmatrix}
	0	&	0	&	1	&	1	&	1	&	2	&	2	&	2\\
	0	&	0	&	0	&	0	&	1	&	1	&	2	&	2\\
	0	&	0	&	0	&	0	&	0	&	0	&	0	&	0\\
	0	&	0	&	0	&	0	&	0	&	0	&	0	&	0
	\end{pmatrix}.
\end{align} 
In \eqref{Eq:partition mat hyper}, rows 1 and 2 in $ P $ specify CCs that connect three sub-blocks (i.e., $ m,\;m+1,m+2 $).  
\end{example}

Further, we present a partition that generates a two-dimensional SC-LDPCL protograph, in which each sub-block is connected to 4 adjacent sub-blocks: to the right, left, up and down. This structure can be valuable for two-dimensional storage topologies. The construction is based on a partition matrix $ P $ that has two CCs: one containing zeros and ones (i.e., connecting SBs horizontally), and the other containing zeros and $ T $'s (vertical connection).

\begin{example}\label{Ex:2D grid}
Let $  l\geq 4 $, let $ r=4K $ for some integer $ K \geq \tfrac{l}{4}$, and let $ T>2=t$. Set
\er{
	\begin{align*}
	P(i,j) = \left\{ 
	\begin{array}{ll}
	0\,, & t<i\leq l\,,\\
	0\,, & i=1,\; \left (1\leq j \leq \tfrac14 r \text{ OR } \tfrac34  r< j \leq r\right )\,,\\
	0\,, & i=2,\; \tfrac12 r <j\leq r\,,\\
	1\,, & i=1,\; \tfrac14 r< j \leq \tfrac34 r\,, \\
	T\,, & i=2,\; 1\leq j\leq \tfrac12 r\,.
	\end{array}\right.\;,
	\end{align*}
}
\end{example}
For example, consider $ l=4,r=8,T=5,t=2,$. Then,
\begin{align}\label{Eq:2-dim partition}
P=\begin{pmatrix}
	0	&	0	&	1	&	1	&	1	&	1	&	0	&	0\\
	0	&	0	&	0	&	0	&	5	&	5	&	5	&	5\\
	0	&	0	&	0	&	0	&	0	&	0	&	0	&	0\\
	0	&	0	&	0	&	0	&	0	&	0	&	0	&	0
\end{pmatrix}.
\end{align} 

Consider four $ (l=4,r=8,t=2) $-regular SC-LDPCL protographs: (A) a protograph constructed by Construction~\ref{Construct: SC-LDPCL}; (B), (C) protographs constructed from \eqref{Eq:partition mat non-hyper} and \eqref{Eq:partition mat hyper}, respectively; (D) a two-dimensional protograph constructed according to \eqref{Eq:2-dim partition}. Assume that all protographs have $ M=25 $ sub-blocks. The local graphs are $ (2,8) $-regular (except for terminating sub-blocks), so the local thresholds of protographs (A)--(D) coincide and equal $ \epsilon_L=0.1429 $. Table~\ref{Tbl:Compare 4} details the design rate and global threshold for each of these protographs. The best rate-threshold trade-off (i.e., smallest gap $ 1-\epsilon_G-R $) is achieved by protograph (C). 
\begin{table}
\caption{\label{Tbl:Compare 4}  Global thresholds of design rates for $ (4,8,2) $ SC-LDPCL protographs.}
\begin{center}
	\begin{tabular}{c|cccc}
			&	(A)			&	(B)		&	(C)		&	(D)	\\
			\hline										
		$R$	&	0.49		&	0.485	&	0.48	&	0.47\\
$\epsilon_G$&	0.4657	&	0.4715	&	0.4864	&	0.4602
	\end{tabular}
\end{center}
\end{table}

\subsection{Semi-Global Decoding}
\label{Sub:Gen SG}

Generalized constructions with $T\geq2$ can also use the semi-global decoding strategy described in Section~\ref{Sec:SG}, with added flexibility in the scheduling of helper decoding.  Since in the general case the protograph’s structure is not a simple chain, the semi-global decoder needs to specify which $d$ helper sub-blocks are decoded, and at what order. 
Further, SG-decoding analysis should be revised since the information flows differently. In view of these observations, we  define the \emph{inter-sub-block graph}. This graph captures the connections between sub-blocks in the coupled protograph while suppressing the intra-sub-block connections (local checks and edges), and the exact node degrees.

\begin{definition}[inter-sub-block graph]\label{Def:ISB graph}
	Let $ \mathcal{G} $ be an SC-LDPCL protograph with $ M $ sub-blocks. The inter-sub-block graph $ \mathcal{G}_{ISB}=(\mathcal{V},\mathcal{E}) $ corresponding to $ \mathcal{G} $ consists of $ |\mathcal{V}|=M $ nodes, each describing a sub-block in $ \mathcal{G} $. An (undirected) edge $ e\in\mathcal{E} $ connects sub-blocks $ \{m_i\}\subset\mathcal{V} $ if there exists a coupling check in $ \mathcal{G} $ that is connected to variable nodes belonging to sub-blocks $ \{m_i\} $.
\end{definition}
\begin{remark}\label{remark:ISB is not bipartite}
	To simplify the presentation, and since for SG decoding we are interested in the inter-sub-block connection, Definition~\ref{Def:ISB graph} absorbs local checks into their sub-block node, and coupling checks into edges connecting sub-block nodes. 
\end{remark}
\begin{remark}\label{remark:ISB is hyper}
	It is possible that a coupling check in the protograph connects more than two sub-blocks, as in \eqref{Eq:partition mat hyper}. In this case, $ \mathcal{G}_{ISB} $ is a \textbf{hyper-graph}, i.e., edges connect sets of nodes. These kind of edges are drawn as split lines in the following graph illustrations.
\end{remark}

\begin{example}\label{Ex:G_ISB}
	Figure~\ref{Fig:G_ISB} illustrates the inter-sub-block graphs for the (A), (B), (C), and (D) protographs listed in Table~\ref{Tbl:Compare 4}. (A) is a simple chain; (B) introduces more memory and connectivity between sub-blocks; (C) is a hyper-graph in which some edges (CCs) connect three nodes; (D) is a grid graph. 
	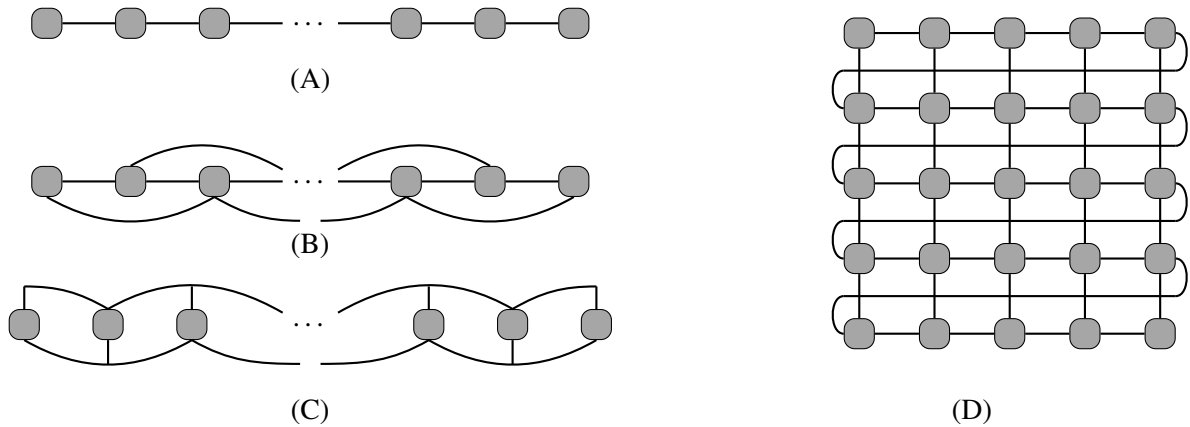
\begin{figure}
		\begin{center}
			\begin{tikzpicture}\label{Tikz:G-ISB}
				
			\tikzstyle{SB}=[rounded corners,draw,fill=gray!70!white,minimum size=4mm]
			\tikzstyle{zero}=[fill,draw,scale=0.6]
			
			% D
			\pgfmathsetmacro{\T}{5}
			\pgfmathtruncatemacro{\Q}{\T - 1}
			\pgfmathsetmacro{\x}{1}
			\pgfmathsetmacro{\y}{2}
			\foreach \t in {1,...,\T}
			{
				\foreach \r in {1,...,\T}
				{
					\node[SB] (sb\t\r) at (\x*\t,\r) {};
				}
			}
			\foreach \t in {1,...,\Q}
			{
				
				\foreach \r in {1,...,\Q}
				{
					\pgfmathtruncatemacro{\k}{\t + 1}
					\pgfmathtruncatemacro{\j}{\r + 1}
					\draw[thick] (sb\t\r)--(sb\k\r);
					\draw[thick] (sb\t\r)--(sb\t\j);
				}
			}
			\foreach \r in {1,...,\Q}
			{
				\pgfmathtruncatemacro{\j}{\r + 1}
				\draw[thick] (sb5\r)--(sb5\j);
			}
			\foreach \t in {1,...,\Q}
			{
				\pgfmathtruncatemacro{\k}{\t + 1}
				\draw[thick] (sb\t5)--(sb\k5);
			}

			\node (a) at ($(sb52.east)!0.5!(sb51.east)$) {};
			\node (b) at ($(sb12.west)!0.5!(sb11.west)$) {};
			\draw[thick] (sb52.east) to[in= 0, out = 0] (a.center);
			\draw[thick] (a.center) -- (b.center);
			\draw[thick] (b.center) to[in= 180, out = 180] (sb11.west);
			
			\node (a) at ($(sb53.east)!0.5!(sb52.east)$) {};
			\node (b) at ($(sb13.west)!0.5!(sb12.west)$) {};
			\draw[thick] (sb53.east) to[in= 0, out = 0] (a.center);
			\draw[thick] (a.center) -- (b.center);
			\draw[thick] (b.center) to[in= 180, out = 180] (sb12.west);
			
			\node (a) at ($(sb54.east)!0.5!(sb53.east)$) {};
			\node (b) at ($(sb14.west)!0.5!(sb13.west)$) {};
			\draw[thick] (sb54.east) to[in= 0, out = 0] (a.center);
			\draw[thick] (a.center) -- (b.center);
			\draw[thick] (b.center) to[in= 180, out = 180] (sb13.west);
			
			\node (a) at ($(sb55.east)!0.5!(sb54.east)$) {};
			\node (b) at ($(sb15.west)!0.5!(sb14.west)$) {};
			\draw[thick] (sb55.east) to[in= 0, out = 0] (a.center);
			\draw[thick] (a.center) -- (b.center);
			\draw[thick] (b.center) to[in= 180, out = 180] (sb14.west);

			\node (dt) at ($(sb21)!0.5!(sb31)$) {};
			\node (d) [below = 6mm of dt] {(D)};
			
			% C
			\node (c) [left = 8cm of d] {(C)};
			\node (sb0) [above=7 mm of c] {\dots};
			\node (sb1) [SB, right=10 mm of sb0] {};
			\foreach \t in {2,...,3}
			{
				\pgfmathtruncatemacro{\k}{\t - 1}
				\node (sb\t) [SB, right=7 mm of sb\k] {};
			}

			\node (sbm1) [SB, left=10 mm of sb0] {};
			\foreach \t in {2,...,3}
			{
				\pgfmathtruncatemacro{\k}{\t - 1}
				\node (sbm\t) [SB, left=7 mm of sbm\k] {};
			}

			%north			
			\draw[thick] (sb0.north west) to[in= 30, out = 150] (sbm2.north); 
			\node (um1) [above=3mm of sbm1] {};
			\draw[thick] (um1)--(sbm1);
			\draw[thick] (sb2.north) to[in= 30, out = 150] (sb0.north east); 
			\node (u1) [above=3mm of sb1] {};
			\draw[thick] (u1)--(sb1);
			\node (u3) [above=3mm of sb3] {};
			\draw[thick] (u3)--(sb3);
			\draw[thick] (u3.south) to[in=30,out=180] (sb2.north);
			\node (um3) [above=3mm of sbm3] {};
			\draw[thick] (um3)--(sbm3);
			\draw[thick] (sbm2.north)to[in=0,out=150] (um3.south) ;
			%south
			
			\node (d0) [below=2.3mm of sb0] {};
			\draw[thick] (sb1.south) to[in= 0, out = 210] (d0.east); 
			\draw[thick] (sbm1.south) to[in= 180, out = 330] (d0.west); 

			\draw[thick] (sb3.south) to[in= 330, out = 210] (sb1.south); 
			\node (d2) [below=3.3mm of sb2] {};
			\draw[thick] (d2)--(sb2);
			\draw[thick] (sbm1.south) to[in= 330, out = 210] (sbm3.south); 
			\node (dm2) [below=3.3mm of sbm2] {};
			\draw[thick] (dm2)--(sbm2);

			%B
			\node (b) [above=16 mm of c] {(B)};
			\node (sb0) [above=4 mm of b] {\dots};
			\foreach \t in {1,...,3}
			{
				\pgfmathtruncatemacro{\k}{\t - 1}
				\node (sb\t) [SB, right=7 mm of sb\k] {};
			}

			\node (sbm1) [SB, left=7 mm of sb0] {};
			\foreach \t in {2,...,3}
			{
				\pgfmathtruncatemacro{\k}{\t - 1}
				\node (sbm\t) [SB, left=7 mm of sbm\k] {};
			}

			%north
			\draw[thick] (sb0.north west) to[in= 30, out = 150] (sbm2.north); 
			\draw[thick] (sb2.north) to[in= 30, out = 150] (sb0.north east);

			%south
			\node (d0) [below=2.3mm of sb0] {};
			\draw[thick] (sb1.south) to[in= 0, out = 210] (d0.east); 
			\draw[thick] (sbm1.south) to[in= 180, out = 330] (d0.west); 
			\draw[thick] (sb3.south) to[in= 330, out = 210] (sb1.south); 
			\draw[thick] (sbm1.south) to[in= 330, out = 210] (sbm3.south);

			%center
			\draw[thick] (sb3)--(sb2)--(sb1)--(sb0)--(sbm1)--(sbm2)--(sbm3); 
			
			%A
			\node (a) [above=16 mm of b] {(A)};
			\node (sb0) [above=3 mm of a] {\dots};
			\foreach \t in {1,...,3}
			{
				\pgfmathtruncatemacro{\k}{\t - 1}
				\node (sb\t) [SB, right=7 mm of sb\k] {};
			}

			\node (sbm1) [SB, left=7 mm of sb0] {};
			\foreach \t in {2,...,3}
			{
				\pgfmathtruncatemacro{\k}{\t - 1}
				\node (sbm\t) [SB, left=7 mm of sbm\k] {};
			}

			%center
			\draw[thick] (sb3)--(sb2)--(sb1)--(sb0)--(sbm1)--(sbm2)--(sbm3); 
			
			\end{tikzpicture}
		\end{center}
	\caption{\label{Fig:G_ISB} Inter-sub-block graphs for the $ (4,8,2) $ protographs (A), (B), (C), and (D) from Example~\ref{Ex:G_ISB}}
	\end{figure}
\end{example}

Consider the inter-sub-block graph in Figure~\ref{Fig:G_ISB}(D), and assume the target SB is located in the grid's center. There are many SG helper schedules to decode the target, each exhibiting a different threshold and complexity (see Section~\ref{Sub: SG thresholds} for the definition of SG thresholds and complexities).
For example, we can choose to decode helpers along the vertical line crossing the target SB (this will be better than the horizontal line since the vertical line ends in termination), similarly to the one-dimensional chain in Section~\ref{Sec:SG}; we call this the \emph{vertical schedule}. 
Alternatively, one can access helpers on both the vertical and horizontal lines crossing the target; we call this the \emph{cross schedule}. 
Another \er{schedule} is the \emph{diamond schedule}, in which helper SBs sharing the same (Manhattan) distance from the target SB are decoded in parallel, and in order of decreasing distances from the target. 
There are many more possible schedules for two-dimensional protographs like protograph D (in contrast to the simple chain in protograph A, where we have only one direction). We compare the three schedules presented above: vertical, cross, and diamond, illustrated in Figures~\ref{Fig:2D SG strat}(a), (b), and (c), respectively. 

Using the density-evolution equations derived in Section~\ref{Sub:SG DE}, the thresholds of all schedules are calculated and listed in Table~\ref{Tbl:SG strats} for the two-dimensional $ (4,8,2) $ protograph with $ M=49$ sub-blocks (i.e, a $ 7\times 7 $ grid). 
In addition, Table~\ref{Tbl:SG strats} lists the number of \er{helpers} decoded, as a measure of decoding complexity.  

\begin{remark}
The diamond and cross schedules offer more parallelism compared to the vertical schedule, and hence reduce latency. For example, if $ d=4 $, then the vertical schedule requires three decoding steps (two helper steps, and one target step), while the cross and diamond schedules require only two steps.
\end{remark}

\begin{figure}
	\begin{center}
	
	\begin{tikzpicture}\label{Tikz:2D-SG strategies}
	\tikzstyle{SB}=[rounded corners,draw,minimum size=4mm]
		\tikzstyle{zero}=[fill,draw,scale=0.6]
		
		% vertical
		\begin{scope}[shift={(0,0)}]
			\pgfmathsetmacro{\T}{5}
			\pgfmathtruncatemacro{\Q}{\T - 1}
			\pgfmathsetmacro{\x}{1}
			\pgfmathsetmacro{\y}{2}
			\foreach \t in {1,...,\T}
			{
				\foreach \r in {1,...,\T}
				{
					\node[SB] (sb\t\r) at (\x*\t,\r) {};
				}
			}
			\node[SB,fill=blue!50!gray] at (sb35) {};
			\node[SB,fill=blue!50!gray] at (sb31) {};
			\node[SB,fill=green!50!gray] at (sb34) {};
			\node[SB,fill=green!50!gray] at (sb32) {};
			\node[SB,fill=gray] at (sb33) {};

			\foreach \t in {1,...,\Q}
			{
				
				\foreach \r in {1,...,\Q}
				{
					\pgfmathtruncatemacro{\k}{\t + 1}
					\pgfmathtruncatemacro{\j}{\r + 1}
					\draw[thick] (sb\t\r)--(sb\k\r);
					\draw[thick] (sb\t\r)--(sb\t\j);
				}
			}
			\foreach \r in {1,...,\Q}
			{
				\pgfmathtruncatemacro{\j}{\r + 1}
				\draw[thick] (sb5\r)--(sb5\j);
			}
			\foreach \t in {1,...,\Q}
			{
				\pgfmathtruncatemacro{\k}{\t + 1}
				\draw[thick] (sb\t5)--(sb\k5);
			}

			\node (a) at ($(sb52.east)!0.5!(sb51.east)$) {};
			\node (b) at ($(sb12.west)!0.5!(sb11.west)$) {};
			\draw[thick] (sb52.east) to[in= 0, out = 0] (a.center);
			\draw[thick] (a.center) -- (b.center);
			\draw[thick] (b.center) to[in= 180, out = 180] (sb11.west);
			
			\node (a) at ($(sb53.east)!0.5!(sb52.east)$) {};
			\node (b) at ($(sb13.west)!0.5!(sb12.west)$) {};
			\draw[thick] (sb53.east) to[in= 0, out = 0] (a.center);
			\draw[thick] (a.center) -- (b.center);
			\draw[thick] (b.center) to[in= 180, out = 180] (sb12.west);
			
			\node (a) at ($(sb54.east)!0.5!(sb53.east)$) {};
			\node (b) at ($(sb14.west)!0.5!(sb13.west)$) {};
			\draw[thick] (sb54.east) to[in= 0, out = 0] (a.center);
			\draw[thick] (a.center) -- (b.center);
			\draw[thick] (b.center) to[in= 180, out = 180] (sb13.west);
			
			\node (a) at ($(sb55.east)!0.5!(sb54.east)$) {};
			\node (b) at ($(sb15.west)!0.5!(sb14.west)$) {};
			\draw[thick] (sb55.east) to[in= 0, out = 0] (a.center);
			\draw[thick] (a.center) -- (b.center);
			\draw[thick] (b.center) to[in= 180, out = 180] (sb14.west);

			\node (a) [below = 10mm of sb31] {(a)};
	    \end{scope} 
	
		% cross
		\begin{scope}[shift={(1.1*5.5,0)}]
			\pgfmathsetmacro{\T}{5}
			\pgfmathtruncatemacro{\Q}{\T - 1}
			\pgfmathsetmacro{\x}{1}
			\pgfmathsetmacro{\y}{2}
			\foreach \t in {1,...,\T}
			{
				\foreach \r in {1,...,\T}
				{
					\node[SB] (sb\t\r) at (\x*\t,\r) {};
				}
			}
			\node[SB,fill=blue!50!gray] at (sb35) {};
			\node[SB,fill=blue!50!gray] at (sb31) {};
			\node[SB,fill=green!50!gray] at (sb34) {};
			\node[SB,fill=green!50!gray] at (sb32) {};
			\node[SB,fill=gray] at (sb33) {};
			\node[SB,fill=blue!50!gray] at (sb53) {};
			\node[SB,fill=blue!50!gray] at (sb13) {};
			\node[SB,fill=green!50!gray] at (sb43) {};
			\node[SB,fill=green!50!gray] at (sb23) {};

			\foreach \t in {1,...,\Q}
			{
				
				\foreach \r in {1,...,\Q}
				{
					\pgfmathtruncatemacro{\k}{\t + 1}
					\pgfmathtruncatemacro{\j}{\r + 1}
					\draw[thick] (sb\t\r)--(sb\k\r);
					\draw[thick] (sb\t\r)--(sb\t\j);
				}
			}
			\foreach \r in {1,...,\Q}
			{
				\pgfmathtruncatemacro{\j}{\r + 1}
				\draw[thick] (sb5\r)--(sb5\j);
			}
			\foreach \t in {1,...,\Q}
			{
				\pgfmathtruncatemacro{\k}{\t + 1}
				\draw[thick] (sb\t5)--(sb\k5);
			}

			\node (a) at ($(sb52.east)!0.5!(sb51.east)$) {};
			\node (b) at ($(sb12.west)!0.5!(sb11.west)$) {};
			\draw[thick] (sb52.east) to[in= 0, out = 0] (a.center);
			\draw[thick] (a.center) -- (b.center);
			\draw[thick] (b.center) to[in= 180, out = 180] (sb11.west);
			
			\node (a) at ($(sb53.east)!0.5!(sb52.east)$) {};
			\node (b) at ($(sb13.west)!0.5!(sb12.west)$) {};
			\draw[thick] (sb53.east) to[in= 0, out = 0] (a.center);
			\draw[thick] (a.center) -- (b.center);
			\draw[thick] (b.center) to[in= 180, out = 180] (sb12.west);
			
			\node (a) at ($(sb54.east)!0.5!(sb53.east)$) {};
			\node (b) at ($(sb14.west)!0.5!(sb13.west)$) {};
			\draw[thick] (sb54.east) to[in= 0, out = 0] (a.center);
			\draw[thick] (a.center) -- (b.center);
			\draw[thick] (b.center) to[in= 180, out = 180] (sb13.west);
			
			\node (a) at ($(sb55.east)!0.5!(sb54.east)$) {};
			\node (b) at ($(sb15.west)!0.5!(sb14.west)$) {};
			\draw[thick] (sb55.east) to[in= 0, out = 0] (a.center);
			\draw[thick] (a.center) -- (b.center);
			\draw[thick] (b.center) to[in= 180, out = 180] (sb14.west);

			\node (a) [below = 10mm of sb31] {(b)};
		\end{scope} 
		
		% diamond
		\begin{scope}[shift={(2.2*5.5,0)}]
			\pgfmathsetmacro{\T}{5}
			\pgfmathtruncatemacro{\Q}{\T - 1}
			\pgfmathsetmacro{\x}{1}
			\pgfmathsetmacro{\y}{2}
			\foreach \t in {1,...,\T}
			{
				\foreach \r in {1,...,\T}
				{
					\node[SB] (sb\t\r) at (\x*\t,\r) {};
				}
			}
			\node[SB,fill=cyan!50!gray] at (sb55) {};
			\node[SB,fill=cyan!50!gray] at (sb11) {};
			\node[SB,fill=cyan!50!gray] at (sb15) {};
			\node[SB,fill=cyan!50!gray] at (sb51) {};
		
			\node[SB,fill=red!50!gray] at (sb45) {};
			\node[SB,fill=red!50!gray] at (sb25) {};
			\node[SB,fill=red!50!gray] at (sb41) {};
			\node[SB,fill=red!50!gray] at (sb21) {};
			\node[SB,fill=red!50!gray] at (sb54) {};
			\node[SB,fill=red!50!gray] at (sb52) {};
			\node[SB,fill=red!50!gray] at (sb14) {};
			\node[SB,fill=red!50!gray] at (sb12) {};
			
			\node[SB,fill=blue!50!gray] at (sb35) {};
			\node[SB,fill=blue!50!gray] at (sb31) {};
			\node[SB,fill=blue!50!gray] at (sb44) {};
			\node[SB,fill=blue!50!gray] at (sb22) {};
			\node[SB,fill=blue!50!gray] at (sb24) {};
			\node[SB,fill=blue!50!gray] at (sb42) {};
			\node[SB,fill=blue!50!gray] at (sb53) {};
			\node[SB,fill=blue!50!gray] at (sb13) {};
			
			\node[SB,fill=green!50!gray] at (sb34) {};
			\node[SB,fill=green!50!gray] at (sb32) {};
			\node[SB,fill=green!50!gray] at (sb43) {};
			\node[SB,fill=green!50!gray] at (sb23) {};
			
			\node[SB,fill=gray] at (sb33) {};

			\foreach \t in {1,...,\Q}
			{
				
				\foreach \r in {1,...,\Q}
				{
					\pgfmathtruncatemacro{\k}{\t + 1}
					\pgfmathtruncatemacro{\j}{\r + 1}
					\draw[thick] (sb\t\r)--(sb\k\r);
					\draw[thick] (sb\t\r)--(sb\t\j);
				}
			}
			\foreach \r in {1,...,\Q}
			{
				\pgfmathtruncatemacro{\j}{\r + 1}
				\draw[thick] (sb5\r)--(sb5\j);
			}
			\foreach \t in {1,...,\Q}
			{
				\pgfmathtruncatemacro{\k}{\t + 1}
				\draw[thick] (sb\t5)--(sb\k5);
			}

			\node (a) at ($(sb52.east)!0.5!(sb51.east)$) {};
			\node (b) at ($(sb12.west)!0.5!(sb11.west)$) {};
			\draw[thick] (sb52.east) to[in= 0, out = 0] (a.center);
			\draw[thick] (a.center) -- (b.center);
			\draw[thick] (b.center) to[in= 180, out = 180] (sb11.west);
			
			\node (a) at ($(sb53.east)!0.5!(sb52.east)$) {};
			\node (b) at ($(sb13.west)!0.5!(sb12.west)$) {};
			\draw[thick] (sb53.east) to[in= 0, out = 0] (a.center);
			\draw[thick] (a.center) -- (b.center);
			\draw[thick] (b.center) to[in= 180, out = 180] (sb12.west);
			
			\node (a) at ($(sb54.east)!0.5!(sb53.east)$) {};
			\node (b) at ($(sb14.west)!0.5!(sb13.west)$) {};
			\draw[thick] (sb54.east) to[in= 0, out = 0] (a.center);
			\draw[thick] (a.center) -- (b.center);
			\draw[thick] (b.center) to[in= 180, out = 180] (sb13.west);
			
			\node (a) at ($(sb55.east)!0.5!(sb54.east)$) {};
			\node (b) at ($(sb15.west)!0.5!(sb14.west)$) {};
			\draw[thick] (sb55.east) to[in= 0, out = 0] (a.center);
			\draw[thick] (a.center) -- (b.center);
			\draw[thick] (b.center) to[in= 180, out = 180] (sb14.west);

			\node (a) [below = 10mm of sb31] {(c)};
		\end{scope} 
\end{tikzpicture}

	\end{center}
	\caption{\label{Fig:2D SG strat} Three SG schedules over the D protograph from Figure~\ref{Fig:G_ISB}. SBs with same color are decoded in parallel (white SBs are not decoded at all, the gray SB is the target): (a) vertical; (b) cross; (c) diamond.}
\end{figure}
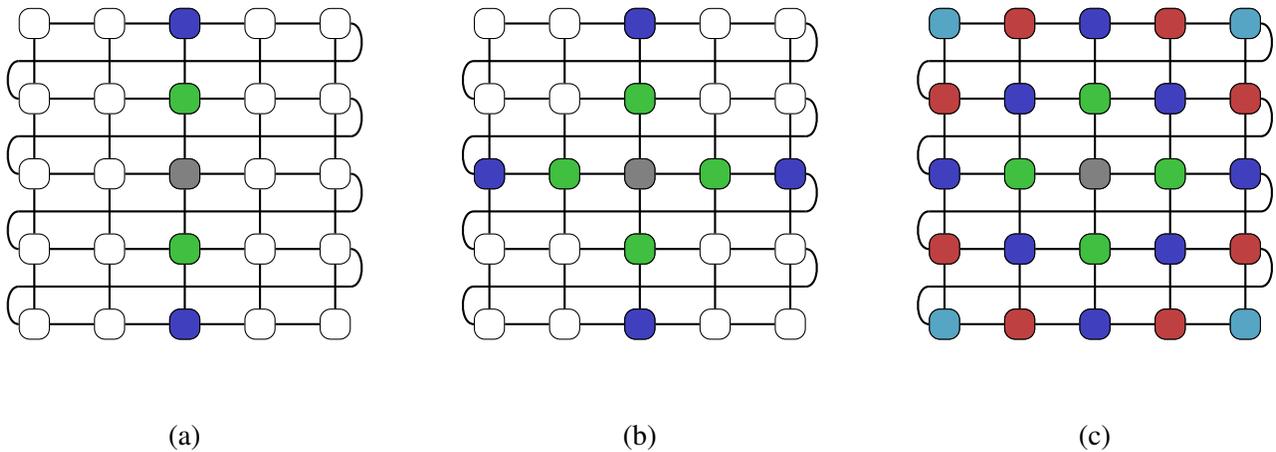

\begin{table}
	\caption{\label{Tbl:SG strats} A comparison between three SG schedules over the $ (l=4,r=8,t=2) $ $ 7\times 7 $ grid protograph graph.}
	\begin{center}
	\begin{tabular}{c|cc|cc|cc}
		
		Steps (latency) &
		\multicolumn{2}{c|}{Vertical} &
		\multicolumn{2}{c|}{Cross} &
		\multicolumn{2}{c}{Diamond} \\
		\hline
				  & Helpers & Threshold 	    & Helpers & Threshold 		  & Helpers		   & Threshold \\
		\cline{2-7}
		1 	 	  &2 		&0.2639 		  	&4 		  &0.3084 	  	 	  &4  			   &0.3084 \\
		2 	 	  &4 		&0.2791 		 	&8 		  &0.3163 	  		  &12 			   &0.3421 \\
		3    	  &6  		&0.3108	 	    	&12		  &0.3175		 	  &24 			   &0.3496 \\
		4    	  &  		& 		 	    	&  		  &		 	  	 	  &36 			   &0.3734 \\
		5    	  &  		& 		 	    	&  		  &		 	  	 	  &44 			   &0.3740 \\
		6    	  &  		& 		 	    	&  		  &		 	  	 	  &48 			   &0.3760
		 
	\end{tabular}
	\end{center}
\end{table}

\section{Acknowledgments}
The authors would like to thank the associated editor and the anonymous reviewers for their constructive comments and suggestions.
They would also like to thank Guy Morag and Adir Kobovich for implementations and simulations for Figures~\ref{Fig:GlobBER48_BEC}--\ref{Fig:LocBER48_AWGN} and Figure~\ref{Fig:SG512_BEC}. This work was supported in part by the US-Israel Binational Science Foundation and in part by the Israel Science Foundation. 

\appendices

\section{Proof of Lemma~\ref{Lemma:Th UB}}
\label{App:Th UB}
We start with a generalization of \eqref{Eq:DE Vars}--\eqref{Eq:DE init} that is useful for the forthcoming analysis. Consider the case where each VN $v \in \mathcal{V}$ is transmitted through an erasure channel with a different erasure probability $\epsilon_v$. In this case\er{,} we say that $\underline{\epsilon}=\left (\epsilon_1,\epsilon_2,\ldots ,\epsilon_{\; \left | \mathcal{V} \right | \;}\right )$ is an erasure constellation, and \eqref{Eq:DE Vars} and \eqref{Eq:P_e,l} are replaced with
\er{
	\begin{align}
	\label{Eq:DE Vars generalized}
	x_\ell \left(e^v_{i},\epsilon_v\right) = \epsilon_v \cdot \prod_{\substack{1\leq i'\leq d_v \\ i'\neq t}} u_{\ell} \left(e^v_{i'}\right), \quad P_{\ell}(v,\epsilon_v)=\epsilon_v \prod_{1\leq i\leq d_v} u_{\ell} \left(e^v_{i}\right).
	\end{align}
}

\begin{lemma} \label{Lemma:mono}
	Let $\mathcal{G}=\left(\mathcal{V}\cup\mathcal{C},\mathcal{E}\right)$ be a protograph, and let $\underline{\epsilon}_1=\left (\epsilon_{1,1},\epsilon_{1,1},\ldots,\epsilon_{1, \left | \mathcal{V} \right | }\right )$ and $\underline{\epsilon}_2=\left (\epsilon_{2,1},\epsilon_{2,1},\ldots,\epsilon_{2, \left | \mathcal{V} \right | }\right )$ be erasure constellations such that $\underline{\epsilon_1} \preceq \underline{\epsilon_2}$, i.e., for every $ v \in \mathcal{V} $, $ \epsilon_{1,v} \leq \epsilon_{2,v}$. Then,
	\begin{align*}
	P_{\ell}(v,\epsilon_{1,v}) \leq P_{\ell}(v,\epsilon_{2,v}) , \quad \forall v\in \mathcal{V}, \quad \forall \ell \geq 0\;.
	\end{align*}
\end{lemma}

\begin{proof}
	The proof follows by induction on $\ell$ and by the monotonicity (in the second argument) of $x_\ell \left(e^v_{i},\epsilon_v\right)$ and $ P_{\ell}(v,\epsilon_v)$ in \eqref{Eq:DE Vars generalized}. Details are left to the reader.
\end{proof}

\begin{corollary}\label{Coro:mono}
	If $\underline{\epsilon_1} \preceq \underline{\epsilon_2}$ and for every $v \in \mathcal{V}$, $\lim\limits_{\ell\to \infty}P_{\ell}(v,\epsilon_{2,v})=0$, then for every $v \in \mathcal{V}$, $\lim\limits_{\ell\to \infty}P_{\ell}(v,\epsilon_{1,v})=0$.
\end{corollary}

Let $\epsilon>\epsilon^*\left(H_\mathcal{J}\right)$. In view of \eqref{Eq:threshold}, there exists a set of VNs $\mathcal{J}' \subseteq \mathcal{J}$ such that 
\begin{align} \label{Eq:unresloved VNs}
\lim_{\ell\to \infty}P_{\ell}(v,\epsilon)>0,\quad  \forall  v \in \mathcal{J}'.
\end{align}
We now continue to the proof of Lemma~\ref{Lemma:Th UB}. Assume w.l.o.g that $\mathcal{J}=\{1,2,\ldots,\; \big| \;\mathcal{J}\; \big| \;\}$ (else, permute the columns in $H$). 
Consider erasure constellations given by 
\begin{subequations}
	\begin{align}
	\label{Eq:err constellations}
	&\underline{\epsilon_1}=(\underbrace{\epsilon,\ldots,\epsilon}_{\; \big| \;\mathcal{J}\; \big| \;},\underbrace{0,\ldots,0}_{\; \big| \;\mathcal{V}\; \big| \;-\; \big| \;\mathcal{J}\; \big| \;}),\\
	&\underline{\epsilon_2}=(\underbrace{\epsilon,\epsilon,\ldots\ldots\ldots,\epsilon}_{\; \big| \;\mathcal{V}\; \big| \;}).
	\end{align}
\end{subequations}
In what follows, we prove that $\epsilon\geq \epsilon^*\left(H\right)$. Assume to the contrary that $\epsilon<\epsilon^*\left(H\right)$. Since $\underline{\epsilon_1} \preceq \underline{\epsilon_2}$, Corollary~\ref{Coro:mono} implies that
\begin{align}\label{Eq:P_v=0}
\lim_{\ell\to \infty}P_{\ell}(v,\epsilon_{1,v})=0,\quad \forall v \in \mathcal{V}.
\end{align}
In addition, applying the DE equations in \eqref{Eq:DE Checks} and \eqref{Eq:DE Vars generalized} on the protograph corresponding to $H$ with the erasure constellation $\underline{\epsilon_1}$ is equivalent to applying the DE equations in \eqref{Eq:DE Vars} and \eqref{Eq:DE Checks} on the protograph corresponding $H_\mathcal{J}$ over the $BEC(\epsilon)$, thus \eqref{Eq:P_v=0} implies that for every $v \in \mathcal{J}$, $\lim_{\ell\to \infty}P_{\ell}(v,\epsilon)=0$ in contradiction to \eqref{Eq:unresloved VNs}. Hence, $\epsilon\geq \epsilon^*\left(H\right)$. Since this is true for all $\epsilon>\epsilon^*\left(H_\mathcal{J}\right)$, we deduce that $\epsilon^*\left(H\right) \leq \epsilon^*\left(H_\mathcal{J}\right)$.

Similarly, let $\epsilon>\epsilon^*\left(H\right)$, and let $v\in \mathcal{V}$ be a VN in the protograph corresponding to $H$ such that
\begin{align}\label{Eq:Pv>0}
\lim_{\ell\to \infty}P_{\ell}(v,\epsilon)>0.
\end{align} 
Since $u_\ell(e^v_t)$ in \eqref{Eq:P_e,l} is less than $1$ for every $t\in\{1,2,\ldots,d_v\}$ and every iteration $\ell$, then  
\begin{align}\label{Eq:P^I}
P^{(\mathcal{I})}_{\ell}(v,\epsilon)\geq P_{\ell}(v,\epsilon),\qquad \forall \ell\geq 0,
\end{align}
where $P^{(\mathcal{I})}_{\ell}$ is the probability that $v$ is erased after $\ell$ BP iterations over the protograph corresponding to $H^{(\mathcal{I})}$. Combining \eqref{Eq:Pv>0} and \eqref{Eq:P^I} implies that $\lim_{\ell\to \infty}P^{(\mathcal{I})}_{\ell}(v,\epsilon)>0$, thus $\epsilon>\epsilon^*\left(H^{(\mathcal{I})}\right)$. Since this is true for every $\epsilon>\epsilon^*\left(H\right)$, then $\epsilon^*\left(H^{(\mathcal{I})}\right)  \leq \epsilon^*\left(H\right)$.

\section{Proof of Theorem~\ref{Th:Recursive SG}}
\label{App:Recursive SG}
The $ j=0 $ case follows from the definition of $ \epsilon^*(\underline\delta_1,\underline\delta_2) $. 
For $ j\geq2 $, assume the target SB is indexed by $ m $. Let $ \mathcal{X}_{j,\underline\delta_1,\underline\delta_2} $ be an indicator random variable that equals $ 1 $ if and only if SG decoding with $ j $ SBs succeeds given that the incoming erasure rates to SB $ m + j $ (resp. $ m - j $) are $\underline\delta_1$ (resp. $\underline\delta_2 $), i.e., 
\begin{align}\label{Eq:Exp indicator2}
p_j\left (\underline\delta_1,\underline\delta_2\right )=\mathbb{E}\left [ \mathcal{X}_{j,\underline\delta_1,\underline\delta_2} \right ].
\end{align}
By the tower rule for expectations,
\begin{align}\label{Eq:Tower2}
\mathbb{E}\left [ \mathcal{X}_{j,\underline\delta_1,\underline\delta_2} \right ]=\mathbb{E}\left [ \mathbb{E}\left [\mathcal{X}_{j,\underline\delta_1,\underline\delta_2}  \; \big| \;E_{-j},E_{+j} \right ]\right ],
\end{align}
where $ E_{-j} $ and $ E_{+j} $ are the random variables corresponding to the erasure probability of SBs $ m - j $ and $m + j $, respectively. In view of the assumption in \eqref{Eq:w assumption}, given $ E_{-j}=\epsilon_1,\;E_{+j}=\epsilon_2 $ we have that the outgoing erasure rates from SB $ m + j $ (resp. $ m - j $) towards SB $ m +j -1 $  (resp. $ m -j +1 $) are $\Delta\left (\epsilon_1,\underline\delta_1\right )$ (resp. $\Delta\left (\epsilon_2,\underline\delta_2\right ) $). Hence
\begin{align}\label{Eq:Conditional Exp2}
\mathbb{E}\left [\mathcal{X}_{j,\underline\delta_1,\underline\delta_2} \; \big| \;E_{-j},E_{+j} \right ]=p_{j-2}\left (\Delta\left (E_{-j},\underline\delta_1\right ),\Delta\left (E_{+j},\underline\delta_2\right ) \right ).
\end{align}
Combining \eqref{Eq:Exp indicator2}--\eqref{Eq:Conditional Exp2} completes the proof.

\section{Proof of Theorem~\ref{Th:SG LB}}
\label{App:SG LB}
We prove by induction on \er{(even)} $ j $. For $ j=0 $ we get from \eqref{Eq:Recursive SG},
\begin{align}\label{Eq: p0=}
p_0(\underline\delta_1,\underline\delta_2)=F\left (\epsilon^*(\underline\delta_1,\underline\delta_2)\right ).
\end{align}
In addition, \eqref{Eq:Delta1} implies that $ y_{\underline i}(\underline \delta_1,\underline \delta_2)=\epsilon^*(\underline \delta_1,\underline \delta_2)
$, which combined with \eqref{Eq: p0=} yields $ p_0(\underline\delta_1,\underline\delta_2)=F\left (\epsilon^*(y_{\underline i}(\underline \delta_1,\underline \delta_2)\right ) $. This proves the $ j=0 $ case. Consider $ j>0 $, let $ j'=j-2 $, and assume that \eqref{Eq:SG LB} holds for $ j' $. In view of \eqref{Eq:Recursive SG} and the induction assumption we have
\begin{align}\label{Eq:pd induction 1}
\begin{split}
p_{j}(\underline \delta_1,\underline \delta_2)
=&\mathbb{E}\left [p_{j'}(\Delta\left (E_{-j/2},\underline \delta_1\right ),\Delta(E_{+j/2},\underline \delta_2) \right ]\\
\geq&\mathbb{E}\sum_{\underline{i'}\in \{1,2,\ldots,K\}^{j'}}   F\left (y_{\underline{i'}}(\Delta\left (E_{-j/2},\underline \delta_1\right ),\Delta(E_{+j/2},\underline \delta_2)\right )  \prod_{k=-j'/2\atop k\neq 0}^{j'/2}   \left [ F(e_{i_k}) - F(e_{i_k-1})\right ]\\
=&\sum_{\underline{i'}\in \{1,2,\ldots,K\}^{j'}}  \mathbb{E}\left [ F\left (y_{\underline{i'}}(\Delta\left (E_{-j/2},\underline \delta_1\right ),\Delta(E_{+j/2},\underline \delta_2)\right ) \right ] \prod_{k=-j'/2\atop k\neq 0}^{j'/2}   \left [ F(e_{i_k}) - F(e_{i_k-1})\right ]\;.
\end{split}
\end{align}
Let $\mathcal{X}_A$ be an indicator random variable that equals 1 if and only if the event $ A $ occur.
Since for every $ \delta $, $ \Delta(\epsilon,\delta)$ is monotonically non-decreasing in $ \epsilon $, then $F\left (y_{\underline{i'}}(\Delta\left (E_{-j/2},\underline \delta_1\right ),\Delta(E_{+j/2},\underline \delta_2)\right )$ is  monotonically non-increasing in $ E_{-j/2} $ and $ E_{+j/2} $. 
Thus, for every $ \underline{i'}\in \{1,2,\ldots,K\}^{j'} $, 
\begin{align}\label{Eq:quant}
\begin{split}
\mathbb{E}\left [F\left (y_{\underline{i'}}(\Delta\left (E_{-j/2},\underline \delta_1\right ),\Delta(E_{+j/2},\underline \delta_2)\right )\right ]
&=\mathbb{E} \sum_{i_{-j/2}=1}^K \sum_{i_{j/2}=1}^K \mathcal{X}_{E_{-j/2}\in (e_{i_{-j/2}-1},e_{i_{-j/2}}]}\mathcal{X}_{E_{+j/2}\in (e_{i_{j/2}-1},e_{i_{j/2}}]}\\
&\qquad\qquad\qquad \cdot F\left (y_{\underline{i'}}(\Delta\left (E_{-j/2},\underline \delta_1\right ),\Delta(E_{+j/2},\underline \delta_2)\right )\\
&\geq\sum_{i_{-j/2}=1}^K \sum_{i_{j/2}=1}^K \mathbb{E} \left [ \mathcal{X}_{E_{-j/2}\in (e_{i_{-j/2}-1},e_{i_{-j/2}}]}\mathcal{X}_{E_{+j/2}\in (e_{i_{j/2}-1},e_{i_{j/2}}]} \right ]\\
&\qquad\qquad\qquad \cdot F\left (y_{\underline{i'}}(\Delta\left (e_{i_{-j/2}},\underline \delta_1\right ),\Delta(e_{i_{j/2}},\underline \delta_2)\right )\\
&=\sum_{i_{-j/2}=1}^K \sum_{i_{j/2}=1}^K \left[(F\left (e_{i_{-j/2}}\right )-F\left (e_{i_{-j/2}-1}\right )\right ]\cdot \left[(F\left (e_{i_{j/2}}\right )-F\left (e_{i_{j/2}-1}\right )\right ]\\
&\qquad\qquad\qquad \cdot F\left (y_{\underline{i'}}(\Delta\left (e_{i_{-j/2}},\underline \delta_1\right ),\Delta(e_{i_{j/2}},\underline \delta_2)\right )\,.
\end{split}
\end{align}
In view of \eqref{Eq:Deltak} and \eqref{Eq:yi}, 
\begin{align} \label{Eq:expand Delta}
y_{\underline{i'}}\left (\Delta\left (e_{i_{-j/2}},\underline \delta_1\right ),\Delta(e_{i_{j/2}},\underline \delta_2)\right )
=y_{\underline{i}}\left ( \delta_1,\delta_2)\right ).
\end{align}
Combining \eqref{Eq:pd induction 1}--\eqref{Eq:expand Delta} completes the proof.

\section{Proof of Proposition~\ref{Prop:SG_bounds}}
\label{App:SG bounds}
Let $ j'=\tfrac{j}{2} $, and let $ E_{+j'} $ and $ E_{-j'} $ be the erasure-probability random variables of sub-blocks $m+j' $ and $ m-j' $, respectively. In view of \eqref{Eq:Recursive SG}, since $ \Delta(\cdot,\cdot)\leq 1 $,
\begin{align*}
p_j(\underline 1,\underline 1)
&=\mathbb{E}\left[ p_{j-2}
\left(\Delta(E_{-j'},\underline 1) \right),\left(\Delta(E_{+j'},\underline 1) \right)\right]\\
&\geq  \Pr(E_{-j'}<\epsilon_L,E_{+j'}<\epsilon_L)p_{j-2}\left(\underline 0,\underline 0\right) \\
&+\Pr(E_{-j'}\geq\epsilon_L,E_{+j'}<\epsilon_L)p_{j-2}\left(\underline 1,\underline 0\right)\\
&+\Pr(E_{-j'}<\epsilon_L,E_{+j'}\geq\epsilon_L)p_{j-2}\left(\underline 0,\underline 1\right)\\
&+\Pr(E_{-j'}\geq\epsilon_L,E_{+j'}\geq\epsilon_L)p_{j-2}\left(\underline 1,\underline 1\right)\\
&=P_L^2\cdot p_{j-2}\left(\underline 0,\underline 0\right)+2P_L(1-P_L)p_{j-2}\left(\underline 1,\underline 0\right) +\left (1-P_L\right )^2 p_{j-2}\left(\underline 1,\underline 1\right),
\end{align*}
where the last equality is due to the symmetry assumption in \eqref{Eq:w assumption}.

Next, let $ \mathcal{D}_{-j',+j'} $ be the event of successful SG decoding of a target sub-block $ m $ with $ d=2j' $ helper sub-blocks. For every $ k\in\{0,1,\ldots,j'\}$ (resp. $ k\in\{-j',\ldots,-1,0\}$), let $ \underline \delta_{R}^{(k)} $ (resp. $ \underline \delta_{L}^{(k)} $) be the input erasure rate to sub-block $ m+k $ from the right (resp. left) during SG decoding. Then, according to our definitions,
\begin{align}\label{Eq:cond on d0 SG}
\begin{split}
p_j(\underline \delta_1,\underline \delta_2) 
\triangleq 
&\Pr\left( \mathcal{D}_{-j',+j'}\; \big|\;
\underline\delta_{L}^{(-j')} = \underline\delta_1,\,\underline\delta_{R}^{(+j')} = \underline\delta_2\right ) \\ %ROW
=&\Pr\left( \mathcal{D}_{-j',+j'}\; \big|\;
\underline\delta_{L}^{(-j')} = \underline\delta_1,\,\underline\delta_{R}^{(+j')} = \underline\delta_2,\;
\underline \delta_{L}^{(0)} = 0,\,\underline \delta_{R}^{(0)} = 0 \right )\\
&\hspace*{1.3cm}
\cdot\Pr\left (
\underline \delta_{L}^{(0)} = 0,\,\underline \delta_{R}^{(0)} = 0 
\;\big|\;
\underline\delta_{L}^{(-j')} = \underline\delta_1,\,\underline\delta_{R}^{(+j')} = \underline\delta_2\right )\\%ROW
+&\Pr\left( \mathcal{D}_{-j',+j'}\; \big|\;
\underline\delta_{L}^{(-j')} = \underline\delta_1,\,\underline\delta_{R}^{(+j')} = \underline\delta_2,\;
\underline \delta_{L}^{(0)} \neq 0,\,\underline \delta_{R}^{(0)} = 0 \right )\\
&\hspace*{1.3cm}
\cdot\Pr\left (
\underline \delta_{L}^{(0)} \neq 0,\,\underline \delta_{R}^{(0)} = 0 
\;\big|\;
\underline\delta_{L}^{(-j')} = \underline\delta_1,\,\underline\delta_{R}^{(+j')} = \underline\delta_2\right )\\%ROW
&\Pr\left( \mathcal{D}_{-j',+j'}\; \big|\;
\underline\delta_{L}^{(-j')} = \underline\delta_1,\,\underline\delta_{R}^{(+j')} = \underline\delta_2,\;
\underline \delta_{L}^{(0)} = 0,\,\underline \delta_{R}^{(0)} \neq 0 \right )\\
&\hspace*{1.3cm}
\cdot\Pr\left (
\underline \delta_{L}^{(0)} = 0,\,\underline \delta_{R}^{(0)} \neq 0 
\;\big|\;
\underline\delta_{L}^{(-j')} = \underline\delta_1,\,\underline\delta_{R}^{(+j')} = \underline\delta_2\right )\\%ROW
&\Pr\left( \mathcal{D}_{-j',+j'}\; \big|\;
\underline\delta_{L}^{(-j')} = \underline\delta_1,\,\underline\delta_{R}^{(+j')} = \underline\delta_2,\;
\underline \delta_{L}^{(0)} \neq 0,\,\underline \delta_{R}^{(0)} \neq 0 \right )\\
&\hspace*{1.3cm}
\cdot\Pr\left (
\underline \delta_{L}^{(0)} \neq 0,\,\underline \delta_{R}^{(0)} \neq 0 
\;\big|\;
\underline\delta_{L}^{(-j')} = \underline\delta_1,\,\underline\delta_{R}^{(+j')} = \underline\delta_2\right ).
\end{split}
\end{align}
In view of \eqref{Eq:extreme th}, we have 
\begin{subequations}
\begin{align}\label{Eq:= PSG}
	\begin{split}
	&\Pr\left( \mathcal{D}_{-j',+j'}\; \big|\;
	\underline\delta_{L}^{(-j')} = \underline\delta_1,\,\underline\delta_{R}^{(+j')} = \underline\delta_2,\;
	\underline \delta_{L}^{(0)} = 0,\,\underline \delta_{R}^{(0)} = 0 \right )
	= P_{D}, \\
	&\Pr\left( \mathcal{D}_{-j',+j'}\; \big|\;
	\underline\delta_{L}^{(-j')} = \underline\delta_1,\,\underline\delta_{R}^{(+j')} = \underline\delta_2,\;
	\underline \delta_{L}^{(0)} \neq 0,\,\underline \delta_{R}^{(0)} = 0 \right)
	\geq P_{S},\\
	&\Pr\left( \mathcal{D}_{-j',+j'}\; \big|\;
	\underline\delta_{L}^{(-j')} = \underline\delta_1,\,\underline\delta_{R}^{(+j')} = \underline\delta_2,\;
	\underline \delta_{L}^{(0)} = 0,\,\underline \delta_{R}^{(0)} \neq 0 \right)
	\geq P_{S},\\
	&\Pr\left( \mathcal{D}_{-j',+j'}\; \big|\;
	\underline\delta_{L}^{(-j')} = \underline\delta_1,\,\underline\delta_{R}^{(+j')} = \underline\delta_2,\;
	\underline \delta_{L}^{(0)} \neq 0,\,\underline \delta_{R}^{(0)} \neq 0 \right)
	\geq P_{L}, 
	\end{split}
\end{align}
and since we assumed symmetry of sub-blocks in \eqref{Eq:w assumption}, then
\begin{align}\label{Eq:Pd-1 SG}
\begin{split}
&\Pr\left (
\underline \delta_{L}^{(0)} = 0,\,\underline \delta_{R}^{(0)} = 0 
\;\big|\;
\underline\delta_{L}^{(-j')} = \underline\delta_1,\,\underline\delta_{R}^{(+j')} = \underline\delta_2\right )=\hat p_{j-1}(\underline\delta_1)p_{j'-1}(\underline\delta_2) \\
&\Pr\left (
\underline \delta_{L}^{(0)} \neq 0,\,\underline \delta_{R}^{(0)} = 0 
\;\big|\;
\underline\delta_{L}^{(-j')} = \underline\delta_1,\,\underline\delta_{R}^{(+j')} = \underline\delta_2\right )=
\left (1-\hat p_{j'-1}(\underline\delta_1)\right )
\hat p_{j'-1}(\underline\delta_2)\\
&\Pr\left (
\underline \delta_{L}^{(0)} = 0,\,\underline \delta_{R}^{(0)} \neq 0 
\;\big|\;
\underline\delta_{L}^{(-j')} = \underline\delta_1,\,\underline\delta_{R}^{(+j')} = \underline\delta_2\right )=
\hat p_{j'-1}(\underline\delta_1)
\left (1-\hat p_{j'-1}(\underline\delta_2)\right )\\
&\Pr\left (
\underline \delta_{L}^{(0)} \neq 0,\,\underline \delta_{R}^{(0)} \neq 0 
\;\big|\;
\underline\delta_{L}^{(-j')} = \underline\delta_1,\,\underline\delta_{R}^{(+j')} = \underline\delta_2\right )=
\left (1-\hat p_{j'-1}(\underline\delta_1)\right )
\left (1-\hat p_{j'-1}(\underline\delta_2)\right ).
\end{split}
\end{align}
\end{subequations}
Combining equations \eqref{Eq:cond on d0 SG} and \eqref{Eq:= PSG}--\eqref{Eq:Pd-1 SG} yields \eqref{Eq:SG_rec2}.

\end{document}